\documentclass{birkjour}
\usepackage{amscd,amssymb,amsthm,amsmath,graphicx,psfrag}
 \newtheorem{thm}{Theorem}[section]
 
 \newtheorem{lem}[thm]{Lemma}
 \newtheorem{prop}[thm]{Proposition}
 \theoremstyle{definition}
 \newtheorem{defn}[thm]{Definition}
 \theoremstyle{remark}
 \newtheorem{rem}[thm]{Remark}
 
 \numberwithin{equation}{section}

\begin{document}

%
%
%
%
%
%
%
%
%

\title[The Wasserstein geometry of non-linear $\sigma$ models]
 {The Wasserstein geometry of non-linear $\sigma$ models and the Hamilton--Perelman Ricci flow}

\author[Mauro Carfora]{Mauro Carfora}

\address{%
Dipartimento di Fisica, Universita` degli Studi di Pavia\\
and\\ 
Istituto Nazionale di Fisica Nucleare, Sezione di Pavia\\
via A. Bassi 6, I-27100 Pavia, Italy}

\email{mauro.carfora@pv.infn.it}

\thanks{This work has been partially supported by the PRIN Grant  2010JJ4KPA 006  \emph{Geometrical and analytical theories of finite and infinite  dimensional Hamiltonian systems}}
\subjclass{Primary 53C44; Secondary 58D25, 58D30, 81T17}

\keywords{Ricci flow, Wasserstein geometry, metric-measure spaces}


\begin{abstract}
Non linear sigma models are quantum field theories describing, in the large deviations sense, random fluctuations of harmonic maps between a Riemann surface and a Riemannian manifold. Via their formal renormalization group analysis, they provide a framework for possible generalizations of the Hamilton--Perelman Ricci flow. 
By exploiting the heat kernel embedding introduced by N. Gigli and C. Mantegazza, we show that the Wasserstein geometry of the space of probability measures over Riemannian metric measure spaces provides a natural setting for discussing the relation between non--linear sigma models and Ricci flow theory. This approach provides a rigorous model for the embedding of  Ricci flow into the renormalization group flow for non linear sigma models, and characterizes a non--trivial generalization of the Hamilton--Perelman version of the Ricci flow. We discuss in detail the monotonicity and gradient flow properties of this extended flow. 
\end{abstract}

\maketitle
\tableofcontents
\vfill\eject
\section{Introduction}

In 1982, inspired by the theory of harmonic maps, R. Hamilton published  his landmark paper introducing Ricci flow \cite{11}. Over the years, Hamilton's work has been the point of departure and the motivating example for important developments in geometric analysis, most spectacularly in G. Perelman's proof \cite{18, 19, 20} of the  Thurston geometrization program for three-manifolds \cite{thurston1,thurston2}. In the late 70's and early 80's, the Ricci flow independently appeared on the scene also in theoretical physics, in the framework of 2--dimensional Non-Linear sigma Model (\emph{NL$\sigma$M}) theory. From the formal point of view of $\infty $--dimensional geometric analysis, NL$\sigma$M's are quantum field theories describing, in the large deviations sense, random fluctuations of harmonic maps between a Riemannian surface $(\Sigma, \gamma)$ and a $n$--dimensional  Riemannian manifold $(M, g)$. Prescient remarks of their connection with Ricci flow theory date back to Polyakov's $O(n)$ sigma model  (1975) \cite{Poly} and to J. Honerkamp's Chiral multiloops paper (1972) \cite{Honer}; they were made manifest in Daniel Friedan's characterization \cite{DanThesis} of the Ricci flow  as the weak coupling limit of the renormalization group flow for NL$\sigma$M. It must be noted that Friedan's approach \cite{DanThesis,DanPRL, DanAnnPhys} was very geometric and exploited a  sophisticated control over the interplay between the Ricci flow and the diffeomorphism group, actually introducing what later on came out to be known as the DeTurck version \cite{DeTurck} of the Ricci flow. As recalled above, the rationale of the connection between NL$\sigma$M and Ricci flow  may be traced back to the common roots in the geometry of harmonic maps. Attempts at a deeper understanding of this common structure are, however, very difficult to formalize since  renormalization group techniques for non--linear quantum fields remain a challenge for mathematicians, (and physicists as well). We are dealing with a conceptual framework which connects the  dynamical behavior of a physical theory to the analysis of flows in the space of its coupling parameters, and which must be adapted case by case, often with very different and sophisticated mathematical techniques. Speaking such different a language, one deeply rooted in geometrical analysis, the other more appropriate to the mathematical subtleties of quantum field theory, Ricci flow and renormalization group theory for non--linear $\sigma$ model evolved independently  until  
G. Perelman acknowledged the inspiring role played by the NL$\sigma$M effective action in his groundbreaking paper \cite{18}.  This  has drawn renewed attention  to the fact that the Ricci flow is the 1--loop approximation of  the renormalization group (\emph{RG}) flow for non--linear $\sigma$ models, thus providing a framework for the Hamilton--Perelman theory which is open to generalizations \cite{Bakas, Bakas2, Bakas3}, \cite{mauro1}, \cite{cremaschi}, \cite{Douglas}, \cite{Gegenberg}, \cite{GGI}, \cite{Guenther}, \cite{Oliynyk}, \cite{Tseytlin}.  In spite of  this renewed interest, real progress in obtaining viable extensions of the Ricci flow along this line of approach has been strongly hampered by the  fact that the renormalization group flow for NL$\sigma$M is constructed perturbatively, under rather delicate geometrical and physical assumptions, often difficult to formalize. In particular,   it is not clear how we should interpret the perturbative \emph{embedding} of the Ricci flow in the renormalization group of non--linear $\sigma$ models. As a matter of fact,  the perturbative approach provides a hierarchy of truncated geometric flows, generated by powers of the Riemann tensor, which are very difficult to handle and to analyze when the flows develop singularities, (see however the  recent papers \cite{cremaschi}, \cite{Guenther}). Moreover, the validity of this hierarchical extension of the Ricci flow is biased by the fact that it is 
not obvious in what sense a perturbative expansion approximates a full--fledged renormalization flow, (this is the \emph{linearization stability problem} for the RG flow). The purpose of this paper is to address  some of these problems in a geometrical analysis setting. In particular, we show  that Wasserstein geometry and optimal transport offer an ideal framework for addressing  the relation between Ricci flow and renormalization group from a novel point of view.  We discuss  the geometry of dilatonic non--linear $\sigma$ models in terms of maps between Riemannian surfaces $(\Sigma, \gamma)$ and Riemannian measure spaces 
$(M,g,\,d\omega)$, where the measure $d\omega$ describes the dilaton field.  This formulation naturally induces a  scale--dependent renormalization map of the theory provided by the heat kernel embedding of $(M,g,\,d\omega)$ into
the Wasserstein space of probability measures $(\rm{P}rob(M,g),\,d_g^w)$ over $(M,g)$, where  $d_g^w$ denotes the (quadratic) Wasserstein distance in $\rm{P}rob(M,g)$. This construction provides a mathematical well--defined (toy) model for the RG flow for NL$\sigma$M theory.  As we show explicitly, the geometrical deformation of $(M,g,\,d\omega)$ induced by the heat kernel embedding  mimics a renormalization group flow for the corresponding dilatonic non--linear $\sigma$ model.  Wasserstein geometry and optimal transport theory have  recently drawn
attention in attempts of extending the notion of Ricci curvature and Ricci flow
to general metric spaces \cite{lott1}, \cite{lott2}, \cite{sturm3}. The analysis presented in this paper shows  that Wasserstein metric and optimal transport also play a significant role in the rationale relating Ricci flow and renormalization group theory. In particular we prove that a natural extension of the (Hamilton--Perelman version of the) Ricci flow is defined by the RG scaling induced by the heat kernel embedding. This result extends in a non--trivial way a remarkable paper by Nicola Gigli and 
Carlo Mantegazza \cite{Carlo} who have been able to show that the heat kernel embedding in the Wasserstein space of probability measure has a tangent at the origin which is the Ricci flow.  In our case we deal with a weighted heat kernel (the weight being associated with the dilatonic measure $d\omega$), but we found that the results in \cite{Carlo} can be naturally extended to this more general case.
Even if the (weighted) heat kernel renormalization  mimics many properties of the renormalization group flow for  
NL$\sigma$M, it should be clearly stressed that it is not the physical RG flow. The way this latter is perturbatively characterized suggests indeed that the actual RG flow is a singular perturbation of the flow defined by the  heat kernel embedding. The situation is akin to that we encounter in quantum field theory where the  quantum fluctuations of an interacting field can be perturbatively considered as a (singular) perturbation of the underlying formal Gaussian measure  associated to the free field. To support this picture we discuss in detail the problem of characterizing a generalized Wiener measure (\emph{\'a} la Gross) associated with our weighted heat kernel embedding, drawing a parallel with the formal \emph{background field quantization} of NL$\sigma$M. Regardless of the status of the (weighted) heat kernel embedding as a toy (but full--fledged) renormalization group flow for NL$\sigma$M, a result of this paper is that the heat kernel embedding of the metric measure space $(M,g,\,d\omega)$ in the Wasserstein space $(\rm{P}rob(M,g),\,d_g^w)$ emerges as a natural candidate for a geometric flow generalizing the (Hamilton--Perelman) Ricci flow. In particular we prove monotonicity and gradient flow properties of the heat kernel induced Ricci flow, and show that these are genuine extensions of the analogous properties in Ricci flow theory. The whole construction is based on a detailed analysis of the weighted heat kernel embedding and of its geometry in $(\rm{P}rob(M,g),\,d_g^w)$, and it suggests potential extensions of the theory to general metric measure spaces. Again, this is in line with the results of \cite{Carlo} which explores in detail the potentialities of the standard heat kernel embedding as the avatar of the Ricci flow in the non smooth setting. To what extent this is a viable alternative remains, however, a formidable open problem in geometrical analysis. \\
\\ 
\noindent
As for what concerns the structure of the paper, the table of contents is self explanatory. Note that when appropriate we shall often use the acronyms  \emph{NL$\sigma$M} and \emph{RG} for \emph{Non--Linear $\sigma$ Model} and \emph{Renormalization Group}, respectively. To set basic notation, we let $M$ denote a $C^{\infty }$ compact $n$--dimensional manifold, ($n\geq 2$), without boundary, and let  $\mathcal{D}iff(M)$ and $\mathcal{M}et(M)$ respectively be the group of smooth diffeomorphisms and the open convex cone of all smooth Riemannian metrics over $M$. For any $g\in \mathcal{M}et(M)$,  we denote by $\nabla_{(g)}$, (or  $\nabla$ if there is no danger of confusion), the Levi--Civita connection of $g$, and let $\mathcal{R}m(g)=\mathcal{R}^{i}_{klm}\,\partial _i\otimes dx^k\otimes dx^l\otimes dx^m$, $\mathcal{R}ic(g)=\mathcal{R}_{ab}\,dx^a\otimes dx^b$ and $\mathcal{R}(g)$ be the corresponding Riemann,  Ricci and  scalar curvature operators, respectively.

\section{Preliminaries: the geometrical setting}
\label{prelim}

Non-linear $\sigma $-models are strictly related to harmonic map theory and to the geometry of the  space  of maps
\begin{equation}
\phi:\,(\Sigma,\gamma)\,\longrightarrow \,(M,g)\;
\end{equation}
between a $2$--dimensional smooth orientable  surface without boundary $(\Sigma, \gamma)$, with  Riemannian metric $\gamma$, and a Riemannian manifold $(M,g)$. Any such map can be thought of as
endowed with the minimal regularity allowing for the characterization of a (generalized) harmonic energy functional. This can make difficult to work in local charts on the target manifold $M$, even at a physical level of rigor. A way out is to use the Nash embedding theorem \cite{GuntNash}, \cite{Schwartz},  according to which any compact Riemannian manifold $(M,g)$ can be isometrically embedded into some Euclidean space $\mathbb{E}^m\,:=\,(\mathbb{R}^m, \delta )$ for $m$ sufficiently large, (\emph{e.g.}, $m\geq \frac{1}{2}n(n+1)+n$ for free local isometric embeddings \cite{GuntNash}). If $J: (M,g)\hookrightarrow \mathbb{E}^m$, is any such an embedding  we can define the Sobolev space of maps
\begin{equation}
\mathcal{H}^{1}_{(J)}(\Sigma,M)\, :=\,\{\phi\in \mathcal{H}^{1}(\Sigma,\,\mathbb{R}^m)\left.\right|\,\phi(\Sigma)\subset J(M) \}\;,
\end{equation} 
where $\mathcal{H}^{1}(\Sigma,\mathbb{R}^m)$ is the Hilbert space of square summable  $\varphi:\Sigma \rightarrow \mathbb{R}^m$, with (first) distributional derivatives $\in L^2(\Sigma,\mathbb{R}^m)$, endowed with the norm
\begin{equation}
 \parallel \phi \parallel_{\mathcal{H}^{1}}\,:=\,\int_{\Sigma}\,\left(\phi^a(x)\,\phi^b(x)\,\delta _{ab}\,+\,
\gamma^{\mu\nu}(x)\,\frac{\partial \phi^{a}(x)}{\partial x^{\mu}}
\frac{\partial\phi^{b}(x)}{\partial x^{\nu}}\,\delta _{ab}  \right)\,d\mu_{\gamma}\;,
\end{equation}
where, for $\phi(x)\in\,J(M)\subset  \mathbb{R}^m$,\, $a,b=1,\ldots,m$ label coordinates in $(\mathbb{R}^m,\,\delta)$, and $d\mu_{\gamma}$ denotes the Riemannian measure on $(\Sigma,\gamma)$. This characterization is independent of $J$ as long as $M$ is compact, since in that case for any two isometric embeddings $J_1$ and $J_2$, the corresponding spaces of maps  $\mathcal{H}^{1}_{(J_1)}(\Sigma,M)$ and $\mathcal{H}^{1}_{(J_2)}(\Sigma,M)$ are homeomorphic \cite{helein} and one can simply write $\mathcal{H}^{1}(\Sigma,M)$.
The space of smooth maps  $C^{\infty }(\Sigma,M)$ is dense \cite{schoen} in the Sobolev space $\mathcal{H}^{1}(\Sigma, M)$, however maps of class $\mathcal{H}^{1}(\Sigma,M)$ are not necessarily  continuous, and to carry out explicit computations, we further require that any such $\phi \in \mathcal{H}^{1}(\Sigma,M)$ is localizable, (cf.  \cite{jost}, Sect. 8.4), and of bounded geometry. Explicitly, we assume that  for every $x_0 \in \Sigma$ there exists a metric disks
$D(x_0,\,\delta  ):=\{x\in \Sigma|\,d_{\gamma}(x_0,x)\,\leq \,\delta\} \subset \Sigma$, of radius $\delta>0$  and a metric ball  $B(r,\,p):=\{z\in M|\,d_g(p,z)\leq  r\}\subset (M,g)$  centered at $p\,\in\, M$, of radius $r>0$ such that $\phi(D(x_0,\,\delta  ))\subset B(r,\,p)$, with 
\begin{equation} 
r\,<\, r_0\,:=\,\min\left\{\frac{1}{3}\,inj\,(M),\,\frac{\pi }{6\,\sqrt{\kappa }}  \right\}\;,
\label{errezero1}
\end{equation}
where $inj\,(M)$ and $\kappa $ respectively denote the injectivity radius of $(M,g)$, and the upper bound to the sectional curvature of $(M,g)$, (we are adopting the standard convention of defining $\pi/2\,\sqrt{\kappa }\doteq\infty$ when $\kappa \leq0$). Under such assumptions, one can use local coordinates also for maps in $\mathcal{H}^{1}(\Sigma, M)$. In particular for any $\phi\in \mathcal{H}^{1}(D(x_0,\,\delta  ), M)$ we can introduce local coordinates
$x^{\alpha }$,\, for points in $(D(x_0,\delta),\Sigma)$, and
 $y^k=\phi^{k}(x)$, \, $k=1,\ldots,n$,  for the corresponding image points in $\phi(D(x_0,\delta))\subset M$, and, by using a partition of unity,  work locally in the 
smooth framework provided by  the  space of  smooth maps 
\begin{equation}
Map\;(\Sigma,M):=\left\{\phi: \Sigma\rightarrow  M,\, x^{\alpha}\longmapsto y^k=\phi^{k}(x) \in C^{\infty }(D(x_0,\delta),M)\right\}\;.
\label{mapdefA}
\end{equation}
\noindent
Under these regularity hypotheses, we can introduce the  pull--back  bundle $\phi^{-1}TM$ whose sections $v\equiv \phi^{-1}V:= V\circ \phi$,\,  $V\in C^{\infty }(M, TM)$,\, are the vector fields over $\Sigma$ covering the map $\phi$. If $T^*\Sigma$ denotes the cotangent bundle to $(\Sigma, \gamma)$, then the differential $d\phi=\frac{\partial\phi^i}{\partial x^{\alpha}}dx^{\alpha}\otimes \frac{\partial }{\partial \phi^i}$ can  be interpreted as a section of  $T^*\Sigma\otimes \phi^{-1}TM$, and its
Hilbert--Schmidt norm, in the bundle metric
\begin{equation}
\label{bundmetr}
\langle\cdot ,\cdot \rangle_{T^*\Sigma\otimes\phi^{-1}TM}\,:=\,
\gamma^{-1}(x)\otimes g(\phi(x))(\cdot ,\cdot )\;,
\end{equation}
is provided by, (see \emph{e.g.} \cite{jost}), 
\begin{equation}
\langle d\phi , d\phi \rangle_{T^*\Sigma\otimes\phi^{-1}TM}\,=\,
\gamma^{\mu\nu}(x)\,\frac{\partial \phi^{i}(x)}{\partial x^{\mu}}
\frac{\partial\phi^{j}(x)}{\partial x^{\nu}}\,g_{ij}(\phi(x))=\,tr_{\gamma(x)}\,(\phi^{*}\,g)\;.
\end{equation} 
 The corresponding density 
\begin{equation}
e(\phi)\,\,d\mu_{\gamma}\,:=\,\frac{1}{2}\,\langle d\phi , d\phi \rangle_{T^*\Sigma\otimes\phi^{-1}TM}\, d\mu_{\gamma}\,
=\,\frac{1}{2}\,tr_{\gamma(x)}\,(\phi^{*}\,g)\,d\mu_{\gamma}\;,
\end{equation}
where $d\mu_{\gamma}$ is the volume element of the Riemannian surface $(\Sigma,\gamma)$,
is conformally invariant under  two-dimensional conformal transformations 
\begin{equation}
(\Sigma ,\gamma_{\mu\nu})\mapsto (\Sigma ,e^{-\psi}\,\gamma_{\mu\nu})\;,\;\;\;\;\;\; \;
\psi\in C^{\infty }(\Sigma ,\mathbb{R})\;,
\end{equation} 
and defines the harmonic map energy density associated with  $\phi\in Map(\Sigma,M)$.
In particular, the critical points of the functional 
\begin{equation}
E[\phi,\,g]_{(\Sigma,M)}\,:=\,\int_\Sigma\,e(\phi)\,d\mu_\gamma\;,
\label{Msigmamod0}
\end{equation}
are  the harmonic maps of the Riemann surface $(\Sigma, [\gamma]) $ into $(M,g)$, where $[\gamma]$ denotes the conformal class  of the metric $\gamma$.
\begin{rem}
More generally, for $\phi\in \mathcal{H}^{1}(\Sigma, M)$, the critical points of  (\ref{Msigmamod0}) define weakly harmonic maps $(\Sigma, \gamma) \rightarrow (M, g)$. However, for surfaces, weakly harmonic maps are harmonic \cite{helein0}. 
\end{rem}
\noindent In terms of  the possible geometrical characterization of $(\Sigma, \gamma)$ and $(M,g)$, important examples of harmonic maps   include harmonic functions, geodesics, isometric minimal immersions, holomorphic (and anti--holomorphic) maps of K\"ahler manifolds. It is worthwhile to observe that in such a rich panorama, also the seemingly trivial case of constant maps  plays a basic role for the interplay between Ricci flow and (the perturbative quantization of)  non--linear $\sigma$ models.\\
\\
\noindent
In order to make the paper self--contained, we conclude this preliminary part with a capsule of the physical framework relating non--linear $\sigma$ models to Ricci flow. 

\subsection{NL$\sigma$M and Ricci flow in a nutshell}
\label{selfcont}

 In the above geometrical framework, let us consider on $Map\;(\Sigma,M)$ the action functional 
\begin{equation}
\mathcal{S}[\gamma,\phi;\, a, g, f]=a^{-1}\,\int_{\Sigma }\left[ tr_{\gamma(x)}\,(\phi^{*}\,g)+\,a\,f(\phi)\,\mathcal{K}\right]\,d\mu _{\gamma }\;.
\label{eqnActInt}
\end{equation}
where  $a>0$ is a parameter with the dimension of a length squared, $f:M\to \mathbb{R}$ is a  function on $M$, and $\mathcal{K}$ is the Gaussian curvature of $(\Sigma,\gamma)$. Geometrically the energy scale of the action $\mathcal{S}[\gamma,\phi;\, a, g, f]$ is set by the \emph{dilaton} coupling $\left[\,f(\phi)\,\mathcal{K}\right]$ and by the length scale of the target space 
metric $g_{ab}$, \emph{i.e.} \; $|Rm(g,y)|\,a$,\; where $|Rm(g,y)|:=[R^{iklm}R_{iklm}]^{1/2}$. It follows that the background fields $f\in C^{\infty }(M,R)$ and $g\in \mathcal{M}et(M)$ play the role of  point dependent coupling parameters $\alpha$ on $M$,  
\begin{equation}
\label{agf}
\alpha\,:=\,\left(a,\,g,\,f\right)\;,
\end{equation}
controlling the energetics of the action $\mathcal{S}[\gamma,\phi;\alpha]$. 
In quantum theory this fiducial action, together with its possible deformations, describes a family of $2$--dimensional QFTs known as (dilatonic) non--linear $\sigma$--models. They find applications ranging from  condensed matter physics to string theory.
\begin{rem}
Further coupling terms can be added to the action $\mathcal{S}[\gamma,\phi;\alpha]$, (see \emph{e.g.} \cite{gawedzki}),  in particular $a^{-1}\, \int_\Sigma U(\phi)\;d\mu_\gamma$ and $ a^{-1}\,\int_\Sigma \phi^*\Xi  $
 where $U\in C^{\infty }(M,R)$ and $\Xi \in C^{\infty }(M,\wedge  ^{2}\,TM^{*})$ is a 2-form on $M$. In order to discuss the role that Wasserstein geometry plays in the interplay between Ricci flow and non--linear $\sigma$ models we limit our analysis to the dilatonic field $f$.
\end{rem}
\noindent The (Euclidean) QFT associated with (\ref{eqnActInt})  is  characterized by the measure on $Map(\Sigma,M)$
 formally defined by 
\begin{equation}
D_{\alpha}[\phi]\;  e^{-S[\gamma,\phi; \alpha] }\;,
\label{correlations}
\end{equation}
and by its moment generating function, (correlations). Here $D_{\alpha}[\phi] $ is a (non--existent)  functional measure on $Map(\Sigma,M)$, possibly depending on the couplings $\alpha$, and normalized so that (\ref{correlations}) is  a probability measure. The somewhat fanciful expression (\ref{correlations}) hardly makes  sense, even at a physical level of rigor,  if we do not devise a way of controlling the spectrum of fluctuations of the fields $\phi\in Map(\Sigma,M)$.  Indeed, if we denote by $\mathcal{C}$ the space in which the couplings $\alpha$ are allowed to vary, then the  fundamental problem concerning (\ref{correlations}) is to introduce a filtration\footnote{We use here the term filtration in a rather loose sense. The relevant notion of progressively measurable maps over a non--decreasing family of sigma algebras in the appropriate functional space will be discussed in Section \ref{renGroup}.} in $\{Map(\Sigma,M)\times\mathcal{C}\,, e^{-S[\gamma,\phi; \alpha] }\,D_{\alpha}[\phi] \}$ ,
\begin{eqnarray}
\mathcal{RG}^{\tau}:\,[Map(\Sigma,M)\times\mathcal{C}]&\longrightarrow& [Map(\Sigma,M)\times\mathcal{C}]\nonumber\\
\label{reflow}\\
(\phi,\alpha)\;\;\;\;\;\;\;\;\;\;&\longmapsto&\;\;\; \mathcal{RG}^{\tau}(\phi,\alpha)=({\phi}_{\tau};{\alpha}(\tau))\;,\nonumber
\end{eqnarray}
which, as we vary the scale of distances $\tau$ at which we probe the Riemannian surface $\Sigma$, allows to tame the energetics of the fluctuations of the fields $\phi : \Sigma\to M$ in terms of a renormalization of the couplings $\alpha\mapsto\alpha(\tau)$. Such a filtration characterizes the 
\emph{renormalization group flow} associated with the measure space $\{Map(\Sigma,M)\times\mathcal{C},\,e^{-S[\gamma,\phi; \alpha] }\,D_{\alpha}[\phi] \}$.\\
\\
\noindent In order to describe this procedure in physical terms, select two scales of distances, say $\Lambda^{-1}$ and $\Lambda'^{-1}$, (one can equivalently interpret $\Lambda$ and $\Lambda'$ as the respective scales of momentum in the spectra of field fluctuations), with   $\Lambda'^{-1}>\Lambda^{-1}$.  The general idea, central in K.G. Wilson's analysis of the renormalization group flow, is to assume that, at least for $({\Lambda'}\setminus {\Lambda})$ small enough, we can put the $\mathcal{RG}^{\tau}$ push--forward of the functional measure $D_{\alpha}[\phi]\;  e^{-S[\gamma,\phi; \alpha] }$, \emph{viz.}\,\,   ${\mathcal{RG}}^{\Lambda\Lambda'}_{\sharp }\,(D_{\alpha}[\phi])\;  e^{-{\mathcal{RG}}^{\Lambda\Lambda'}\,S[\gamma,\phi; \alpha]}$, in the same form as the original functional measure, except for a small modification of the couplings $\alpha$. Explicitly, let $\Lambda'=e^{-\tau}\,\Lambda$, with $0<\tau<1$, and assume that for every such $\tau$ there exists a corresponding coupling $\alpha+\delta\,\alpha$ such that the following identity holds
\begin{equation}
\mathcal{RG}^{\tau}_{\sharp }\,(D_{\alpha}[\phi])\;  e^{-{\mathcal{RG}}^{\tau }\,S[\gamma,\phi; \alpha]}= D_{\alpha+\delta\alpha}[\phi]\;  e^{-S[\gamma,\phi; \alpha+\delta\alpha] }\;,
\label{infrenor}
\end{equation}
where we have denoted ${\mathcal{RG}}^{\tau}_\sharp $ the push--forward action of the map ${\mathcal{RG}}^{\Lambda\Lambda'}$ for $\Lambda'=e^{-\tau}\,\Lambda$.
In other words, we assume that an infinitesimal change in the cutoff can be completely \emph{absorbed} in an infinitesimal change of the couplings. If this equation  is valid at least to some order in $\tau$,  we can iteratively use it to see how the couplings $\alpha$ are affected by a finite change of the cutoff. If the theory is, along the lines sketched above, renormalizable by a renormalization of the couplings, many of its properties can be obtained by the analysis of the $\beta$--flow \emph{vector field} defined on the space of couplings $\mathcal{C}$ by 
\begin{equation}
\beta(\alpha(\tau)) := -\frac{\partial}{\partial \tau}\,\alpha(\tau)\;,
\label{betaflow}
\end{equation}
which formally appears as a natural geometrical flow on the measure space  $\left\{Map(\Sigma,M)\times\mathcal{C};\,D_{\alpha(\tau)}[\phi_{\tau}]\right\}$.\\
\\
\noindent
 According to (\ref{agf}),  the space of couplings $\mathcal{C}$ for the dilatonic non-linear $\sigma$-model action (\ref{eqnAct})  can be identified with the product of  $C^{\infty }(M,\mathbb{R})$, (where the dilatonic coupling $f$ varies), with the infinite--dimensional stratified manifold of Riemannian structures on $M$, (parametrizing the $\mathcal{D}iff(M)$-classes of metric couplings $g$), modulo overall length rescalings, (associated with the choice of the length parameter $a^{1/2}$), \emph{i.e.}
\begin{equation}
\mathcal{C}=C^{\infty }(M,\mathbb{R})\times \frac{\mathcal{M}et(M)}{\mathcal{D}iff(M)\times \mathbb{R}^{+}}
\end{equation}
where $\mathbb{R}^{+}$ denotes the group of rescalings defined by $a^{1/2}\mapsto\lambda a^{1/2}$, for $\lambda$ a positive number. As we have recalled, the true dimensionless coupling constants of the theory are the ratio of the length scale of the target space 
metric $g_{ab}(\phi(x))$  to $a$, and the dilatonic coupling $f(\phi(x))$. This remark has two important consequences:\\
\noindent \emph{(i)} It implies that as long as the curvature of target Riemannian manifold $(M,g)$ is small as seen by $(\Sigma, \gamma(x))$, \emph{viz.}\,$|Rm(g,\phi(x))|\,a\,<<1$, $\forall x\in\Sigma$, then  the formal  measure  $D_{(g,f)}[\phi]\,e^{-S[\phi;\,\alpha]}$ is concentrated
around the minima of the fiducial action $S[\phi;\,\alpha_{f=0}]$, \emph{i.e.} the constant maps $x\to\phi(x)=\phi_{0}$. \;\\
\noindent  \emph{(ii)} It also implies that the dilatonic coupling $f$ fluctuates around the constant value $f_0:=f(\phi_{0})$. In such a framework one can
 control the nearly Gaussian fluctuations $\delta\phi$ of $\phi$, and one can  address the (perturbative) renormalization group analysis described above \cite{DanAnnPhys, shore, Tseytlin3, Tseytlin4} so as to obtain  a perturbative $\beta$--flow for the coupling fields $\alpha=(g,\,a f)$. At leading order 
one gets
\\
\begin{equation}
\frac{\partial }{\partial \tau}\,g_{ik}(\tau)=2a\,\left(R_{ik}(\tau) \,+2\nabla _i\nabla _k{f}(\tau)\right)\,+\,\mathcal{O}\left((|Rm|a)^2\right)\;,
\label{2loopsfN}
\end{equation}
\begin{equation}
\frac{\partial }{\partial \tau}\,{f}(\tau)={c_{0}}- 2a\,\left(\frac{1}{2}\,\Delta {f}(\tau)-|\nabla {f}(\tau)|^{2}\right) \,+\,\mathcal{O}\left((|Rm|a)^2\right)\;,
\label{2f}
\end{equation} 
where $R_{ik}(\tau)$ denotes the Ricci tensor of $(M,g(\tau))$, and where  $c_{0}:=(\dim M-26)/6$ is the central charge, playing in QFT the role of dimension. 
If we set $\widehat{f}:=2f-2\,c_0\,\tau$ and  pass to the dimensionful variable $t:=\,-\,a\tau$, then, as $a\searrow 0$, the renormalization group flow (\ref{2loopsfN}) reduces to  Hamilton's Ricci flow deformed \cite{DeTurck} by the action of the $t$--dependent diffeomorphism generated by the vector field $W^i(t):=g^{ik}\nabla_k\,\widehat{f}(t)$, \emph{i.e.}, 
\\
\begin{equation}
\frac{\partial }{\partial t}\,g_{ik}(t)=-2\,\left(R_{ik}(t) \,+\,\nabla _i\nabla _k\widehat{f}(t)\right)\;,
\label{2loopsfRF}
\end{equation}
\begin{equation}
\frac{\partial }{\partial t}\,\widehat{f}(t)=\Delta\,\widehat{f}(t)-|\nabla \widehat{f}(t)|^2 \;.
\label{2fRF}
\end{equation}
\begin{rem}
If along the above RG flow we impose the (somewhat unphysical) dilatonic measure preservation constraint $\tfrac{\partial }{\partial t}\,e^{-\widehat{f}(t)}\,d\mu_{g_t}=0$, we get 
\begin{equation}
\frac{\partial }{\partial t}\,e^{-\widehat{f}(t)}\,d\mu_{g_t}=0=\left(|\nabla \widehat{f}(t)|^2-2\Delta\,\widehat{f}(t)-R(g(t)) \right)e^{-\widehat{f}(t)}\,d\mu_{g_t}\;.
\end{equation}
Inserted back in (\ref{2fRF}), this constraint gives rise to the Perelman version of the Ricci flow 
\begin{equation}
\frac{\partial }{\partial t}\,g_{ik}(t)=-2\,\left(R_{ik}(t) \,+\,\nabla _i\nabla _k\widehat{f}(t)\right)\;,
\label{2loopsfRFPer}
\end{equation}
\begin{equation}
\frac{\partial }{\partial t}\,\widehat{f}(t)=-\Delta\,\widehat{f}(t)-\mathcal{R}(g(t)) \;,
\label{2fRFPer}
\end{equation}
which couples the forward Ricci flow with a backward parabolic evolution of the dilatonic potential.  
\end{rem}

\section{The Wasserstein geometry of the dilaton field}

Let $(M,g,\,d\omega)$ be a  $n$--dimensional compact \emph{Riemannian metric  measure space} \cite{grigoryan}, \cite{gromov}, \emph{i.e.} a smooth orientable manifold, without boundary, endowed with a  Riemannian metric $g$ and a positive Borel measure $d\omega\,<<\,d\mu_g$, absolutely continuous  with respect to  the  Riemannian 
volume element, ${d\mu_g}$. Strictly speaking,  $(M,g,\,d\omega)$ characterizes a \emph{weighted Riemannian manifold}, (or \emph{Riemannian manifold with density}), the corresponding metric measure space being actually defined by $(M, d_g(\cdot ,\cdot ),\,d\omega)$,  where 
$d_g(x,y)$ denotes the Riemannian distance on $(M,g)$.  By a slight abuse of notation, we shall use $(M,g,\,d\omega)$ and $(M, d_g(\cdot ,\cdot ),\,d\omega)$  interchangeably. In such a framework, the two--dimensional dilatonic non--linear $\sigma$ model is defined by a rather natural  extension of the harmonic energy functional $E[\phi,\,g]$ to  $(M,g,\,d\omega)$. We start by recalling that the set ${\rm Meas}(M)$ of all smooth  Riemannian metric measure spaces  can be characterized as
\begin{equation}
{\rm Meas}(M):=\left\{(M,g; d\omega)\left.\right|(M,g)\in\mathcal{M}{\rm
et}(M),\,d\omega\in \mathcal{B}(M,g)  \right\}\;,
\end{equation}
where  $\mathcal{B}(M,g)$ is the set of positive Borel measure on $(M,g)$ with $d\omega<<d\mu_g$. Since in the compact--open $C^{\infty }$ topology $\mathcal{M}{\rm et}(M)$ is contractible, 
the space ${\rm Meas}(M)$ fibers trivially over $\mathcal{M}{\rm et}(M)$. In particular, the fiber $\pi^{-1}(M,g )$ can be identified with the set of all (orientation preserving) measures $d\omega\,<<\,d\mu_g$ over 
the given $(M,g)$, 
\begin{equation}
{\rm Meas}(M,g)\,:=\, \left\{d\omega\,\in {\rm Meas}(M)\,:\, d\omega<< d\mu_g \right\}\;, 
\end{equation}
endowed with the topology of weak convergence. There is a natural action of the group of 
diffeomorphisms  $\mathcal{D}iff(M )$ on  the space ${\rm Meas}(M)$, defined by
\begin{eqnarray}
\label{Dequiv}
\mathcal{D}iff(M )\times {\rm Meas}(M) &\;\longrightarrow\;& \;\;\;\;\;{\rm Meas}(M)\\
(\varphi;\;g,\,d\omega)\;\;\;\;\;\;\;\;\;&\longmapsto \;& (\varphi^*g,\;\varphi^*d\omega)\nonumber\;,
\end{eqnarray}
where $(\varphi^*g,\,\varphi^*d\omega)$ is the pull--back under $\varphi\in \mathcal{D}iff(M )$. The Radon--Nikodym derivative $\wp (g,\,d\omega)\,:=\,\tfrac{d\omega}{d\mu_g}$ is  a local \emph{Riemannian measure space invariant} \cite{gursky} under this action, \emph{i.e.} 
\begin{equation}
\label{equiv}
\wp (\varphi^*g,\,\varphi^*(d\omega))=\varphi^*\,\wp(g,\,d\omega)=\tfrac{d\omega}{d\mu_g}\circ \varphi\;,
\;\;\;\;\forall \varphi\in \mathcal{D}iff(M )\;,
\end{equation}
  and  we can introduce the 
\begin{defn} The geometrical dilaton field associated with the Riemannian metric measure space $(M,g,\,d\omega)$ is defined by the map
\begin{eqnarray}
\label{geodildef}
f\,:\,{\rm Meas}(M)\,&\longrightarrow &\,C^\infty(M,\mathbb{R})\\
(M,g,\,d\omega)\,&\longmapsto &\,f(M,g,\,d\omega)\,:=\,-\ln\left(V_g(M)\,\frac{d\omega}{d\mu_g} \right)\;,\nonumber
\end{eqnarray}
where $V_g(M):=\int_M\,d\mu_g$, and we can parametrize ${\rm Meas}(M,g)$ according to 
\begin{equation}
{\rm Meas}(M,g)\,=\, \left\{d\omega = e^{-f}\,\tfrac{d\mu _{g}}{V_g(M)}: f\in
C^{\infty }(M ,\mathbb{R})\right\} \;. \label{unomeas}
\end{equation}
\end{defn}
\begin{rem}
\label{remDiff}
In Riemann measure spaces, (and Ricci flow theory), sometimes it is necessary to restrict the action of $\mathcal{D}iff(M )$ to the metric $g$ alone and leave the measure $d\omega$ fixed, (cf. section 5 in \cite{gursky}). Absolute continuity of $d\omega$ with respect to $d\mu_g$ implies that in such a case  $f$ is not a scalar
and  under the action of a diffeomorphism $\varphi \in \mathcal{D}iff(M )$ it  transforms as a density according to    
\begin{equation}
\label{expJac}
e^{-\,f(\varphi^*g,\;d\omega)}\,=\,\,\left(Jac_{\,d\mu_g}\varphi^{-1}\right)\,e^{-f}\;,
\end{equation}
where $Jac_{\,d\mu_g}\varphi^{-1}$ denotes the Jacobian of $\varphi^{-1}$ with respect to $d\mu_g$. However,  a subtler interpretation is possible \cite{gursky} which still preserves the invariant nature of the dilaton. To wit, according to (\ref{expJac}) and Moser theorem \cite{moser}, one can consider $\varphi\in \mathcal{D}iff(M )$ as generating an active transformation on ${\rm Meas}(M,g)$ mapping the measure $d\omega$ into another measure $d\omega':= (\varphi^{-1})^*(d\omega)$, (with the same total mass $\int_M\,d\omega$). By composing this mapping on ${\rm Meas}(M,g)$ with the induced pull--back action on the resulting measure  we get from (\ref{equiv}) that the Radon--Nikodym derivative $\wp (g,\,d\omega)$ transforms according to
\begin{eqnarray}
\label{}
&&\wp (g,\,d\omega)\longmapsto \wp (\varphi^*g,\,\varphi^*(d\omega'))\,=\,\varphi^*\,\wp(g,\,d\omega')\\
\nonumber\\
&&=\,\varphi^*\,\wp(g,\,(\varphi^{-1})^*(d\omega))\,=\,\left(Jac_{\,d\mu_g}\varphi^{-1}\right)\frac{d\omega}{d\mu_g}\;,\nonumber
\end{eqnarray}
which is consistent both with (\ref{expJac}) and the local Riemannian measure space invariant nature of the dilaton.
If not otherwise stated we shall consider the natural action on $(M,g,\,d\omega)$ given by (\ref{Dequiv}).
\end{rem}

\noindent With these remarks along the way, we can geometrically characterize the action functional (\ref{eqnActInt}) according to
\begin{defn} (\emph{The dilatonic NL$\sigma$M action}).
If $\phi\in \mathcal{H}^{1}(\Sigma,M)$ is a localizable map, then the associated non--linear $\sigma$ model dilatonic action with coupling parameters $a\in\mathbb{R}_{>0}$ and  $(M,g,\,d\omega)\in {\rm Meas}(M,g)$,  is defined by 
\begin{eqnarray}
\label{eqnAct}
&& (\Sigma,\gamma)\times \mathcal{H}^{1}(\Sigma,M)\times \left[\mathbb{R}_{>0}\times {\rm Meas}(M,g)\right]\,\longrightarrow \,\mathbb{R}\\
\nonumber\\
&&(\gamma,\phi;\,a,\,(M,g,\,d\omega))\,\;\;\;\;\;\;\longmapsto \,\;\;\;\;\;\;\;\;\;\;\;\;{S}[\gamma,\phi;\,\,a,\,g,\,d\omega]\nonumber\\
\nonumber\\
&&:=\, a^{-1}\,\int_{\Sigma}\,\left[ tr_{\gamma(x)}\,(\phi^{*}\,g)-\,a\,\mathcal{K}_{\gamma}\,
\ln\,\phi^*\left(\frac{d\omega}{d\mu_g}\,V_g(M)\right)\right]\,d\mu _{\gamma }\nonumber\\
\nonumber\\
&&:=\,\frac{2}{a}\,E[\phi,\,g]_{(\Sigma,M)}\,+\,
\int_{\Sigma}\,\mathcal{K}_{\gamma}\,f(\phi)\,d\mu _{\gamma } \;,\nonumber
\end{eqnarray}
\\
\noindent
where $\mathcal{K}_{\gamma}$ is the  Gaussian curvature of the Riemannian surface $(\Sigma,\gamma)$.
\end{defn}
\begin{rem} 
As already stressed in the introductory remarks, the parameter $a>0$ sets the (squared) length scale at which the pair  $\left(\phi(\Sigma),\,{S}\right)$ probes the target metric measure space $(M,g,\,d\omega)$. Writing $\ln\phi^*\left(\frac{d\omega}{d\mu_g}\,V_g(M)\right)$ in place of  $-f(\phi)$ may appear pedantic, however it emphasizes the often overlooked fact that the \emph{dilaton} coupling $f$ actually is the ($Diff(M)$--equivariant) assignment of a measure $d\omega$ in $(M,g)$, and that we are dealing with a metric measure space and not just with a Riemannian manifold. In the definition, we also stressed the role of ${\rm Meas}(M,g)$ as the space of \emph{point--dependent coupling parameters} for dilatonic non--linear $\sigma$ models. As we have recalled, this latter aspect features prominently in the perturbative quantization of the model and its connection to Ricci flow. As we shall see, it plays a basic role also in motivating the use of Wasserstein geometry  in the theory. 
\end{rem}
\begin{rem}
It is also important to recall again that, in stark contrast with the harmonic map energy (\ref{Msigmamod0}), the dilatonic term in (\ref{eqnAct}),  
\begin{equation}
-\,\int_{\Sigma}\,\mathcal{K}_{\gamma}\,
\ln\,\phi^*\left(\frac{d\omega}{d\mu_g}\,V_g(M)\right)\,d\mu _{\gamma }= \,\int_{\Sigma}\,\mathcal{K}_{\gamma}\,f(\phi)\,d\mu _{\gamma } \;,
\label{dilaction}
\end{equation}
is not conformally invariant. As is well known, and as first stressed by E. Fradkin and A. Tseytlin \cite{FradTseyt}, the role of this term is to restore the conformal invariance of $E[\phi,g]$ which is broken upon quantization.
\end{rem}
 To discuss the role that Wasserstein geometry plays in non--linear $\sigma$ model theory,  we introduce the following characterization
\begin{defn}
If ${\rm Prob}(M)$ denote the set of all Borel probability measure on the manifold $M$, then the space of probability--normalized dilaton fields over $(M,g)\in\mathcal{M}{\rm et}(M)$, is the dense subspace of ${\rm Prob}(M)$  defined by  
\begin{equation}
{\mathcal DIL}_{(1)}(M,g)\,:=\, \left\{\left. d\omega\in {\rm Prob}(M)\right|\, d\omega:=e^{-f}\,\tfrac{d\mu _{g}}{V_g(M)},\, f\in
C^{\infty }(M ,\mathbb{R})\right\}\;. \label{measuno}
\end{equation}
\end{defn} 
\begin{rem} 
It is worthwhile stressing that the restriction  to ${\mathcal DIL}_{(1)}(M,g)$  is somewhat unphysical from the point of view of non--linear $\sigma$ model theory, since it constrains a priori the dilaton field to be associated to a probability measure. As artificial as it may appear, this restriction plays a basic role in Perelman's analysis of the Ricci flow, and it turns out to be appropriate also  in the geometric analysis of the interaction between non--linear $\sigma$ model, heat kernel embedding, and Ricci flow discussed here.
\end{rem}
Clearly ${\mathcal DIL}_{(1)}(M,g)\,\approx {\rm Prob}_{ac}(M,g)$, the set of absolutely continuous probability measures $d\omega<<d\mu_g$ on $(M,g)$. We use this identification to characterize  
${\mathcal DIL}_{(1)}(M,g)$   as an infinite dimensional manifold locally modelled over the Hilbert space completion of the smooth tangent space
\begin{equation}
T_{\omega}Prob_{ac}(M,g )\,\simeq 
\{h \in C^{\infty }(M,\mathbb{R}),\,\int_{M}h\,d\omega =0\}\;,
\end{equation}
with respect to the Otto inner product on ${\rm Prob}_{ac}(M,g)$  defined, at the given $d\omega=\,V^{-1}_g(M)e^{-f}d\mu_{g}$, by the $L^2(M,d\omega)$ Dirichlet form \cite{otto1}
\begin{equation}
\left\langle \nabla \varphi , \nabla \zeta \right\rangle _{(g,d\omega)}\,\doteq\,
\int_{M }\left(g^{ik}\,\nabla _{k}\varphi \,\nabla _{i}\zeta\right) \,d\omega\;,
\label{inner}
\end{equation}
for any $\varphi$, $ \zeta $ $\in C^{\infty }(M, \mathbb{R})/\mathbb{R}$. We set
\\
\begin{eqnarray}
\label{HilbSPace}
&&T_{f}{\mathcal DIL}_{(1)}(M,g)\\
\nonumber\\
&&\doteq \, \overline{\left\{h \in C^{\infty }(M,\mathbb{R}),\,\int_{M} h\,d\omega =0\right\}}^{L^2(M, d\omega)}\; .\nonumber
\end{eqnarray}
Under this identification, one can represent infinitesimal deformations of the dilaton field $d\omega$, (thought of as vectors in $T_{f}{\mathcal DIL}_{(1)}(M,g)$ $\simeq $  $T_{\omega}Prob_{ac}(M,g )$), in terms of the mapping 
\begin{eqnarray}
T_{f}{\mathcal DIL}_{(1)}(M,g)\times {\mathcal DIL}_{(1)}(M,g)&\longrightarrow & 
C^{\infty }(M ,\mathbb{R})/\mathbb{R}\; ,\\
(h,\,d\omega=\,V^{-1}_g(M)e^{-f}d\mu_{g})\;\;\;\;\;\;\;&\longmapsto &\;\;\;\;\; \psi \nonumber\;,
\end{eqnarray}
where the function $\psi $ associated to the given $(h,d\omega)$ is formally determined on  $(M ,g)$ by the elliptic partial differential equation
\begin{equation}
-\nabla ^{i}\left( e^{-f}\,\nabla _{i}\psi \right) =h\,e^{-f},
\label{otto110}
\end{equation}
under the equivalence relation  identifying any two such 
solutions differing by an additive constant. Recall that if $\mathcal{L}_V\,d\omega$ denotes the Lie derivative of the volume form $d\omega$ along the vector field $V\in\,C^{\infty }(M,TM)$, then the weighted divergence  associated with the Riemannian measure space $(M,g,\,d\omega)$ is defined by
\begin{equation}
\label{weightdiv}
\mathcal{L}_V\,d\omega\,=\,\left({div}_{\omega}\,V\right)\,d\omega\,=\, \left[e^{f}\,\nabla_i\,\left(e^{\,-\,f}\,V^i  \right)\right]\,d\omega\;.
\end{equation}
It follows that the elliptic equation (\ref{otto110}) can be equivalently written as
\begin{equation}
\bigtriangleup _{\omega}\,\psi\, =\,-\,h\;,
\label{otto11}
\end{equation}
where $\bigtriangleup _{\omega}$ denotes the  weighted Laplacian on $(M,g,\,d\omega)$   \cite{chavelfeldman}, \cite{davies}, \cite{grigoryan},   
\begin{equation}
\label{Wlap0}
\bigtriangleup _{\omega}\,:=\,{div}_{\omega}\,\nabla \,=\,\bigtriangleup _{g}\,-\,\nabla f\cdot \nabla\;.
\end{equation}
Let us denote by $G_\omega:\,M\times M\,\longrightarrow\, \mathbb{R}$ the Green function associated to the operator $-\,\bigtriangleup _{\omega}$ acting on functions with vanishing $d\omega$--mean, $\int_{M}\,\psi(y)\,d\omega(y)\,=\,0$, and satisfying $\int_{M}\,G_\omega(x,y)\,d\omega(y)\,=\,0$,
\emph{i.e.} the solution of 
\begin{eqnarray}
\bigtriangleup _{\omega,x}\,G_\omega(x,y)\,=\,-\,\delta_y\,+\,1\;,
\end{eqnarray}
where $\delta_y\,d\omega$ is the Dirac measure at $y\in(M,g,\,d\omega)$. On $M\times M$ minus the diagonal $\{x=y\}$,\, $G_\omega(x,y)$ is smooth and we can write
\begin{equation}
\psi(x)\,=\, \mathcal{G}_\omega\,h(x)\,:=\int_M\,h(y)\,G_\omega(x,y)\,d\omega(y)\;,
\label{GWL}
\end{equation}
whenever the integral makes sense. Conversely,   given $\varphi \in C^{\infty }(M ,\mathbb{R})/\mathbb{R}$,  we can formally define a vector
\begin{eqnarray}
{\rm Prob}_{ac}(M,g)\,&\longrightarrow &\,T_{\omega}{\rm Prob}_{ac}(M,g)\\
d\omega\,&\longmapsto &\,\vee _\varphi\nonumber\;,
\end{eqnarray}
by its action on smooth functionals $\mathcal{J}\in C^{\infty}({\rm Prob}_{ac}(M,g),\,\mathbb{R})$ according to \cite{lott1b}
\begin{equation}
\label{lapvee}
\left(\vee_\varphi\, \mathcal{J}\right)(d\omega)\,=\,\left.\frac{d}{d\epsilon}\right|_{\epsilon=0}
\,\mathcal{J}\left(d\omega-\epsilon\,{\bigtriangleup }_{\omega}\,\varphi\,d\omega\right)\;,
\end{equation}
and we have the isomorphism
\begin{eqnarray}
\label{veciso}
C^{\infty }(M ,\mathbb{R})/\mathbb{R}\,&\longrightarrow &\,T_\omega{\rm Prob}_{ac}(M,g)\simeq T_{f}{\mathcal DIL}_{(1)}(M,g)\\
\psi\,&\Longleftrightarrow &\,\vee _\psi\nonumber\;.
\end{eqnarray}

\subsection{Diffeomorphism group and metric measure spaces}
\label{dgmms}
According to the definition (\ref{HilbSPace}) and the above remarks,  the tangent space $T_{\omega}{\rm Prob}_{ac}(M,g)$ at the given dilatonic measure $d\omega\in {\rm Prob}_{ac}(M,g)$ can be identified with the subspace of gradient vector fields in the set of all $L^2(M,d\omega)$ vector fields over $(M,g,\,d\omega)$. This is deeply connected with the geometry of $\mathcal{D}iff(M )$. The link is well--known in the case of Riemannian manifolds \cite{lukatsky}, \cite{smolentsev1}, \cite{smolentsev2}, where the Remark \ref{remDiff} applies, (a very readable account is provided in \cite{khesin}). If one wants to stress the underlying structure of the Riemannian metric measure space $(M,g,\,d\omega)$ over which $\mathcal{D}iff(M )$ is acting, then there are some peculiarities  which, for our purposes, it is useful to make explicit. For technical reasons \cite{marsdenLibro} let us consider  the (topological) group, $\mathcal{D}iff^{\;s}(M )$, defined by the set of diffeomorphisms which, as maps $M \rightarrow M$ are an open subset of the Sobolev space of maps  $\mathcal{H}^{\,s}(M, M)$, with $s\,>\, \tfrac{n}{2}+1$. As we mentioned above, for a given $(M,g,\,d\omega)$, $M$ compact, and for any  $d\omega'\in {\rm Prob}_{ac}(M,g)$ Moser theorem \cite{moser} implies that there is a diffeomorphism $\eta\in \mathcal{D}iff^{\;s}(M )$,\,$s\,>\, \tfrac{n}{2}+1$,\ such that $d\omega'\,=\,\eta^*(d\omega)$. In more precise terms, there is a smooth map $\chi\,:{\rm Prob}_{ac}(M,g)\longrightarrow \mathcal{D}iff^{\;s}(M )$, canonically defined for the given $(M,g,\,d\omega)$, such that the pull back--action
\begin{eqnarray}
\label{pbJact}
J_{\omega}:\,\mathcal{D}iff^{\;s}(M )\,&\longrightarrow&\, {\rm Prob}_{ac}(M,g)\\
\eta\,&\longmapsto&\, J_{\omega}(\eta)\,:=\,\eta^*(d\omega)\,=\,d\omega'\nonumber
\end{eqnarray}  
satisfies $J_{\omega}\circ \chi\,=\,e $, where $e $ is the identity in $\mathcal{D}iff^{\;s}(M )$.
\begin{rem} This well--known result (cf. \cite{ebinmarsden}) is usually stated in terms of the reference measure provided by the (normalized) Riemannian volume element $V_g(M)^{-1}\,d\mu_g$. On a given Riemannian metric measure space $(M,g,\,d\omega)$ one can use as reference measure $d\omega$ as well.
\end{rem}
\noindent If we denote by $J _{\omega}^{-1}(d\omega')$  the fiber of   $J _{\omega}$ over $d\omega'\in {\rm Prob}_{ac}(M,g)$ then we have $J_{\omega}^{-1}(d\omega')\,=\,\mathcal{D}iff^{\;s+1}_{\omega'}(M )$, where 
\begin{equation}
\mathcal{D}iff^{\;s+1}_{\omega}(M )\,:=\,\left\{\left.\varphi \,\in\,\mathcal{D}iff^{\;s+1}(M )\right|\,\varphi^*(d\omega')=\,d\omega'  \right\}
\end{equation} 
is the group of $d\omega'$--preserving diffeomorphisms. Let us recall that the tangent space to $\mathcal{D}iff^{\;s}(M )$ at the identity, \;$T_e \,\mathcal{D}iff^{\;s}(M )$,\, consists of all  $\mathcal{H}^{s}(M,\, TM)$ vector fields on $M$. More generally, the tangent space $T_\vartheta\,\mathcal{D}iff^{\;s}(M )$ at a diffeomorphism $\vartheta\in \mathcal{D}iff^{\;s}(M)$ is given by
\begin{equation}
T_\vartheta\,\mathcal{D}iff^{\;s}(M )\,:\,\left\{\left.U\in \mathcal{H}^{s}(M,\, TM)\right|\;U \;\;
{\rm covers}\;\;\;\vartheta\in  \mathcal{D}iff^{\;s}(M )  \right\}\;,
\end{equation}  
\emph{i.e.} by the space of $ \mathcal{H}^{s}$  sections of the pull--back bundle
\begin{equation}
\vartheta^{-1} TM:=\{U=u\circ \vartheta,\;u\in C^{\infty}(M,TM)\}\;. 
\end{equation}  
Hence, if $U\in T_\vartheta\,\mathcal{D}iff^{\;s}(M )$ then $U\circ \vartheta^{-1}\in \mathcal{H}^{s}(M,\, TM)$. By extending the standard approach \cite{ebinmarsden} to the Riemannian measure space $(M,g,\,d\omega)$,  we introduce on $\mathcal{D}iff^{\;s}(M )$ the weak  $L^{2}(M, d\omega)$ inner product defined 
at $T_\vartheta\,\mathcal{D}iff^{\;s}(M )$  by
\begin{equation}
\label{inndiff}
\left\langle U,\,V   \right\rangle_{(\omega,\vartheta)}\,:=\,\int_M\,g(u\circ\vartheta(x),\,v\circ\vartheta(x) )_{\vartheta(x)}\,d\omega(\vartheta(x))\;,
\end{equation}
for all $U,\,V\,\in T_\vartheta\,\mathcal{D}iff^{\;s}(M )$. We can equivalently write (\ref{inndiff}) as
\begin{eqnarray}
&&\left\langle U,\,V   \right\rangle_{(\omega,\vartheta)}\,=\,\int_{\vartheta^{-1}(M)}\,g(u(x),\,v(x))_x\,(\vartheta^{-1})^*(d\omega)(x)\\
\nonumber\\
&&\,=\,\int_{\vartheta^{-1}(M)}\,g(u,\,v)\,Jac_{\omega}\,(\vartheta^{-1})\,d\omega=\int_{M}\,g(u,\,v)\,
Jac_{\omega}\,(\vartheta^{-1})\,d\omega\nonumber\;,
\end{eqnarray}
where we have used $\vartheta^{-1}(M)=M$ and denoted by $Jac_{\omega}\,(\vartheta^{-1})$ the Jacobian of $\vartheta^{-1}$ computed with respect to $d\omega$. Let $\eta \in \mathcal{D}iff^{\;s}(M )$ and denote by
\begin{eqnarray}
\label{puforw}
T_\vartheta\,\mathcal{D}iff^{\;s}(M ) &\longrightarrow& T_{\vartheta\circ \eta}\,\mathcal{D}iff^{\;s}(M )\\
\nonumber\\
U_\vartheta\,\;\;\;\;\;&\longmapsto &\,\left(\mathfrak{R}_\eta \right)_*\,U_\vartheta\,:=\,U_\vartheta\circ \eta\nonumber\;,
\end{eqnarray}
the push--forward of $U_\vartheta$ under the (smooth) right action defined on $\mathcal{D}iff^{\;s}(M )$ by $\mathfrak{R}_\eta:\vartheta\longmapsto \vartheta\circ \eta$. According to (\ref{inndiff}) one computes
\begin{eqnarray}
\label{inndiffR}
&&\left\langle \left(\mathfrak{R}_\eta \right)_*U,\,\left(\mathfrak{R}_\eta \right)_*V   \right\rangle_{(\omega,\vartheta\circ \eta)}\\
\nonumber\\
&&\,=\,\int_M\,g(U\circ\eta(x),\,V\circ\eta(x) )_{\vartheta\circ \eta (x)}\,d\omega\;,\nonumber
\nonumber\\
&&\,=\,\int_{\eta^{-1}(M)}\,g(U(x),\,V(x) )_{\vartheta (x)}\,(\eta^{-1})^*(d\omega)\;.\nonumber
\end{eqnarray} 
If $\eta\in \mathcal{D}iff^{\;s}_\omega(M )$ then $(\eta^{-1})^*(d\omega)\,=\,d\omega$ and  
 (\ref{inndiff}) is right--invariant when restricted to the group of $d\omega$-volume preserving diffeomorphisms $\mathcal{D}iff^{\;s}_{\omega}(M )$. In this connection, let us consider the weighted $L^2(M,d\omega)$ Helmholtz decomposition of vector fields which are of Sobolev class $\mathcal{H}^{s}(M,TM)$,\, $s> n/2 +1$,  on $(M,g,\,d\omega)$
\begin{equation}
\label{TdiffHelm}  
\mathcal{H}^{s}(M,TM)\,\simeq \,T_e\mathcal{D}iff^{\;s}(M)\,=\, Im\nabla \oplus Ker \nabla ^{*_\omega}\;,
\end{equation}
where $\nabla ^{*_\omega}=\,-div_\omega$ denotes the $L^2(M,d\omega)$ adjoint of the gradient $\nabla: \mathcal{H}^{s+1}(M,\mathbb{R})\rightarrow \mathcal{H}^{s}(M,TM)$. Hence, according to the Otto characterization 
(\ref{inner}) of $T_{\omega}{\rm Prob}_{ac}(M,g)$ we can write 
\begin{equation}
T_e\mathcal{D}iff^{\;s}(M)\, =\, T_{\omega}{\rm Prob}_{ac}(M,g)\,\oplus\,T_e\mathcal{D}iff^{\;s}_\omega(M)\;, 
\end{equation}
where we have identified $Ker \nabla ^{*_\omega}$ with the tangent space to $\mathcal{D}iff^{\;s}_\omega(M)$ at the identity  map. Explicitly, for a given $v\in \mathcal{H}^{s}(M,TM)$ we have
\begin{equation}
v\,=\,-\,\nabla \varphi\,\oplus v_\dagger\;,
\label{Helmholtz} 
\end{equation}
where $v_\dagger$ is $div_\omega$--free,\, and $\varphi$ is defined by the solution of the elliptic equation
\begin{equation}
\bigtriangleup _\omega \varphi \,=\,-\,div_\omega\,v\;,
\end{equation}
which, according to (\ref{lapvee}), characterizes the correspondence $\vee _v\Leftrightarrow \varphi$. It is useful to rewrite (\ref{Helmholtz})  as
\begin{equation}
\label{proj1}
v\,=\,\Pi _\omega\,v\,\oplus \,(I-\Pi _\omega)\,v\;
\end{equation}
where 
\begin{equation}
\label{proj2}
\Pi _\omega\,v(x)\,:=\,\nabla_x\,\left(\mathcal{G}_\omega\circ div_\omega\right)\,v\,=\,\nabla_x\,
\int_M\,G_\omega(x,y)\,div_{\omega,y} v(y)\,d\omega(y)\;,
\end{equation}
is the $L^2(M,d\omega)$--orthogonal projector onto  $Im\nabla $ and $(I-\Pi _\omega)$ is the projector onto $Ker \nabla ^{*_\omega}\simeq T_e\mathcal{D}iff^{\,s}_\omega(M)$. We can  extend  $(I-\Pi _\omega)$ to the elements $V$ of the tangent space $T_\eta\mathcal{D}iff^{\,s}$,\; $\eta\in \mathcal{D}iff^{\;s}_\omega(M)$, by defining the right--invariant projector
\begin{equation}
(I-\Pi _\omega)_\eta\,V\,:=\,\left((I-\Pi _\omega)(V\circ \eta^{-1})\right)\circ \eta\;.
\end{equation}  
Correspondingly we have the $L^2(M,d\omega)$ weak orthogonal splitting
\begin{equation}
T_\eta\mathcal{D}iff^{\,s}\,\simeq \,\left(\mathfrak{R}_\eta\right)_*\left(T_{\omega}{\rm Prob}_{ac}(M,g)\right)\,\oplus\,T_\eta\mathcal{D}iff^{s}_\omega(M)\;,
\end{equation} 
where $(\mathfrak{R}_\eta)_*\left(T_{\omega}{\rm Prob}_{ac}(M,g)\right):=(\mathfrak{R}_\eta)_*\left(Im\,\nabla \right)$ is the $\eta$--right translated spaces of gradient vector fields $Im\,\nabla \subset T_e\mathcal{D}iff^{\,s}$. 
Note that the weak Riemannian metric (\ref{inndiff}) is right invariant on $\mathcal{D}iff^{s}_\omega(M)$, and this allows to consider the Otto inner product (\ref{inner}) as inducing on the space of smooth probability measure  ${\rm Prob}_{ac}(M,g)\simeq \mathcal{D}iff^{\,s}(M)/\mathcal{D}iff^{s}_\omega(M)$, \,$s>\,n/2\,+1$,\, a weak Riemannian structure such that the 
projection, (see (\ref{pbJact})),
\begin{eqnarray}
\label{Jproject}
\mathfrak{J}_{\omega}:\,\mathcal{D}iff^{\;s}(M )\,&\longrightarrow&\, {\rm Prob}_{ac}(M,g)\\
\eta\,&\longmapsto&\, \mathfrak{J}_{\omega}(\eta)\,:=\,\left(\eta^{\,-1}\right)^*(d\omega)\,=\,d\omega'\nonumber
\end{eqnarray} 
is a Riemannian submersion. According to \cite{ebinmarsden} one can exploit (\ref{puforw}) to introduce on $\mathcal{D}iff^{\;s}(M )$ a torsion--free connection which is compatible with the weak Riemannian structure defined by (\ref{inndiff}),\,\emph{viz.}
\begin{equation}
\label{diffconn0}
T_\eta\mathcal{D}iff^{\,s}\,\ni \breve{\nabla }_{U}\,V\,:=\,\left(\mathfrak{R}_\eta \right)_*\,\left(\nabla^{(g)} _u\,v  \right)\;, 
\end{equation}
or, more explicitly,
\begin{equation}
\label{diffconn}
\left\langle \breve{\nabla }_{U}\,V,\,W\right\rangle_{(\omega,\eta)}
=\,\int_M\,g\left(\nabla ^{(g)}_{u\circ \eta}(v\circ \eta),\,w\circ\eta\right)\,d\omega(\eta(x))\;,  \nonumber
\end{equation}
where the tangent vectors $u,\,v,\,w$,\,$\in \mathcal{H}^{\,s}\simeq T_e\mathcal{D}iff^{\,s+1}(M)$,\; $U=u\circ \eta,\,V=v\circ \eta,\,W=w\circ \eta$,\,$\in T_\eta\mathcal{D}iff^{\,s}$,\, and $\nabla^{(g)}$ denotes the Levi--Civita connection of $(M,g,\,d\omega)$. 
The curvature of such a connection at $\eta\in \mathcal{D}iff^{\,s}$ can be computed  according to \cite{smolentsev1}
\begin{equation}
\breve{R}(U,\,V)\,Z:=\,\left(\mathfrak{R}_\eta \right)_*\,\left(R_{(g)}(u,\,v)\,z\right)\;, 
\end{equation}
where $Z=z\circ \eta$ and $R_{(g)}(u,\,v)\,z$ is the Riemann tensor of $(M,g,\,d\omega)$. Explicitly,
\begin{eqnarray}
&&\left\langle \breve{R}(U,\,V)\,Z,\,W\right\rangle_{(\omega,\eta)}\\
\nonumber\\
&&=\,\int_M\,g\left(R_{(g)}(u\circ \eta,v\circ \eta)z\circ\eta,\,w\circ\eta\right)\,d\omega(\eta(x))\;,  \nonumber
\end{eqnarray}
provides the component of the curvature tensor associated with the weak Riemannian structure (\ref{inndiff}) at the \emph{point} $\eta\in \mathcal{D}iff(M)$. In particular, the components at the identity $e \in \mathcal{D}iff(M)$ are given by the $d\omega$--average
\begin{equation}
\label{curbreve}
\left\langle \breve{R}(U,\,V)\,Z,\,W\right\rangle_{(\omega,e)}\,=\,\int_M\,g\left(R_{(g)}(u,v)z,\,w\right)\,d\omega(x)\;. 
\end{equation}
 In \cite{smolentsev1} this analysis is extended to the computation of the curvature of the space of probability measures ${\rm Prob}_{ac}(M,g)$, a point that we address in the next section.   
\begin{rem}
\label{heuremotto}
The Helmholtz decomposition can be seen as a linearization of 
 Brenier's polar factorization theorem \cite{mccan}, \cite{Villani} and plays a subtle role in the Wasserstein geometry of the smooth $\infty$--dimensional manifold ${\rm Prob}_{ac}(M,g)$. 
 The above analysis, based on the geometry of $\mathcal{D}iff(M )$ is somehow formal since it heavily relies on the smoothness of ${\rm Prob}_{ac}(M,g)$ and of the Riemannian submersion map (\ref{Jproject})  which typically fails for general probability measures. Nonetheless, as presciently shown by F. Otto \cite{otto1}, this formal Riemannian structure has many remarkable properties and provides the backbone for a fully rigorous description of the geometry of $({\rm Prob}(M),\,d_g^W)$, (see \emph{e.g.} \cite{savare}, \cite{ambrosio}, \cite{nicolaG}, \cite{lott2},  \cite{sturm}, \cite{VillaniON}). 
\end{rem}

\subsection{Wasserstein distance and calculus on dilaton space}
Let ${\rm Prob}(M\times M)$ denote the set of Borel probability measures on the product space $M\times M$.
In order to compare any two distinct dilaton fields in  ${\mathcal DIL}_{(1)}(M,g)$, say $(M, g, d\omega_1=e^{-f_1}\,\tfrac{d\mu _{g}}{V_g(M)})$ and $(M, g, d\omega_2=e^{-f_2}\,\tfrac{d\mu _{g}}{V_g(M)})$, let us consider  the set of measures $d\sigma\in {\rm Prob}(M\times M)$  which reduce to $d\omega_1$ when restricted to the first factor and to $d\omega_2$  when restricted to the second factor, \emph{i.e.}
\begin{eqnarray}
\label{probprod}
&&{\rm Prob}_{\,\omega_1,\,\omega_2}\,(M\times M)\\
\nonumber\\
&&\,:=\, \left\{d\sigma \,\in {\rm Prob}(M\times M)\,
\left.  \right|\,\pi^{(1)}_\sharp \,d\sigma=d\omega_1,\, \pi^{(2)}_\sharp \,d\sigma=d\omega_2 \right\}\;,\nonumber
\end{eqnarray}
where $\pi^{(1)}_\sharp $ and $\pi^{(2)}_\sharp $ refer to the push--forward of  $d\sigma$ under  the projection maps $\pi^{(i)}$ onto the factors of $M\times M$. Measures $d\sigma \in {\rm Prob}_{\,\omega_1,\,\omega_2}\,(M\times M)$ are often referred to as \emph{couplings} between $d\omega_1$ and $d\omega_2$. We shall avoid such a terminology since in our setting  the term coupling has quite different a meaning.  Let us recall that given  a (measurable and non--negative) cost function $c\,:\,M\times M\rightarrow \mathbb{R}$, an optimal transport plan \cite{kantoro} $d\sigma_{opt}\in {\rm Prob}_{\,\omega_1,\,\omega_2}(M\times M)$ between the probability measures $d\omega_1$ and $d\omega_2$ in ${\rm Prob}(M)$, (not necessarily in ${\rm Prob}_{ac}(M,g)$),    is defined by the infimum, over all  $d\sigma(x, y)\in {\rm Prob}_{\,\omega_1,\,\omega_2}(M\times M)$, of the total cost functional
\begin{equation}
\int_{M\times M}\,c(x, y )\,d\sigma(x, y)\;.
\end{equation}
 On a Riemannian manifold $(M,g)$, the typical cost function is provided  \cite{lott1}, \cite{lott2}, 
\cite{mccan} by the squared Riemannian distance function $d^{\,2}_{\,g}(\cdot ,\cdot )$, and a major result of the theory \cite{brenier}, \cite{cordero},  \cite{mccan}, \cite{sturm}, is that  for any pair $d\omega_1$ and $d\omega_2$ $\in{\rm Prob}(M)$, there is an optimal transport plan $d\sigma_{opt}$, induced by a map $\Upsilon _{opt}:M\rightarrow M$.  The resulting expression for the total cost of the plan
\begin{equation}
\label{wassdist}
d_{\,g}^{\,W}\,(d\omega_1, d\omega_2)\,:=\,\left(\,\inf_{d\sigma\in {\rm Prob}_{\,\omega_1,\,\omega_2}(M\times M)} 
\;\int_{M\times M}\,d^{\,2}_{\,g}(x, y )\,d\sigma(x, y) \right)^{1/2}\;,
\end{equation}
characterizes the quadratic Wasserstein, (or more appropriately, Kantorovich-Rubinstein) distance between the two probability measures $d\omega_1$ and $d\omega_2$.
\begin{rem}
 Note that there can be distinct optimal plans $d\sigma_{opt}$ connecting general probability measures $d\omega_1$ and $d\omega_2$ $\in \rm{Prob}(M)$, whereas on ${\rm Prob}_{ac}(M,g)$ the optimal transport plan is unique.
\end{rem}
\noindent
 The quadratic Wasserstein distance  $d_{\,g}^{\,W}$ defines a finite metric on ${\rm Prob}(M)$ and it can be shown that $({\rm Prob}(M),\,d_{\,g}^{\,W})$ is a geodesic space, endowed with the weak--* topology, (we refer to \cite{savare}, \cite{Villani}, \cite{VillaniON} for  the relevant properties of Wasserstein geometry and optimal transport we freely use in the following). 
By an obvious dictionary, we identify the distance between the two dilaton fields $f_1$ and $f_2$ with the Wasserstein distance  $d_{\,g}^{\,W}\,(d\omega_1, d\omega_2)$ between the corresponding probability measures. This allows to  characterize $\left({\mathcal DIL}_{(1)}(M,g),\,d_{\,g}^{\,W}\right)$ as the (dense) subset,\,$\left({\rm Prob}_{ac}(M, g),\,d_{\,g}^{\,W}\right)$,   
of the (quadratic) Wasserstein space $({\rm Prob}(M),\,d_{\,g}^{\,W})$. By way of illustration of the interplay between dilaton fields and Wasserstein geometry we have the 
\begin{lem}
\label{McClemma}
Let $f_1$ and $f_2$ two distinct dilaton fields in ${\mathcal DIL}_{(1)}(M,g)$, and let $x\longmapsto \exp_x(-\nabla \psi(x))$ be the optimal map between the corresponding probability measures $d\omega_1=\,V_g^{-1}(M)\,e^{-f_1}d\mu_g$ and $d\omega_2=\,V_g^{-1}(M)\,e^{-f_2}d\mu_g$. Then the unique geodesic in 
$({\mathcal DIL}_{(1)}(M,g),\,d_{\,g}^{\,W})$ connecting $d\omega_1$ and $d\omega_2$ is provided by
the push--forward map,
\begin{eqnarray}
\label{conngeod}
\Gamma\,:\,[0,1]&\longrightarrow &\;\;\;\;\; {\rm Prob}_{ac}(M,g)\\
\lambda\;\;\;\;&\longmapsto &\Gamma(\lambda)\,:=\, \left(\exp_x(-\lambda\nabla \psi(x)) \right)_\sharp \,d\omega_1\;.\nonumber
\end{eqnarray}
 \end{lem} 
\begin{proof} This is the transcription in our setting of a well known basic result in optimal transport theory (cf.  \cite{figalliV} for a very readable short introduction). For the convenience of the reader and for later use we give a brief outline of the proof.  Let $\exp^{(g)}_x$  denote the exponential map on $(M,g)$ based at $x\in M$, then, as originally proven by R. McCann \cite{mccan}, there exists  a function $\psi :\,M\rightarrow \mathbb{R}$, the Kantorovich potential, such that, for $d\omega_1$--almost all points $x\in M$, we can define the map \cite{cordero},  \cite{mccan}, \cite{sturm},
\begin{eqnarray}
\aleph _\lambda\,:\,[0,1]\times M\,&\longrightarrow& \, M \\
(\lambda, x)\,\;\;\;\; &\longmapsto &\, \aleph_\lambda(x)\,:=\, \exp^{(g)}_x\left(\,-\,\lambda\,\nabla\,\psi(x)  \right)\;,\nonumber
\end{eqnarray}
 with $\left. \aleph_\lambda \right|_{\lambda=1}\,=\,\aleph _{opt}$ providing the optimal transport map. The associated plan $d\sigma_{opt}\,\in\, {\rm Prob}_{\,\omega_1,\,\omega_2}(M\times M)$  is given by 
\begin{equation}
d\sigma_{opt}\,=\,\left(Id_M,\,\aleph _1 \right)_\sharp\,d\omega_1\;,
\end{equation}
where  $\left(Id_M,\,\aleph _1 \right):\,M\times M\rightarrow M\times M$ is defined by $(y,x)\mapsto (y, \aleph _1(x) )$.
The push--forward of $d\omega_1$ under $\aleph_t$,
\begin{eqnarray}
\label{Pgeodesic}
\Gamma\,:\,[0,1]&\longrightarrow &\;\;\;\;\; {\rm Prob}_{ac}(M,g)\\
\lambda\;\;\;\;&\longmapsto &\Gamma(\lambda)\,:=\, \left(\aleph_\lambda \right)_\sharp \,d\omega_1\;,\nonumber
\end{eqnarray}
generates the (unique) geodesic in $\left({\rm Prob}(M),\; d_{\,g}^{\,W}\, \right)$  connecting $d\omega_1=\Gamma(0)$ and $d\omega_2=\Gamma(1)$. Moreover, under the stated hypotheses, $\Gamma(\lambda)\,<<\,d\mu_g$,\, $\lambda\in\,[0,1]$, thus the geodesic lies in $({\mathcal DIL}_{(1)}(M,g),\,d_{\,g}^{\,W})$. Note that we have  the identity
\begin{equation}
\label{innprod}
\left(d^W_g(d\omega_1,\,d\omega_2)\right)^2\,=\,\int_M\,\left|\nabla \psi   \right|_g^2\,d\omega_1\;,
\end{equation} 
where $\left|\nabla \psi   \right|_g^2\,:=\,g^{ab}\,\nabla _a\psi\,\nabla _b\psi$. \end{proof}
In other words, geodesics in $({\mathcal DIL}_{(1)}(M,g),\,d_{\,g}^{\,W})$  are naturally associated with the transport of the corresponding dilatonic measures along the geodesics of the underlying Riemannian manifold $(M,g)$. We remark that the optimal transport map $\aleph _{opt}$ associated with the optimal transport plan $d\sigma_{opt}$, is not generally smooth, (smoothness in optimal transport theory is related to rather sophisticated curvature assumptions \cite{figalliV}).\\ 
\\
\noindent
According to (\ref{innprod}), the Wasserstein metric structure on ${\mathcal DIL}_{(1)}(M,g)\approx {\rm Prob}_{ac}(M,g)$ is induced by the Otto inner product $\langle\cdot ,\cdot \rangle_{(g,\,d\omega)}$ defined by (\ref{inner}). This basic remark allows the introduction of a (weak) Riemannian calculus on the \emph{smooth} Wasserstein space $({\rm Prob}_{ac}(M,g),\,d_{\,g}^{\,W})$,  which is related to the weak Riemannian geometry of the diffeomorphism group $\mathcal{D}iff(M)$ discussed in section \ref{dgmms}. In particular one can compute \cite{lott1b}  the Riemannian curvature of $({\rm Prob}_{ac}(M,g),\,d_{\,g}^{\,W})$. We stress once more, (cf. Remark \ref{heuremotto}), that the Riemannian structure and the attendant geometric analysis is quite more delicate in the general Wasserstein space $({\rm Prob}(M),\,d_{\,g}^{\,W})$, where ${\rm Prob}(M)$ has no smooth structure and hence one lacks any obvious notion of smoothness for vector fields.   The nature of the difficulties and the development of a suitable weak Riemannian calculus on $(Prob(M), d_{\,g}^{\,W})$ are presented, from different point of views, in \cite{savare}, \cite{ambrosio},  \cite{nicolaG}, \cite{lott2}, \cite{sturm}, (\cite{nicolaG} provides a thorough analysis with examples and useful interpretative remarks). Since we confine to the \emph{smooth} Wasserstein space $({\rm Prob}_{ac}(M,g),\,d_{\,g}^{\,W})$,  the Riemannian interpretation can be implemented rigorously, and in what follows we shall make use of the following result proved by J. Lott, which we rephrase for the dilaton space ${\mathcal DIL}_{(1)}(M,g)$, (for the last time we recall that we use ${\mathcal DIL}_{(1)}(M,g)$ and $\rm{Prob}_{ac}(M,g)$ interchangeably). 
\begin{thm} (J. Lott \cite{lott1b}, Lemma 3)
\label{lotth1}
 Let $\vee _{\psi_{(i)}}:\,f\,\longmapsto  T_{f}{\mathcal DIL}_{(1)}(M,g)$ be vector fields on ${\mathcal DIL}_{(1)}(M,g)$ associated to potentials $\psi_{(i)}\in C^{\infty }(M ,\mathbb{R})/\mathbb{R}$ under the isomorphism (\ref{veciso}). The Riemannian connection $\overline{\bigtriangledown }$ of the Otto Riemannian metric $\langle\cdot ,\cdot \rangle_{(g,\,d\omega)}$ is  given by
\begin{equation}
\label{conn}
\langle \overline{\bigtriangledown }_{\vee _{\psi_{(i)}}}\,\vee _{\psi_{(j)}},\,\vee _{\psi_{(k)}} \rangle_{(g,\,d\omega)}\,=\,\int_M\,\nabla_a\psi_{(i)}\,\nabla^a \nabla^b\psi_{(j)}\,\nabla _b\psi_{(k)}\,d\omega\;.  
\end{equation}
 \end{thm}
 \noindent In particular, we can define the Lie derivative of the inner product $\langle\cdot,\,\cdot    \rangle_{(g,\,d\omega)}$ along the vector field  $f\longmapsto \vee _{\psi}\in T_{f}{\mathcal DIL}_{(1)}(M,g)$ according to
 \begin{eqnarray}
 \label{Lieder}
 &&\overline{\mathcal{L}}_{\vee _{\psi}}\,\langle \vee _{\psi_{(i)}},\,\vee _{\psi_{(k)}} \rangle_{(g,\,d\omega)}\\
\nonumber\\
&&:=\, \langle \overline{\bigtriangledown }_{\vee _{\psi_{(i)}}}\,\vee _{\psi},\,\vee _{\psi_{(k)}} \rangle_{(g,\,d\omega)}\,+\,\langle \,\vee _{\psi_{(i)}},\,\overline{\bigtriangledown }_{\vee _{\psi_{(k)}}}\vee _{\psi} \rangle_{(g,\,d\omega)}\nonumber\\
\nonumber\\
&&=\,
2\,\int_M\,\nabla^a\psi_{(i)}\,Hess_{ab}\,\psi\,\nabla ^b\psi_{(k)}\,d\omega\;,
\nonumber
 \end{eqnarray}
 where $Hess\,\psi$ denotes the Hessian on $(M,g)$.
 Note that according to (\ref{conn}) and (\ref{Helmholtz} ) the connection $\overline{\bigtriangledown }$ at $d\omega$ is characterized by
the gradient part of the $(M,g,d\omega)$ Helmholtz decomposition of the vector field $\nabla_a\psi_{(i)}\,\nabla^a \nabla^b\psi_{(j)}\,\partial _b$, \emph{i.e.} we can write 
\begin{equation}
\label{conn2}
\langle \overline{\bigtriangledown }_{\vee _{\psi_{(i)}}}\,\vee _{\psi_{(j)}},\,\vee _{\psi_{(k)}} \rangle_{(g,\,d\omega)}\,=\,\int_M\,\Pi _{\omega}\left(\nabla_a\psi_{(i)}\,\nabla^a \nabla^b\psi_{(j)}\right)\,\nabla _b\psi_{(k)}\,d\omega\;,  
\end{equation}
where the $L^2(M, d\omega)$ projection operator $\Pi_{\omega}$ is defined by (\ref{proj2}). Since the vector field $\nabla\psi_{(i)}\,\cdot Hess\,\psi_{(j)}:=\,\nabla_a\psi_{(i)}\,\nabla^a \nabla^b\psi_{(j)}\,\partial _b$ does not, in general, belong to $T_{\omega}\rm{Prob}_{ac}(M,g)\simeq T_{f}{\mathcal DIL}_{(1)}(M,g)$, we can interpret it as an element of $T_e\mathcal{D}iff(M)$, and for $\eta\in \mathcal{D}iff^{\,s}(M)$,\,$s>n/2\,+1$,  consider
\begin{equation}
\nabla\psi_{(i)}\,\cdot Hess\,\psi_{(j)}\circ \eta\,\in T_\eta{D}iff^{\,s}(M)\,
\end{equation}
as a section of the pull-back bundle $\eta^{-1}TM$.\; A similar remark applies also to the commutator  of two vector fields $\vee _{\psi_{(i)}}$ and $\vee _{\psi_{(k)}}$ $\in T_{\omega}\rm{Prob}_{ac}(M,g)$, defined for any $\varphi\in C_0^{\infty}(M,\mathbb{R})$ by \cite{lott1b}
\begin{eqnarray}
&&\langle \left[\vee _{\psi_{(i)}},\,\vee _{\psi_{(k)}} \right],\,\vee _{\varphi} \rangle_{(g,\,d\omega)}\,=\,
\langle \overline{\bigtriangledown }_{\vee _{\psi_{(i)}}}\vee _{\psi_{(k)}}-
\overline{\bigtriangledown }_{\vee _{\psi_{(k)}}}\vee _{\psi_{(i)}},\,\vee _{\varphi} \rangle_{(g,\,d\omega)}\nonumber\\
\nonumber\\
&&:=\,\int_M\,\left(\nabla_a\psi_{(i)}\,\nabla^a \nabla^b\psi_{(k)}-\nabla_a\psi_{(k)}\,\nabla^a \nabla^b\psi_{(i)}\right)\,\nabla _b\varphi\,d\omega\;,
\end{eqnarray}
and we have $[\vee _{\psi_{(i)}},\,\vee _{\psi_{(k)}}]\in T_e{D}iff(M)$. As already recalled we are in the typical situation of the Riemannian submersion described in section \ref{dgmms}.  To put this at work, one may proceed as in the classical O'Neill's analysis of Riemannian submersion \cite{oneill}, \cite{solovev}, (see also Chap. 9 of \cite{ABesse}), and introduce the vector field, (actually a $(2,1)$ tensor field) 
\begin{eqnarray}
\label{Tvector}
&&d\omega\longmapsto T_{\psi_{(i)}\psi_{(j)}}\,:=\,\frac{1}{2}(I-\Pi_{\omega})\left[\vee _{\psi_{(i)}},\,\vee _{\psi_{(j)}} \right]\,\in T_e\mathcal{D}iff_\omega(M)\nonumber\\
\nonumber\\
&&=\,(I-\Pi_{\omega})\left(\nabla_a\psi_{(i)}\,\nabla^a \nabla^b\psi_{(j)}\,\partial _b  \right)\,
\in\,Ker\,{\nabla^*}^\omega 
\;,
\end{eqnarray} 
defined by the $div_{\omega}$--free part of $\nabla_a\psi_{(i)}\,\nabla^a \nabla^b\psi_{(j)}\,\partial _b$. Note that in the last line of we have exploited the antisymmetry of $(I-\Pi_{\omega})(\nabla_a\psi_{(i)}\,\nabla^a \nabla^b\psi_{(j)}\,\partial _b )$, (cf. Lemma 6 of \cite{lott1b}).
The vector field (\ref{Tvector}) can be thought of as describing to what extent the distribution of gradient vector fields $\vee _{\psi_{(i)}}$ fails to be integrable in $\mathcal{D}iff(M)$.  With these remarks along the way one can    
characterize the Riemannian curvature of the Wasserstein space ${\mathcal DIL}_{(1)}(M,g)(M,g)$ according to the
\begin{thm} (J. Lott \cite{lott1b}, Th.1)
\label{LottCurv}
 Let $\vee _{\psi_{(i)}}\in T_{f}{\mathcal DIL}_{(1)}(M,g)$ be given vector fields on ${\mathcal DIL}_{(1)}(M,g)$ associated to corresponding potentials $\psi_{(i)}\in C^{\infty }(M ,\mathbb{R})/\mathbb{R}$. The Riemannian curvature of the connection $\overline{\bigtriangledown }$ is  given, at $d\omega$, by\\
\begin{eqnarray}
\label{curvProb}
&&\left\langle \overline{R}\left(\vee _{\psi_{(i)}},\,\vee _{\psi_{(j)}}\right)\,\vee _{\psi_{(k)}},\,
\vee _{\psi_{(h)}} \right\rangle_{(g,\,d\omega)}\\
\nonumber\\
&&\,=\,\int_M\,
\nabla^b\psi_{(i)}\,\nabla^c\psi_{(j)}\,\nabla^d\psi_{(k)}\,R^a_{bcd}\,\nabla _a\psi_{(h)}d\omega\nonumber\\
\nonumber\\
&&-\int_M g_{ab}\left(2T_{\psi_{(i)}\psi_{(j)}}^aT_{\psi_{(k)}\psi_{(h)}}^b-
T_{\psi_{(j)}\psi_{(k)}}^aT_{\psi_{(i)}\psi_{(h)}}^b+
T_{\psi_{(i)}\psi_{(k)}}^aT_{\psi_{(j)}\psi_{(h)}}^b
\right)d\omega.  \nonumber
\end{eqnarray}
\end{thm}
\begin{rem}
\label{remnonsm}
According to (\ref{curbreve}) we can interpret the $d\omega$-average of the Riemann tensor $R^a_{bcd}$ 
in  (\ref{curvProb}) as the curvature tensor  of $\left(\mathcal{D}iff(M),\,\langle\cdot,\,\cdot\rangle_{\omega,\eta}\right)$,
\begin{equation}
\left\langle \breve{R}\left(\nabla\psi_{(i)},\,\nabla\psi_{(j)}\right)\,\nabla\psi_{(k)},\,\nabla\psi_{(h)}\right\rangle_{(\omega,e)}\;,
\end{equation}
evaluated at the identity map. Indeed, quite independently from optimal transport techniques, one can derive (\ref{curvProb}) as the formula computing the  \emph{horizontal} components of the Riemannian curvature associated with the Riemannian  submersion 
\begin{equation}
\mathfrak{J}_{\omega}:\,\mathcal{D}iff(M )\,\longrightarrow\, {\rm Prob}_{ac}(M,g)
\simeq \mathcal{D}iff(M)/\mathcal{D}iff_\omega(M)
\end{equation}
defined by  (\ref{Jproject}), (cf. \cite{smolentsev1}). We emphasize once more that the characterization of curvature for the Wasserstein space $(Prob(M), d_{\,g}^{\,W})$, associated to general probability measures on $M$, is quite more delicate. Details can be found in  \cite{nicolaG}. 
\end{rem}

\section{Constant maps localization and warping}
\label{lacmadw}

Important insight into the structure of the Wasserstein geometry of the dilatonic non--linear $\sigma$ model  can be gained by the following analysis, showing how the theory behaves under a localization around constant maps.\\
\\
\noindent To begin with,  
let $\kappa $ denote the  upper bound to the sectional curvature of $(M,g,\,d\omega)$, and let us consider a metric ball $B(r,\,p):=\{z\in M|\,d_g(p,z)\leq r\}$, centered at  $p\,\in\, M$, with radius $r\,<\,r_0$, where
$r_0$, defined by (\ref{errezero1}) sets the length scale of the target $(M,g)$.
For $q\,\in\,\mathbb{N}$,  let $\{\phi_{(k)}\}_{k=1\,\ldots,\,q}\,\in\, Map\,(\Sigma, M)$ denote  a  collection of \emph{reference} constant maps, (hence harmonic), taking values in the interior of $B(r,\,p)$
\begin{eqnarray}
\phi_{(k)}\,:\,\Sigma &\longrightarrow& \;\;\;\;B(r,\,p)\setminus \partial B(r,\,p) \subset M\\
x\,&\longmapsto& \,\phi_{(k)}(x)\,=\, y_{(k)},\;\;\;\;\;\;\;\;\forall x\,\in\,\Sigma,
\;\;\;\;\;\;k=1,\,\ldots,\,q\;. \nonumber
\end{eqnarray} 
We explicitly assume that $r<\frac{\pi }{6\,\sqrt{\kappa }}$,\, $inj\;(y)>3r$ for all $y\in B(r,p)$,   and consider  the  \emph{center of mass} \cite{buser}   of the maps  $\{\phi_{(k)}\}$,
\begin{equation}
\phi_{cm}\doteq cm\, \left\{\phi_{(1)},\ldots,\phi_{(q)} \right\}\;, 
\label{cm}
\end{equation}
 characterized as the minimizer of the function 
\begin{equation} 
\label{Fcm}
F(y;\,q)\doteq \frac{1}{2}\,\sum_{k=1}^{q}\,d_{g}^{2}(y,y_{(k)})\;,
\end{equation}
where $d_{g}^{2}(\cdot \, ,\cdot \, )$ denotes the distance in $(M,g)$.
\begin{lem}
\label{CMlemma}
Under the stated hypotheses the minimizer exists, is unique and $cm\, \left\{\phi_{(1)},\ldots,\phi_{(q)} \right\}\,\in B(2r,p)$.
\end{lem}
\begin{proof}
This is a standard property of the center of mass \cite{buser}. See for instance \cite{Glickenstein} (Chap.4, p.175).
\end{proof}
We denote by $\{d_{g}(\phi_{cm},\phi_{(k)})\}$ the distances between the maps $\{\phi_{(k)}\}$  and their center of mass $\phi_{cm}$. The overall strategy for introducing the constant maps $\{\phi_{(k)}\}$  is to use the distances $\{d_{g}(\phi_{cm},\phi_{(k)})\}$ and the dilatonic measure $d\omega$  to set the scale at which $(\Sigma,\gamma)$ probes the geometry of $(M,g)$. To this end,  we localize   $\phi\in \rm{Map}(\Sigma, M)$, around the center of mass of $\{\phi_{(k)}\}_{k=1}^q$, by choosing  $d\omega$ according to
\begin{equation}
\label{kprobmeas}
d\omega(z;q)\,:=\,C_{\,r}^{-1}(q)\,e^{\,-\,\tfrac{F(z;\,q)}{2\,r^2}}\,
\frac{d\mu_g(z)}{V_g(M)}\;,\;\;\;\;\;\;\;\;\;\;\;\;\; z\in M\;,
\end{equation}
where  $F(z;\,q)$ is the center of mass function (\ref{Fcm}),\, $V_g(M)$ is the Riemannian volume of $(M,g)$, and $C_{\,r}(q)$ denotes a  normalization constant  such that $\int_M\,d\omega(z;q)\,=\,1$. Since $F(z;\,q)$ attains its minimum at $\phi_{cm}$, the measure $d\omega$ is concentrated around the center of mass of the $\{\phi_{(k)}\}$'s, and, as $r\searrow 0^+$, \, weakly converges to the Dirac measure $\delta _{cm}$ supported at $\phi_{cm}$. Note that according to (\ref{Fcm}) we can factorize the density $d\omega/d\mu_g(z)$ as
\begin{equation}
\label{measFact}
\frac{d\omega(z;q)}{d\mu_g(z)}\,=\,V_g(M)^{-1}\,\prod_{k=1}^q\,\,e^{\,-\,\tfrac{d^2_g(z,\,\phi{(k)})}{4\,r^2}\,-\,\frac{\ln C_{\,r}(q)}{q}}\;.
\end{equation}
This latter remark suggests to interpret the distances $d_{g}(\phi_{cm},\phi_{(k)})$,\,$k=1,\ldots\,q$ as coordinates $\{\xi_{(k)}\}$ in a  $q$--dimensional flat torus $\mathbb{T}^q_{cm}$ of unit volume, and 
consider the product manifold
\begin{equation} 
N^{n+q}\,:=\, M\times_{(\omega)} \mathbb{T}^q\;,
\end{equation}
endowed with the warped product measure
\begin{equation}
\label{measProd}
d\mu_{N}(z,\xi )\,:=\,d\mu_g(z)\,\prod_{k=1}^q\,\,e^{\,-\,\tfrac{d^2_g(z,\,\phi{(k)})}{4\,r^2}\,-\,\frac{\ln C_{\,r}(q)}{q}}\,d\xi_{(k)}\;,\;\;\;\;\;\; \{\xi_{(i)}\}\in\,\mathbb{T}^q
\;,
\end{equation}
with total volume
\begin{eqnarray}
\\
&&\int_{N^{n+q}}\,d\mu_{N}(z,\xi )\,=\,\int_{N^{n+q}}\,d\mu_g(z)\,\prod_{k=1}^q\,\,e^{\,-\,\tfrac{d^2_g(z,\,\phi{(k)})}{4\,r^2}\,-\,\frac{\ln C_{\,r}(q)}{q}}\,d\xi_{(k)}\nonumber\\
\nonumber\\
&=&\,C^{-1}_{\,r}(q)\,\int_M\,e^{\,-\,\tfrac{F(z;\,q)}{2\,r^2}}\,d\mu_g(z)\,\int_{\mathbb{T}^q}\,\prod_{k=1}^q\,d\xi_{(k)}\nonumber\\
\nonumber\\
&=&\,
V_g(M)\,\int_{M}\,d\omega(z;q)\,=\,V_g(M)\;.\nonumber
\end{eqnarray}
The  measure (\ref{measProd}) is clearly Riemannian, being induced on $N^{n+q}\,:=\, M\times \mathbb{T}^q$ by the
warped metric
\\ 
\begin{equation}
h^{(q)}(y,\,\xi )\,:=\, g(y)\,+\,\left(\frac{d\omega(y;q)}{d\mu_g(y)}\,V_g(M)\right)^{\tfrac{2}{q}}\,\sum_{i=1}^q\,d\xi_{(i)}^2,\;\;\;\;\;\; \xi_{(i)}\in\,[0,1]\;,
\;,
\label{new012}
\end{equation}
according to
\begin{eqnarray}
\label{tqmeas}
\\
d\mu_{h^{(q)}}(y,\xi )\,&=&\,d\mu_g(y)\,\left(\frac{d\omega(y;q)}{d\mu_g(y)}\,V_g(M)\right)\,\prod_{i=1}^q\,d\xi_{(i)}\nonumber\\
\nonumber\\
&=&\,d\mu_g(y)\,\prod_{k=1}^q\,\,e^{\,-\,\tfrac{d^2_g(y,\,\phi{(k)})}{4\,r^2}\,
-\,\frac{\ln C_{\,r}(q)}{q}}\,d\xi_{(k)}\,=\,d\mu_{N}(y,\xi )\;.\nonumber
\end{eqnarray}
\begin{rem}
Trading the Riemannian metric measure space $(M,g,\,d\omega)$ with the warped Riemannian manifold 
\begin{equation}
\left(N^{n+q}\,:=\, M\times_{(\omega)}  \mathbb{T}^q,\,h^{(q)}  \right)\;,
\end{equation}
is a  standard procedure in the Riemannian measure space setting, (cf.  \cite{gursky}, and \cite{lott3} for the application to Perelman's reduced volume). This construction can be naturally extended to the injection of 
$(M,g,\,d\omega)$ in the $\infty $--dimensional fibering $ \left(M\times \mathbb{T}^{\infty },\,h^{(\infty )}\right)$. 
\end{rem}
\noindent  
As a function of the distances from the constant maps $\{\phi_{(k)}\}$, the probability measure $d\omega\in \rm{Prob}_{ac}(M,g)\approx {\mathcal DIL}_{(1)}(M,g)$ defined by (\ref{kprobmeas}) is Lipschitz on $(M,g)$ and smooth on $M\,\backslash \,\cup_{k=1}^q\left\{\phi_{(k)},\,\rm{Cut}(\phi_{(k)}) \right\}$, where $\rm{Cut}(\phi_{(k)})$ denotes the cut locus of each $\phi_{(k)}$. In terms of the associated dilaton field we have
\begin{lem}
\label{finW}
The dilaton field associated to the measure $d\omega$,
\begin{eqnarray}
\label{fdil}
f(z;q)\,&:=&\,-\,\ln\,\left(\frac{d\omega(z;q)}{d\mu_g(z)}\,V_g(M)\right)\\
\nonumber\\
&=&\,\frac{1}{4\,r^2}\,\sum_{k=1}^q\,d_g^2\left(z,\,\phi_{(k)}\right)\,+\,\ln C_{\,r}(q)\;,\nonumber
\end{eqnarray}
has a gradient $\nabla f$ which exists a.e. on $(M,g)$ and $f\in \mathcal{H}^1(M,\mathbb{R})$.  
\end{lem}
\begin{proof}
As a sum of squared distances from the constant maps $\{\phi_{(k)}\}$, $f$ is smooth 
on $M\,\backslash \,\cup_{k=1}^q\left\{\phi_{(k)},\,\rm{Cut}(\phi_{(k)}) \right\}$, and Lipschitz on $(M,g)$. Thus 
the first part of the lemma is a direct consequence of Rademacher's theorem. In particular the distributional gradient of $f$ with respect to the Riemannian metric,  $\nabla f$ is an $L^\infty$ vector field on $(M,g)$, that is $f\in\,W^{1, \infty}(M,\mathbb{R})$ where the Sobolev norm $\parallel f\parallel _{W^{1,\infty }}$ is defined by
$\parallel f\parallel _{W^{1,\infty }}:=\, ess\,\sup_M\,(|f|+|\nabla f|)$, (cf. Chap.4, Th. 5 of \cite{gariepy}).  Since $M$ is compact, it follows \cite{grigoryan2}, (Chap. 11, Corollary 11.4), that $f\in L^2(M,\mathbb{R})$ and $\nabla f \in L^2 (M,TM)$, \emph{i.e.}, $f\in \mathcal{H}^{1}(M,\mathbb{R})$. Let us observe that $\nabla f$ can  be integrated by parts over $(M,g)$ against locally Lipschitz vector field $v$, in the sense that $\int_M\,f\,\nabla _i\,v^i\,d\mu_g\,=\,\int_M\,v^i\nabla _i f d\mu_g$, (cf. Chap.9, lemma. 7.113 of \cite{Glickenstein} for a proof of this well--known property of Lipschitz functions).  Better global regularity cannot be expected since, if we compute the Laplacian of $f$, we get
\begin{equation}
\Delta _g f(z;q)\,=\,\frac{1}{4\,r^2}\,\sum_{k=1}^q\,\Delta _g\,d_g^2\left(z,\,\phi_{(k)}\right)\;,
\label{}
\end{equation}
which, from the standard properties of the distance function, implies that $\Delta _g\, f(z;q)$ is a distribution with a singular part supported on the cut locus $\cup_{k=1}^q\,\rm{Cut}(\phi_{(k)})$ 
of $(M,g,\,\{\phi_{(k)}\})$.       
\end{proof}
If we localize $f$ to the geodesic ball $B(2r,p)$ the situation is simpler, and we have 
\begin{lem}
For  $z\,\in\,B(2r,p)$,\;$r<\frac{\pi }{6\,\sqrt{\kappa }}$, \; the dilaton field $f(z;q)$ is a strictly convex function. In particular, if $[0,1]\,\ni \,s\mapsto \lambda^{(k)} (s)$,\;$k=1,\ldots,q$, are minimal geodesics connecting $\lambda^{(k)}(0)\,=\,z$ to $\lambda^{(k)}(1)\,=\,\phi_{(k)}$, and  $\exp_{z}:T_z\,M\longrightarrow M$ denotes the exponential mapping based at $z$, then we can write
\begin{eqnarray}
f(z;q)\,&=&\,\frac{1}{4\,r^2}\,\sum_{k=1}^q\,\int_0^1\,g(\dot{\lambda}^{(k)},\dot{\lambda}^{(k)})\,ds\,
-\,{\ln C_{\,r}(q)}\\
\nonumber\\
&=&\,\frac{1}{4\,r^2}\,\sum_{k=1}^q\,\left|\exp_{{z}}^{-1}\,\phi_{(k)}\right|^2\,-\,{\ln C_{\,r}(q)}\;,\nonumber
\end{eqnarray}
where $\exp_{{z}}^{-1}\,\phi_{(k)}\,\in\,T_zM$, \; $\sum_{k=1}^q\,\exp_{\phi_{cm}}^{-1}\,\phi_{(k)}\,=\,\vec {0}$, and
\\
\begin{equation}
 \left|\exp_{{z}}^{-1}\,\phi_{(k)}\right|\,:=\,\left[\left(\exp_{{z}}^{-1}\,\phi_{(k)}\right)^a\,
\left(\exp_{{z}}^{-1}\,\phi_{(k)}\right)^b\,\delta _{ab} \right]^{1/2}\,=\,d_g\left(\phi_{(k)},\,z\right)\;.
\end{equation}
In particular,  away from the cut locus $\rm{Cut}(\phi_{(k)})$ of each $\phi_{(k)}$, the dilaton field $f$ is differentiable on $M\,\backslash \,\cup_{k=1}^q\left\{\phi_{(k)},\,\rm{Cut}(\phi_{(k)}) \right\}$, and one computes
\\
\begin{equation}
\frac{\partial }{\partial z^h}\,f(z;q)\,=\,-\,\frac{1}{2\,r^2}\,\sum_{k=1}^q\,\left.\dot{\lambda}_h^{(k)}(s)\right|_{s=1}\,=
\,-\,\frac{1}{2\,r^2}\,g_{ih}(z)\,\sum_{k=1}^q\,\left(\exp_{z}^{-1}\,\phi_{(k)}\right)^i\;.
\end{equation}
\end{lem}
\begin{proof} These results directly follow from the standard properties of the distance function over a convex metric ball, (in our case $B(2r,p)$).
\end{proof}
\begin{rem} 
In terms of $f$ the
warped metric (\ref{new012})  takes the form
\begin{equation}
h^{(q)}(y,\,\xi )\,:=\, g(y)\,+\,e^{-\,\tfrac{2f(y;q)}{q}}\,\sum_{i=1}^q\,d\xi_{(i)}^2,\;\;\;\;\;\; \xi_{(i)}\in\,[0,1]
\;,
\label{newmetr0}
\end{equation}
with
\begin{equation}
d\mu_{h^{(q)}}(y)\,=\,e^{-\,f(y;q)}\,d\mu_g(y)\;.
\label{newdmug}
\end{equation}
\end{rem}
The metric warping (\ref{newmetr0}) on the torus fiber $\mathbb{T}^q_{y}$ over $y\in M$ can be compensated by the point--dependent ($y$) rescaling  $\xi_{(k)}\,\mapsto\,\,\xi_{(k)}\, \exp[{\tfrac{f(y;q)}{q}}]$ of the fiber itself. This suggests a natural extension of $\phi\in \rm{Map}(\Sigma, M)$ associated to the warping $M\times_{(f)} \mathbb{T}^q$ generated by the dilatonic measure.
\begin{lem}
Let $\phi\in {\mathcal{H}^{1}}(\Sigma, M)$ denote a  localizable map and let $f\circ \phi$ be the induced dilaton field over $\phi(\Sigma)$. Then, the function $f\circ \phi$ is of class ${\mathcal{H}^{1}}(\Sigma, \mathbb{R})$ and 
\begin{eqnarray}
\label{newMap0}
\\
&&\Phi_{(q)}\,:\,(\Sigma,\gamma)\,\longrightarrow \;\;\;\;\;\;\;\;\;\;\;(N^{n+q}:=M^n\times_{(f)} \mathbb{T}^q,\,h^{(q)})\nonumber\\
\nonumber\\
&& x\;\; \longmapsto  \Phi_{(q)}(x):=\left(\phi^i(x),\,\tfrac{1}{2}\,e^{\tfrac{f(\phi(x);q)}{q}}\,
d_{g}(\phi_{cm},\phi_{(1)})\,,\ldots,\,\right.\nonumber\\
\nonumber\\
&&\left.\ldots,\tfrac{1}{2}\,e^{\tfrac{f(\phi(x);q)}{q}}\,d_{g}(\phi_{cm},\phi_{(q)}) \right)\;  \nonumber
\end{eqnarray}
is  a localizable map $\in\,{\mathcal{H}^{1}}(\Sigma, M\times_{(f)} \mathbb{T}^q)$ describing the (fiber--wise uniform) dilatation of the torus fiber $\mathbb{T}^q_{\phi(x)}$ over $\phi(x)\in M$.
\end{lem}
\begin{proof}
According to Lemma \ref{finW}, the function $f$ is Lipschitz on $(M,g)$. In general, the composition between a map $\phi\in {\mathcal{H}^{1}}(\Sigma, M)$  and a  Lipschitz function $f:M\rightarrow \mathbb{R}$ does not map naturally in ${\mathcal{H}^{1}}(\Sigma, \mathbb{R})$ since the distributional gradient $\nabla \phi (x)$ does not necessarily vanish almost everywhere for $\phi(x)\in \cup_{k=1}^q\{\rm{Cut}(\phi_{(k)})\}$, and a chain rule for differentiating weakly $f\circ \phi$ may be not available. However, if the map $\phi\in {\mathcal{H}^{1}}(\Sigma, M)$ is of finite harmonic energy $E[\phi,\,g]_{(\Sigma, M)}$ then the function $d_g(\phi(x),y)$ is, for any point $y\in M$ of class ${\mathcal{H}^{1}}(\Sigma, \mathbb{R})$. This is a particular case of a more general result concerning harmonic maps $L^2(X,Y)$ between a Riemannian polyhedra, with simplexwise smooth Riemannian metric $(X,\gamma)$, and a metric space $(Y,d_Y)$, (cf. Eells--Fuglede \cite{eells} th. 9.1 pp. 153-154; for a general metric space set--up in our case see also Section \ref{digression} below).  Hence we have $f\circ \phi\in \mathcal{H}^{1}(\Sigma,\mathbb{R})$. Moreover, if $\phi\in {\mathcal{H}^{1}}(\Sigma, M)$ is localizable, we can introduce local coordinates
$\{D_{(\alpha)},\,x^{\alpha }\}$ in $(\Sigma,\gamma)$,\, and
 $y^k=\phi^{k}(x)$, \, $k=1,\ldots,n$,  for the corresponding image points in $\phi(D_{(\alpha)})\subset M$, and  work locally in the 
smooth framework provided by  the  space of  smooth maps   $Map\,(\Sigma, M\times_{(f)} \mathbb{T}^q)$. 
\end{proof}
As an elementary consequence of the duality between the metric warping (\ref{newmetr0}) and the map warping (\ref{newMap0}) we get 
\begin{lem}
The harmonic energy functional associated with the map $\Phi_{(q)}$ is provided by 
\begin{eqnarray}
\label{qharmwarp}
E[\Phi_{(q)},\,h^{(q)}]_{(\Sigma, N^{n+q})}\,&=&\,E[\phi,\,g]_{(\Sigma, M)}\\
\nonumber\\
&+&\,\frac{F(\phi_{cm};q)}{2}\,
\mathfrak{D}[q^{-1}\,f(\phi;q)]_{(\Sigma,\,\mathbb{R})}\;,\nonumber
\end{eqnarray}
where $\mathfrak{D}[q^{-1}\,f(\phi;q)]_{(\Sigma,\,\mathbb{R})}$ is the  Dirichlet energy 
\begin{equation}
\label{harmdiren}
\mathfrak{D}[q^{-1}\,f(\phi;q)]_{(\Sigma,\,\mathbb{R})}\,:=\, 
\,\frac{1}{2}\,
\int_{\Sigma }\,\left|\tfrac{d f(\phi(x);q)}{q}\right|^2_{\gamma}\,d\mu_{\gamma}(x)
\end{equation}
 associated to the 
map $q^{-1}\,f\circ \phi\,:\, (\Sigma,\gamma)\longrightarrow \mathbb{R}^1$, and $F(\phi_{cm};q)$ is the minimum of the center of mass function (\ref{Fcm}).
\end{lem}
\begin{proof} Since the map $\Phi_{(q)}$ is localizable and $f\circ \phi\in \mathcal{H}^{1}(\Sigma,\mathbb{R})$,  a direct computation in local charts  provides
\begin{eqnarray}
\label{EPhiq}
\\
&&E[\Phi_{(q)},\,h^{(q)}]_{(\Sigma, N^{n+q})}\,:=\,
\frac{1}{2}\,\int_\Sigma\,\gamma^{\mu\nu}\,\frac{\partial \Phi_{(q)}^{a}(x,\xi )}{\partial x^{\mu}}
\frac{\partial\Phi_{(q)}^{b}(x,\xi )}{\partial x^{\nu}}\,h_{ab}(\phi)\,d\mu_{\gamma}\nonumber\\
\nonumber\\
&&\,=\,\frac{1}{2}\,\int_\Sigma\,\gamma^{\mu\nu}\,\frac{\partial \phi^{i}}{\partial x^{\mu}}
\frac{\partial\phi^{j}}{\partial x^{\nu}}\,g_{ij}(\phi)\,d\mu_{\gamma}\,\nonumber\\
\nonumber\\
&&\,+\,\frac{1}{8}\,\sum_{k=1}^q\,d_{g}^2(\phi_{cm},\phi_{(k)})\,
\int_{\Sigma}\,\left|\tfrac{d f(\phi(x);q)}{q}\right|^2_{\gamma}\,d\mu_{\gamma}\nonumber\\
\nonumber\\
&&\,=\,E[\phi,\,g]_{(\Sigma, M)}\,+\,\,\frac{F(\phi_{cm};q)}{2}\,
\mathfrak{D}[q^{-1}\,f(\phi;q)]_{(\Sigma,\,\mathbb{R})}\;,\nonumber
\end{eqnarray}
\\
\noindent
where $a,b\,=\,1,\ldots,n+q$,\, $\left|d f(\phi;q)\right|^2_{\gamma}\,:=\,\gamma^{\mu\nu}\,\tfrac{\partial f(\phi;q)}{\partial x^{\mu}}
\tfrac{\partial f(\phi;q)}{\partial x^{\nu}}$, and where we  exploited the relation
\\
\begin{equation}
\gamma^{\mu\nu}\,\frac{\partial\,\Phi_{(q)}^{l}(x,\xi ) }{\partial x^{\mu}}
\frac{\partial\, \Phi_{(q)}^{m}(x,\xi )}{\partial x^{\nu}}\,\delta _{lm}\,=\,\frac{1}{4}\,
\left|\tfrac{d f(\phi;q)}{q}\right|^2_{\gamma}\,e^{\,\tfrac{2f(\phi;q)}{q}}\,\sum_{k=1}^q\,d_{g}^2(\phi_{cm},\phi_{(k)})\;,
\end{equation}
for $l,\,m\,=\,n+1,\,\ldots,\,q$.  
\end{proof}
This extended harmonic map set--up is interesting in many respects. In particular,
if we choose the given surface $(\Sigma, \gamma)$ to be topologically the 
$2$--torus $\mathbb{T}^2$ endowed with a conformally flat metric associated to the dilaton field, then the functional $E\,[\Phi_{(q)},\,h^{(q)}]$ can be directly connected  to the non--linear $\sigma$ model dilatonic action (\ref{eqnAct}). The underlying formal procedure is well--known, but in our setting the center of mass localization makes its role rather subtle and far reaching. 
\begin{prop} 
\label{propconf}
Let $(\Sigma\simeq \mathbb{T}^2,\,\delta)$ be a flat 2--torus. For $\phi\in {\mathcal{H}^{1}}(\mathbb{T}^2, M)$ a localizable map taking values in $M\,\backslash \,\cup_{k=1}^q\,\rm{Cut}(\phi_{(k)})$, we denote by $f\circ \phi$ the induced dilaton field over $\phi(\mathbb{T}^2)$. If we endow $\mathbb{T}^2$ with the conformally flat metric
\begin{equation}
\gamma_{\mu\nu}=e^{\tfrac{f(\phi;q)}{q}}\, \delta_{\mu\nu}\;, 
\label{psiconf00}
\end{equation}
then in the resulting conformal gauge $(\mathbb{T}^2,\,\gamma)$ we can write
\begin{equation}
\label{eqnAct100}
{S}\left[\gamma,\phi;\,\,F(\phi_{cm};q),\,g,\,d\omega\right]
\;\underset{(f)}{=}\; 
\frac{2\,q}{F(\phi_{cm};q)}\;
E[\Phi_{(q)},\,h^{(q)}]_{(\Sigma, N^{n+q})}\;,
\end{equation}
\end{prop}
\begin{proof}
Under the stated hypotheses, 
$\phi(\mathbb{T}^2)\,\subseteq \,M\,\backslash \,\cup_{k=1}^q\,\rm{Cut}(\phi_{(k)})$ so that $f$, restricted to $\phi(\mathbb{T}^2)$, is smooth. Since $\phi$ is  localizable,  we can use  local charts and compute (in the weak sense) the Gaussian curvature of $(\mathbb{T}^2, \gamma)$ according to $\mathcal{K}_{f}\,=\,-\,\frac{1}{2q}\,\Delta _{\gamma}\,f(\phi;q)$.  By integrating $\mathcal{K}_{f}\,f(\phi;q)$ over the surface $\Sigma\simeq \mathbb{T}^2$ we get 
\begin{eqnarray}
\label{Kdecrease00}
&&\int_{\Sigma}\,f(\phi;q)\,\mathcal{K}_{f}\,d\mu _{\gamma}\,\underset{(f)}{=}\,
-\,\frac{1}{2\,q}\,\int_{\Sigma}\,f(\phi;q)\,\Delta _{\gamma}\,f(\phi;q)\,d\mu _{\gamma}\\
\nonumber\\
&&= \,\frac{1}{2\,q}\,\int_{\Sigma}\,
 \left|d\,f(\phi;q)\right|_{\gamma}^2\,d\mu _{\gamma}\,\underset{(f)}{=}\,
q\,\mathfrak{D}[q^{-1}\,f(\phi;q)]_{(\Sigma,\, \mathbb{R})} \;,\nonumber
\end{eqnarray}
where we have integrated by parts, and where the underset ${(f)}$ stresses the fact that the relation holds in the dilaton induced conformal gauge (\ref{psiconf00}). Thus, (\ref{qharmwarp}) can be rewritten as  
\begin{eqnarray}
\label{Erelation00}
E[\Phi_{(q)},\,h^{(q)}]_{(\Sigma, N^{n+q})}\, &\underset{(f)}{=}&\,
E[\phi,\,g]_{(\Sigma,\, M)}\\
\nonumber\\
&+&
\,\frac{F(\phi_{cm};q)}{2\,q}\, \int_{\Sigma}\,f(\phi;q)\,\mathcal{K}_{f}\,
d\mu _{\gamma}\;,\nonumber
\end{eqnarray}
and if we identify the non--linear $\sigma$ model coupling constant $a$ with
\begin{equation}
\label{Facoupl}
a\,\equiv \,\frac{F(\phi_{cm};q)}{q}\;,
\end{equation}
then (\ref{eqnAct100}) follows from the definition (\ref{eqnAct}) of the dilatonic action.
\end{proof}
\begin{rem}
If in the statement of Prop. \ref{propconf}  we do not restrict $\phi(\mathbb{T}^2)$ to $M\,\backslash \,\cup_{k=1}^q\,\rm{Cut}(\phi_{(k)})$, then the distribution $\Delta _{\gamma}\,f(\phi;q)$ would acquire a singular part supported on the inverse image of the cut locus $\cup_{k=1}^q\,\rm{Cut}(\phi_{(k)})$ of $(M,g,\,\{\phi_{(k)}\})$. This singular part contributes to $\mathcal{K}_{f}$ with conical--metric singularities, as in the case of polyhedral surfaces. It is not difficult to extend the above results to this more general setting, by considering from the very outset polyhedral surfaces $(\Sigma\simeq \mathbb{T}^2)$ supporting a piecewise flat metric with conical singularities. We do not belabor on this point, since it does not add much to the analysis that follows.   
\end{rem}

\begin{rem}
Note that $q^{-1}\,{F(\phi_{cm};q)}$ is the average squared distance between the constant maps $\{\phi_{(k)}\}$  and their center of mass $\phi_{cm}$.  According to the assumptions leading to lemma \ref{CMlemma}, the center of mass is contained in a metric ball $B(p, 2r)$  of radius $2r$ with $r<\frac{\pi }{6\,\sqrt{\kappa }}$, where $\kappa$ is the upper bound to the sectional curvature of $(M,g)$. This imples that the average squared distance between the $\{\phi_{(k)}\}$'s  and their center of mass $\phi_{cm}$ is bounded by the (squared) diameter of $B(p, 2r)$. In particular,  under the identification (\ref{Facoupl}), the 
non--linear $\sigma$ model coupling constant $a$ satisfies 
\begin{equation}
a\,|\kappa|\,\leq\,\frac{4}{9}\pi^2\;,
\label{AK}
\end{equation}
which is weaker than  the bound $a\,|\kappa|<<1$ characterizing the point--like limit in which the renormalization group for non--linear $\sigma$ model yields the Ricci flow.  
\end{rem}

\subsection{Harmonic energy as a center of mass functional }
\label{digression}
The presence of the center of mass $q^{-1}\,{F(\phi_{cm};q)}$ as a dilatonic coupling  in the above analysis  may appear incidental to the particular set up we have concocted.  Actually, it is a manifestation  of the fact that the center of mass plays a basic role in characterizing harmonic map functionals in a general metric space setting \cite{eells}, \cite{jost}, \cite{korevaar}. It is important to make this role explicitly available also for non--linear $\sigma$ models, since it provides the rationale for the use of the heat kernel embedding in  Wasserstein space. \\
\\
\noindent  We start by characterizing  the metric space structure associated with the warped manifold
$\left(N^{n+q}:=M\times_{(f)}\mathbb{T}^q,\,h^{(q)}\right)$. If $\lambda:\,[0,1]\rightarrow N^{n+q}$,\, $s\longmapsto \lambda(s)=(y(s),\,\xi(s))$ is a curve in  $(N^{n+q},\,h^{(q)})$, and $\{s_j\}\,:=\,0=s_0<\ldots<s_m=1$ is a partition of $[0,1]$,  then the length of $\lambda$ in $(N^{n+q},\,h^{(q)})$ is defined by, (see \emph{e.g.} \cite{chien}),  
\begin{equation}
\label{distl}
L_N(\lambda)\,:=\,\lim_{\{s_j\}}\,\sum_{i=1}^m\,\left[d^2_g(y(s_{i-1}),\,y(s_{i}))\,+\,e^{-\,\tfrac{2\,f(y(s_{i-1}))}{q}}
\,d^2_{\mathbb{T}^q}(\xi(s_{i-1}),\,\xi(s_{i}))  \right]^{\tfrac{1}{2}}\;,
\end{equation}
where the limit is with respect to the refinement of ordering of the partitions $\{s_j\}$ of $[0,1]$, and
\begin{equation}
d^2_{\mathbb{T}^q}(\xi(s_{i-1}),\,\xi(s_{i}))\,:=\,\inf_{m_k\in\mathbb{Z}^q}\sum_{k=1}^{q}\,\left| \overline{\xi}_{(k)}(s_{i-1})\,-\,\overline{\xi}_{(k)}(s_{i})\,+\,m_k \right|^2
\end{equation}
is the squared distance in $\mathbb{T}^q$,\;  $\overline{\xi}_{(k)}\in\mathbb{R}^q$ being the representative components of ${\xi}\in \mathbb{T}^q$. If we denote by $\Gamma_{(Y,Z)}$ the set of all (piecewise) smooth curves connecting two points $Y=(y,\xi)$ and $Z=(z, \zeta )$ in $N^{n+q}$, then their distance
\begin{equation}
\label{disth}
d_h(Y,Z)\,:=\,\inf_{\lambda\in\Gamma_{(Y,Z)} }\,L_N(\lambda)\;,
\end{equation}
characterizes the metric space $(N^{n+q},\, d_h)$ associated with the warped Riemannian manifold $(N^{n+q},\,h^{(q)})$.\\
\\
\noindent With these preliminary remarks along the way, fix $x\in \Sigma$, and for $0<\epsilon <1$ small enough, let $D(x,\epsilon)$ be the  disk of radius $\epsilon$ in the tangent space $T_x\Sigma$, and denote by $\exp_x^{(\gamma)}\,:\,T_x\Sigma\longrightarrow (\Sigma, \gamma)$ the exponential mapping based at $x$. We assume that 
$\exp_x^{(\gamma)}(D(x,\epsilon))$ is convex and that $D(x,\epsilon)$ is endowed with the pull--back measure $d\widehat{\mu}_\gamma\,:=\,\left(\exp_x^{(\gamma)}\right)^{*}\,d\mu_\gamma $. Let us consider the map
\begin{eqnarray}
\label{xphi}
&&\widetilde{\Phi}_{(q)}\,:\,D(x,\epsilon)\subset T_x\Sigma\,\longrightarrow\, N^{n+q}\\
&&\;\;\;\;\;\;\upsilon \,\;\longmapsto \,\;\;\;\;\;\widetilde{\Phi}_{(q)}(\upsilon)\,:=\,\Phi_{(q)}\left(\exp_x^{(\gamma)}(\upsilon )\right)\;,\nonumber
\end{eqnarray}
 and define, for almost all  $x\in\Sigma$, the center of mass function  
\begin{eqnarray}
\label{HACM}
&&\Sigma\,\longrightarrow \,\mathbb{R}_{>0}\\
&&x\,\longmapsto \,G\left[\Phi_{(q)}(x),\,D(x,\epsilon)\right]\nonumber\\
\nonumber\\
&&=\,\,\frac{1}{\left|D_\delta(\epsilon)\right|}\,\int_{D(x,\epsilon)}\,d^2_h\left(\Phi_{(q)}(x),\,\Phi_{(q)}\left(\exp_x^{(\gamma)}(\upsilon )\right)\right)\,d\widehat{\mu}_\gamma(\upsilon )\;,\nonumber
\end{eqnarray}
 where $\left|D_\delta(\epsilon)\right|=4\pi\epsilon^2$ is the Euclidean area of the disk $D(x,\epsilon)\subset T_x\Sigma$,\; and $d_h$ is the distance (\ref{disth}). 
 \begin{rem}
 Note that $\Phi_{(q)}(x)$ is an actual center of mass of the subset 
$\widetilde{\Phi}_{(q)}(D(x,\epsilon))$, with respect to the push--forward measure $\left(\widetilde{\Phi}_{(q)} \right)_{\sharp }d\widehat{\mu}_\gamma$, if the point $\Phi_{(q)}(x)\in N^{n+q}$ minimizes the function
\begin{eqnarray}
\label{GCM}
&&Y\longmapsto G\left[Y,\,D(x,\epsilon)\right]\\
\nonumber\\
&&=\,\,\frac{1}{\left|D_\delta(\epsilon)\right|}\,\int_{D(x,\epsilon)}\,d^2_h\left(Y,\,\Phi_{(q)}\left(\exp_x^{(\gamma)}(\upsilon )\right)\right)\,d\widehat{\mu}_\gamma(\upsilon )\;.\nonumber
\end{eqnarray}
In particular, we shall say that the map $\Phi_{(q)}\in \mathcal{H}^{1}\left(\Sigma,\,N^{n+q}\right)$ is an $\epsilon$--center of mass in the above sense if $\Phi_{(q)}(x)$ minimizes (\ref{GCM}) for almost every $x\in\Sigma$.
\end{rem} 
\noindent From the very structure of (\ref{HACM}) it follows that, by integrating the center of mass function $G\left[D(x,\epsilon)\right]$ over $(\Sigma, \gamma)$, we can characterize an approximate energy functional for maps of low regularity according to
\begin{defn}
Let $\Phi_{(q)}\in L^2\left(\Sigma,\,N^{n+q}\right)$ be a square summable map $\Phi_{(q)}:(\Sigma,\gamma)\rightarrow (N^{n+q},\,d_h)$, and let
\begin{equation}
\label{approxE}
E_{\epsilon }\,[\Phi_{(q)},\,d_h]\,:=\,\frac{1}{2}\,\int_{\Sigma}\,\frac{G\left[\Phi_{(q)}(x),\,D(x,\epsilon)\right]}{\epsilon ^2}
\,d\mu_{\gamma}(x)\;,
\end{equation}
denote the $\epsilon$--approximate energy associated with the center of mass 
function $x\longmapsto G\left[D(x,\epsilon)\right]$, then the $L^2$--energy of the map $\Phi_{(q)}$ is defined by
\begin{equation}
\label{L2def}
\mathcal{E}\,[\Phi_{(q)},\,d_h]\,:=\,\lim_{\epsilon \rightarrow 0}\,E_{\epsilon }\,[\Phi_{(q)},\,d_h]\,
\in\,\mathbb{R}\,\cup \,\{+\infty \}\;,
\end{equation}
where the limit is in the sense of weak convergence of measures.
\end{defn}
\begin{rem}
It is elementary to prove, (as in \cite{jost}, Lemma 8.4.1), that the $\epsilon$--approximate energy functional $E_{\epsilon }\,[\Phi_{(q)},\,d_h]$ is minimized iff  $\Phi_{(q)}(x)$ is a center of mass for almost all $x\in \Sigma$.  
\end{rem}
\noindent The basic observation in this general setting is that the center of mass functional $\mathcal{E}\,[\Phi_{(q)},\,d_h]$ reduces to the standard harmonic map energy ${E}\,[\Phi_{(q)},\,h^{(q)}]$ whenever the map $\Phi_{(q)}$ is regular enough. In particular we have
\begin{prop}
If $\Phi_{(q)}$ is a localizable map $\in \mathcal{H}^{1}\left(\Sigma,\,N^{n+q}\right)$, then
\begin{equation}
\label{harmoneq}
\lim_{\epsilon \rightarrow 0}\,E_{\epsilon }\,[\Phi_{(q)},\,d_h]\,=\,{E}\,[\Phi_{(q)},\,h^{(q)}]\;.
\end{equation}
\end{prop}
\begin{proof}
This is a particular case of a more general result, (for instance one can adapt, with obvious modifications, Theor. 8.4.1 of \cite{jost}), stating that the equality (\ref{harmoneq}) holds as soon as $\Phi_{(q)}\in L^{2}\left(\Sigma,\,N^{n+q}\right)$ is localizable and  the functional ${E}\,[\Phi_{(q)},\,h^{(q)}]$ can be defined, (in the local chart associated to the localization adopted).
\end{proof}
\noindent It is worthwhile to note that we can express the center of mass functional $E_{\epsilon }\,[\Phi_{(q)},\,d_h]$ in terms of the metric geometry of the factor manifolds $(M,d_g)$ and $\mathbb{T}^q$ defining $N^{n+q}$. Explicitly, according to the definition of the map $\Phi_{(q)}$, and of the characterization (\ref{distl}) and (\ref{disth})  of the distance function $d_h$,  we can write
\begin{eqnarray}
&& d_h^2\,\left(\Phi_{(q)}(x),\Phi^{(q,x)}(\upsilon )\right)\,=\,d_g^2\,\left(\phi(x),\phi(\exp_x^{(\gamma)}(\upsilon ))\right)\\
\nonumber\\
&&+\,\frac{q^2\,e^{-\,\tfrac{2f(\phi(x))}{q}}}{4}\,\sum_{k=1}^q\,d_{g}^2(\phi_{cm},\phi_{(k)})\,
\biggl|e^{\,\tfrac{f(\phi(x))}{q}}\,-\,e^{\,\tfrac{f(\phi(\exp_x^{(\gamma)}(\upsilon )))}{q}}\,  \biggr|^2\nonumber\;,
\end{eqnarray} 
which, inserted in (\ref{approxE}), provides
\begin{eqnarray}
\label{ripart}
&& \;\;\;\;\;\;\;\;E_{\epsilon }\,[\Phi_{(q)},\,d_h]\\
\nonumber\\
&&=\,\frac{1}{2\,D_\delta(\epsilon)}\,\int_{\Sigma}\,
\Biggl(\int_{D(x,\epsilon)}\,\frac{d_g^2\,\left(\phi(x),\phi(\exp_x^{(\gamma)}(\upsilon ))\right)}{\epsilon^2}\,d\widehat{\mu}_{\gamma}(\upsilon )   \Biggr)\,d\mu_{\gamma}(x)\nonumber\\
\nonumber\\
\nonumber\\
&&+\,\frac{q^2}{\epsilon^2\,|D_\delta(\epsilon)|}\,\frac{\sum_{k=1}^q\,d_{g}^2(\phi_{cm},\phi_{(k)})}{8}\,\times\nonumber\\
\nonumber\\
&&\times\int_{\Sigma}
\Biggl(\int_{D(x,\epsilon)}{e^{-\tfrac{2f(\phi(x))}{q}}}\left|e^{\tfrac{f(\phi(x))}{q}}-e^{\tfrac{f(\phi(\exp_x^{(\gamma)}(\upsilon )))}{q}}\right|^2d\widehat{\mu}_{\gamma}(\upsilon )   \Biggr)d\mu_{\gamma}(x)\nonumber\;.
\end{eqnarray}  
\\
\\
\noindent In analogy with (\ref{L2def}), we define the harmonic map energy functional for maps $\phi$ in $L^2\left(\Sigma, M \right)$ according to
\begin{eqnarray}
&&\;\;\;\;\;\;\;\;\;\mathcal{E}\,[\phi,\,d_g]\\
\nonumber\\
&&:=\,\lim_{\epsilon \rightarrow 0}\,
\,\frac{1}{D_\delta(\epsilon)}\,\int_{\Sigma}\,
\Biggl(\int_{D(x,\epsilon)}\,\frac{d_g^2\,\left(\phi(x),\phi(\exp_x^{(\gamma)}(\upsilon ))\right)}{\epsilon^2}\,d\widehat{\mu}_{\gamma}(\upsilon )   \Biggr)\,d\mu_{\gamma}(x)\;,\nonumber
\end{eqnarray}  
and  the Dirichlet energy for functions $f\circ \phi$ in $L^2\left(\Sigma, \mathbb{R} \right)$
\begin{eqnarray}
&& \mathcal{D}[f(\phi;q)]\,:=\, \lim_{\epsilon \rightarrow 0}\,\int_{\Sigma}\,\Biggl(
\int_{D(x,\epsilon)}\,q^2\,\frac{e^{-\,\tfrac{2f(\phi(x))}{q}}}{2\,\epsilon^2\,|D_\delta(\epsilon)|}\,\times\Biggr.\\
\nonumber\\
\nonumber\\
&&
\times\,\Biggl.\left|e^{\,\tfrac{f(\phi(x))}{q}}\,-\,e^{\,\tfrac{f(\phi(\exp_x^{(\gamma)}(\upsilon )))}{q}}\right|^2\,d\widehat{\mu}_{\gamma}(\upsilon )   \Biggr)\,d\mu_{\gamma}(x)\;.\nonumber
\end{eqnarray}
\\
We have   
\begin{equation}
\label{qharmwarpd}
\mathcal{E}\,[\Phi_{(q)},\,d_h]\,:=\,\mathcal{E}\,[\phi,\,d_g]\,+\,\frac{F(\phi_{cm};\,q)}{2}\,\mathcal{D}[f(\phi;q)]\;,
\end{equation}
with the obvious proviso that for square summable maps  each one of the (lower semicontinuous) functionals $\mathcal{E}\,[\Phi_{(q)},\,h^{(q)}]$, $\mathcal{E}\,[\phi,\,d_g]$, and $\mathcal{D}[f(\phi;q)]$  defined above  can be infinite. Again, for localizable Sobolev maps $\Phi_{(q)}$ one can prove that  $\mathcal{E}\,[\phi,\,d_g]$, and $\mathcal{D}[f(\phi;q)]$ can be identified with their corresponding functional   $E[\phi,\,g]$, and 
$\mathfrak{D}[f(\phi;q)]$, appearing in the factorization (\ref{qharmwarp}).
\begin{rem}  The approximation scheme underlying (\ref{approxE}) and (\ref{ripart}) naturally appears,  somewhat in disguise, when considering the cut--off action functional for the regularized lattice versions of quantum non--linear $\sigma$ models, (see \emph{e.g.} \cite{gawedzki}).
\end{rem}

\subsection{Heat Kernel embedding in Wasserstein space}
\label{wass}  

To put things in context,  let us consider the isometric embedding of $(M,d_g)$ into $(\rm{Prob}(M),\,d_g^W)$ defined by  
\begin{eqnarray}
\label{isoemM}
(M,d_g)\,&\hookrightarrow&\, \left(\rm{Prob}(M),\; d_{\,g}^{\,W}\, \right)\\
y\,\;\;\;\;\;&\longmapsto&\, \;\;\;\;\;\delta _y\;,\nonumber
\end{eqnarray} 
where $\delta_y$ is the Dirac measure supported at the generic $y\in (M,g)$. It is easy to prove  that this is indeed an isometry since one directly computes $d_g^W(\delta_y, \delta_z)\,=\,d_g(y,z)$, by  
using the obvious optimal plan $d\sigma(u,v)=\delta_y(u)\otimes\delta_z(v)$ in (\ref{wassdist}).  
This isometry allows to represent the center of mass coupling as
\begin{equation}
\label{CMWASS}
a\,:=\,q^{-1}\,F(\phi_{cm};q)\,=\,\frac{1}{2q}\,\sum_{k=1}^q\,\left[d_{\,g}^{\,W}\,(\delta_{(k)},\,\delta_{cm}) \right]^2\;,
\end{equation}
where $\delta_{(k)}$ and $\delta_{cm}$ are the Dirac measures supported at $\phi_{(k)}$ and $\phi_{cm}$, respectively.  The advantage of the, otherwise rather formal, Wasserstein representation (\ref{CMWASS}) is that we can view $\delta_{(k)}$ and $\delta_{cm}$ as heat sources.  This suggests a natural heat kernel technique for the  geometrical analysis of the scaling 
properties of  $q^{-1}\,F(\phi_{cm};q)$ and, as we shall see, of the functionals $E[\phi,\,g]$, and $S[\gamma,\phi;\,a,g, d\omega]$. \\
\\
\noindent
We start exploiting the Wasserstein representation (\ref{CMWASS}) and discuss the deformation of the center of mass coupling $q^{-1}\,F(\phi_{cm};q)$ by extending (\ref{isoemM}) to a heat kernel embedding of the constant maps $\{\phi_{(k)}\}$.  To this end, let $t$ be a running (squared) length scale and denote 
by $p^{M}_t\,:\,M\times M\times\mathbb{R}_{>0}\longrightarrow\mathbb{R}_{>0}$ the heat kernel on $(M,g)$, \emph{i.e.} the (minimal positive) solution of the heat equation
\begin{eqnarray}
\label{heatfloKM}
&&\left(\frac{\partial }{\partial  t}\,-\,\bigtriangleup^M _{(g,y)}\right)\,p^M _t(y, z)=0\;,\\
\nonumber\\
&&\lim_{t\searrow 0^+}\;p^M _t(y, z)\,d\mu_g(z)\,=\, \delta_{z}\;,\nonumber
\end{eqnarray}
where $\delta_{z} (y)$ and $\bigtriangleup^M _{(g,y)}$ denote the Dirac measure at $z\in M$ and the Laplace--Beltrami operator on $(M,g)$, respectively. Also let 
\begin{equation}
p^M_t(y,\phi_{(k)})\,=\,\int_{M_z}\,p^M_{t-s}(y,z)\,
p^M_s(\phi_{(k)},z )\,\,d\mu_g(z),\;\;\;\;\;\;t>s>0\;,
\end{equation}
and 
\begin{equation}
p^M_t(y,\phi_{cm})\,=\,\int_{M_z}\,p^M_{t-s}(y,z)\,
p^M_s(\phi_{cm},z )\,\,d\mu_g(z),\;\;\;\;\;\;t>s>0\;,
\end{equation}
be the heat kernels with sources the Dirac measures $\delta _{(k)}$ and $\delta _{cm}$, respectively. Note that,
away of the cut locus $\rm{Cut}\,(z)$, \; $p^M_t(y,z)$ can be differentiated any number of times  and its spatial derivatives commute with the $t\searrow 0^+$ limit, (this is an elementary consequence of a theorem by Malliavin and Strook \cite{malliavin}--see also \cite{berline}). In particular, it is useful to recall the following asymptotics related to the geometry of heat diffusion on Riemannian manifolds.
\begin{thm}\cite{pleijel}, (see also \cite{chavel}, \cite{neel}).
There are  smooth functions $\Xi  _k(y, z)$ defined on $(M\times M)/\rm{Cut}\,(z)$, with $\Upsilon _0(y, z)\,=\,1$, such that the asymptotic expansion
\begin{equation}
p^M_t(y,z)\;d\mu_g(y)\,- \,\left(\frac{1}{4\pi t} \right)^{n/2}\,e^{-\,\tfrac{d_g^2(y, z)}{4\,t}}\;\sum_{k=0}^{N }\,\Xi _k(y, z)\,t^k\;d\mu_g(y)\,=\,\mathcal{O}\left(t^{N+1-\tfrac{n}{2}}\right)
\label{asympt5}
\end{equation}
holds uniformly as $t\searrow 0^+$ on compact subsets of $(M\times M)/\rm{Cut}\,(z)$.
\end{thm}
\begin{thm} (T. H. Parker \cite{parker}). \label{parkth}
Let $\rm{Cut}_L\subset (M\times M)$ be the set of points $(y,z)$ such that $z$ is conjugate to $y$ along 
some geodesics $\gamma$ with length at most $L$. If $D_\gamma$ denotes the (absolute value of the) Jacobian of the exponential map along $\gamma$, then
\begin{equation}
\label{parkeras}
p^M_t(y,z)\,=\,\left(\frac{1}{4\pi t} \right)^{n/2}\,\sum_{\gamma}\,e^{-\,\tfrac{1}{4}\,\int_0^t\left|\dot{\gamma(s)} \right|^2\,ds}\,D_\gamma^{-\,\tfrac{1}{2}}\;\left(1+\mathcal{O}(t)\right)\;, 
\end{equation}
where the convergence is uniform as $t\searrow 0^+$ on compact subsets of $(M\times M)/\rm{Cut}_L\,(y,z)$, and where the summation  is over all locally minimal geodesics $\gamma$ in $(M,g)$ connecting $y$ to $z$.
\end{thm}
\begin{rem}
The former is the well--known heat kernel asymptotics, the latter is a small--time asymptotics recently obtained by T. H. Parker \cite{parker}, showing quite explicitly that  \emph{heat tends to diffuse instantaneously along geodesics}. In particular, (\ref{parkeras}) can be heuristically identified with the semiclassical approximation to the (Euclidean) path integral representation of $p^M_t(y,z)$
provided by
\begin{equation}
p^M_t(y,z)\,=\, \int_{\mathcal{P}_{yz}}\, e^{-\,\tfrac{1}{4}\,\int_0^t\left|\dot{\gamma(s)} \right|^2\,ds}\,\mathcal{D[\gamma]}\;,
\end{equation}
where $\mathcal{D[\gamma]}$ is a functional measure on the space $\mathcal{P}_{yz}$ of all paths $\gamma$ connecting, in time $t$, the points $y$ and $z$. 
\end{rem} 
Since $\tfrac{1}{4}\,\int_0^t\left|\dot{\gamma(s)} \right|^2\,ds\,=\,\frac{d_g^2(y, z)}{4\,t}$ along a minimal geodesic $\gamma$, both these asymptotic expansions provide the heuristics of the celebrated  Varadhan's large deviation formula \cite{varadhan},
\begin{equation}
\label{largedev0}
-\,\lim_{t\searrow 0^+}\,{t}\,\ln\,\left[\,p^M_{t}(y, z)\right]\,=\,
\frac{d_g^2(y, z)}{4}\;,
\end{equation}
which, together with the expression (\ref{Fcm}) for the center of mass function, (evaluated at $\phi_{cm}$), allows us to provide yet another formal representation for the coupling,
\begin{equation}
\label{hrep}
F(\phi_{cm};q)\,=\,-\,\lim_{t\searrow 0^+}\,{2\,t}\,\ln\,\prod_{k=1}^q\,p_t^M(\phi_{(k)},\,\phi_{cm})\;.
\end{equation}
\\
\noindent
The convergence in (\ref{largedev0}), and hence in (\ref{hrep}) is uniform, and
together with (\ref{CMWASS}), the expression (\ref{hrep}) suggests the possibility of deforming  $q^{-1}\,F(\phi_{cm};q)$ by using the heat kernel 
flow  $(\delta_{(k)},\,t)\mapsto p^M_t$, $t\in (0,\infty )$ in $(\rm{Prob}(M),\,d_g^W)$. This is tantamount  
to extending the isometry (\ref{isoemM}) to the injective embedding defined, for any fixed length scale $t\in [0, \infty)$, by the map 
\begin{eqnarray}
\label{heatinjemb}
(M,d_g)\times \mathbb{R}_{\geq0}\,&\hookrightarrow&\, \left(\rm{Prob}(M),\; d_{\,g}^{\,W}\, \right)\\
(y,\,t)\,\;\;\;\;\;&\longmapsto&\, \;\;\;\;\;p_t^M(\cdot,y)d\mu_g(\cdot )\;,\nonumber
\end{eqnarray}
which associates to each point $y\in (M, g)$ the corresponding heat kernel measure with source at $y$. 
\begin{rem} 
The injectivity of the heat kernel embedding in quadratic Wasserstein space is discussed at length in \cite{Carlo}. Below we shall analyze (\ref{heatinjemb}) in detail in the case of the weighted heat kernel associated with $(M,g,\,d\omega)$.
\end{rem}
 Under the action of (\ref{heatinjemb}),  the constant maps $\{\phi_{(k)}\}$ and their center of 
mass $\phi_{cm}$  are injected in $\left(\rm{Prob}(M),\; d_{\,g}^{\,W}\, \right)$ according to
\begin{eqnarray}
(\phi_{(k)},\,t)\,&\longmapsto&\,(\delta_{(k)},\,t)\,\longmapsto\,p_t^M(\cdot, \phi_{(k)})d\mu_g(\cdot )\\
(\phi_{cm},\,t)\,&\longmapsto&\,(\delta_{cm},\,t)\,\longmapsto\,p_t^M(\cdot, \phi_{cm})d\mu_g(\cdot )\,.\nonumber
\end{eqnarray}
The coupling $a:=q^{-1}\,F(\phi_{cm};q)$ is correspondingly  deformed into
\begin{eqnarray}
\label{defCMWASS}
&&(a,\,t)\,\longmapsto \,a_t\\
\nonumber\\
&&:=\,\frac{1}{2\,q}\,\sum_{k=1}^q\,
\left[d_{\,g}^{\,W}\,\left(p_t^M(\cdot, \phi_{(k)})d\mu_g(\cdot ),\,p_t^M(\cdot, \phi_{cm})d\mu_g(\cdot )\right) \right]^2\;,\nonumber
\end{eqnarray}  
with
\begin{equation}
\lim_{t\searrow 0}a_t\,=\,\frac{1}{2\,q}\,\sum_{k=1}^q\,\left[d_{\,g}^{\,W}\,(\delta_{(k)},\,\delta_{cm}) \right]^2\,=\,\frac{F(\phi_{cm};q)}{q}=a\;.
\end{equation}
 The induced scaling in passing from $a$ to $a_t$  is geometrically controlled 
by the 
\begin{prop} 
If $K$ denote the lower bound of the Ricci curvature of $(M,g)$, (\emph{i.e.} $Ric_g\,(v,v)\,\geq K\,|v|^2$, \,$\forall \,v\,\in\, TM$), then the  coupling  $a$ scales under heat kernel deformation of the constant maps $\{\phi_{(k)}\}$ according to 
\begin{equation}
\label{FWD}
a_t\,\leq\,e^{-\,2\,K\,t}\,a\;.
\end{equation}
\end{prop}
\begin{proof}
(\ref{FWD}) directly follows from the basic inequality  \cite{sturm}, \cite{sturm2}, \cite{sturm3},
\begin{eqnarray}
&& d_g^W\left(p^M_t(\cdot \, ,y)d\mu_g(\cdot )\,,\,p^M_t(\cdot \, ,z)d\mu_g(\cdot )\right)\\
\nonumber\\
&\leq&\,e^{-K\,t}\,
d_g^W\left(\delta _{y},\,\delta _{z}\right)\,=\,e^{-K\,t}\,
d_g\left({y},\,{z}\right)\;,\;\;\;\;\;\forall\, t>0,\;,\nonumber
\end{eqnarray}
governing the Wasserstein geometry of heat diffusion on a Riemannian manifold.
\end{proof}
\noindent The general role of the center of mass, discussed in Section \ref{digression}, indicates that we can analyze the scaling behavior of the harmonic energy functional $E[\Phi_{(q)},\,h^{(q)}]$, and hence of the dilatonic action $S[\gamma,\phi;\,a,g, d\omega]$,  along the same lines described above. To this end, we need to characterize the heat kernel injection of $(\Phi_{(q)}(\Sigma),\,h^{(q)})$ in the Wasserstein space $\left(\rm{Prob}(N),\; d_{\,h}^{\,W}\, \right)$ of Borel probability measures on $N^{n+q}$. 

\section{A warped Gigli--Mantegazza heat kernel embedding}
\label{awhke}
The structure of $\left(M\times_{\omega}\mathbb{T}^q,\,h^{(q)}\right)$ allows us to model a significant part of our analysis on the results discussed in a terse paper by N. Gigli and C. Mantegazza \cite{Carlo}. They analyze in detail the  geometry of the heat kernel embedding associated to the Laplace--Beltrami heat semigroup, connecting it to the Ricci flow.  In our case, the relevant heat semigroup is the one generated by a weighted Laplacian on $(M,g,\,d\omega)$, and we provide an extension of their results to the warped Riemannian manifold  $\left(M\times_{\omega}\mathbb{T}^q,\,h^{(q)}\right)$. This allows us discuss the scaling behavior of $E[\Phi_{(q)},\,h^{(q)}]$ and connect it to the Hamilton--Perelman version of the Ricci flow.\\
\\
\noindent 
As a preliminary step, we need to define a heat kernel adapted to the warped product structure of $\left(M\times_{\omega}\mathbb{T}^q,\,h^{(q)}\right)$. To this end, let $Y^a:=(y^i,\,\upsilon ^{\alpha})$, with $a=1,\ldots,n+q$,\, 
$i=1,\ldots,n$,\, and $\alpha=1,\ldots,q$ denote coordinates adapted to  $M\times_{\omega}\mathbb{T}^q$. Under such a splitting, a standard computation for the Laplacian $\bigtriangleup _{h}$ on $\left(M\times_{\omega}\mathbb{T}^q,\,h^{(q)}\right)$  provides the relation
\begin{eqnarray}
\label{Lapfactor}
\\
\bigtriangleup _{h}\,&:=&\,\frac{1}{\sqrt{\det h}}\,\frac{\partial }{\partial Y^a}\,\left(\sqrt{\det h}\,h^{ab}\, \frac{\partial }{\partial Y^b} \right)\nonumber\\
\nonumber\\
&=&\,\frac{e^{\,f}}{\sqrt{\det g}}\,\frac{\partial }{\partial y^i}\,\left(\sqrt{\det g}\,e^{-\,f}\,g^{ik}\, \frac{\partial }{\partial y^k} \right)\,+\,e^{\tfrac{2f}{q}}\,\frac{\partial }{\partial \upsilon^\alpha}\,\left(\delta^{\alpha\beta}\, \frac{\partial }{\partial \upsilon^\beta} \right)\nonumber\\
\nonumber\\
&=&\,\bigtriangleup _{g}\,+\,e^{\tfrac{2f}{q}}\,\bigtriangleup _{\mathbb{T}}\,-\,\nabla f\cdot \nabla \nonumber\;,
\end{eqnarray}
where  $\bigtriangleup _{g}$, and $\bigtriangleup _{\mathbb{T}}$ respectively denote the Laplace--Beltrami operator on  $(M,g)$, \,and on $(\mathbb{T}^q,\,\delta)$, and where $\nabla f\cdot \nabla\,=\,g^{ik}\,\nabla_i f\,\nabla_k$ is the distributional directional derivative along the gradient of the Lipschitz function $f$. According to (\ref{weightdiv}) the operator $\bigtriangleup _{g}\,-\,\nabla f\cdot \nabla$ appearing in (\ref{Lapfactor}) is the  weighted Laplacian on $(M,g,\,d\omega)$, (cf. (\ref{weightdiv})),   
\begin{equation}
\label{Wlap}
\bigtriangleup _{\omega}\,:=\,{div}_{\omega}\,\nabla \,=\,\bigtriangleup _{g}\,-\,\nabla f\cdot \nabla\;,
\end{equation}
and we can write
\begin{equation}
\label{Lapfactor2}
\bigtriangleup _{h}\,=\,\bigtriangleup _{\omega}\,+\,e^{\tfrac{2f}{q}}\,\bigtriangleup _{\mathbb{T}}\;.
\end{equation}
The properties of the heat kernel associated to $\bigtriangleup _{\omega}$ are provided by \cite{grigoryan}
\begin{thm}
The weighted Laplacian  $\bigtriangleup _{\omega}$ is symmetric  with respect to the defining measure $d\omega$, and can be extended to a self--adjoint operator in $L^2(M,\,d\omega)$ generating  the heat semigroup $e^{t\,\bigtriangleup _{\omega}}$, $t\in \mathbb{R}_{>0}$. The associated heat kernel $p^{(\omega)}_t(\cdot \, ,z)$ is defined as the minimal positive solution of
\begin{eqnarray}
\label{heatfloweight}
&&\left(\frac{\partial }{\partial  t}\,-\,\bigtriangleup _{\omega}\right)\,p^{(\omega)} _t(y, z)=0\;,\\
\nonumber\\
&&\lim_{t\searrow 0^+}\;p^{(\omega)} _t(y, z)\,d\omega(z)=\, \delta_{z}\;,\nonumber
\end{eqnarray}  
with $\delta_{z}$ the Dirac measure at $z\in (M,\,d\omega)$.\; The heat kernel $p^{(\omega)} _t(y, z)$ is $C^\infty$ 
on $\mathbb{R}_{>0}\times M\times M$, is symmetric $p^{(\omega)} _t(y, z)\,=\,p^{(\omega)} _t(z, y)$, satisfies the semigroup 
identity  $p^{(\omega)} _{t+s}(y, z)=\int_M\,p^{(\omega)} _t(y, x)p^{(\omega)} _s(x, z)\,d\omega(x)$, and $\int_M\,p^{(\omega)} _t(y, z)\,d\omega(z)=1$. Moreover, Varadhan's large deviation formula holds
\begin{equation}
\label{largedevweigh}
-\,\lim_{t\searrow 0^+}\,{t}\,\ln\,\left[\,p^{(\omega)}_{t}(y, z)\right]\,=\,
\frac{d_g^2(y, z)}{4}\;,
\end{equation}
where the convergence is uniform over all $(M,g,\,d\omega)$.
\end{thm}
\begin{proof}
This characterization of the heat kernel $p^{(\omega)} _t(\cdot , \cdot )$ is a direct consequence of the properties of the weighted Laplacian  $\bigtriangleup _{\omega}$. For a more detailed discussion see Grigoryan \cite{grigoryan} Th.7.13, and $\S $ 7.5, Th. 7.20, for the smoothness property of the weighted heat kernel $p^{(\omega)} _t(\cdot , \cdot )$.
\end{proof}
\begin{rem}
From the point of view  of heat theory, the operator $\bigtriangleup _{h}$ generates a diffusion  in the base manifold $(M,g\,,d\omega)$ driven by $\triangle _\omega$, weakly coupled to an independent heat propagation in the torus fibers $\mathbb{T}_y^q$ via a thermal diffusivity given by $\exp\{{2f(y)}/{q}\}$. Since we do not need to locally vary the geometry of the torus fibers, (up to a fiberwise rescaling), we can freeze the heat diffusion in each $\mathbb{T}_y^q$, and consider only the heat kernel embedding induced by $p^{(\omega)}_{t}(y, z)$.
\end{rem}
\begin{lem}
Let  $\left(\rm{Prob}(N),\,d_h^W\right)$,  be the Wasserstein spaces of probability measures over the warped manifold $\left(N\simeq M\times_{\omega}\mathbb{T}^q,\,h^{(q)}\right)$. If $p^{(\omega)} _t(\cdot \, , z)\,d\omega(\cdot )\otimes \,\delta_\zeta^{\;\mathbb{T}^q}$ denotes the tensor product between the weighted heat kernel on $(M,g,\,d\omega)$ and the Dirac measure supported at $\zeta\in \mathbb{T}^q$, then the map 
\begin{eqnarray}
\label{warpembed}
\\
\Upsilon  _{t}\,:\,(M\times_{(\omega)}\mathbb{T}^q,\,h^{(q)})\,&\hookrightarrow &\,\;\;\left(\rm{Prob}(N),\,d_h^W\right)\nonumber\\
\nonumber\\
(z,\zeta)\,\;\;\;\;\;\;&\longmapsto &\,\Upsilon  _{t}(z,\zeta)\,:=\,p^{(\omega)} _t(\cdot \, , z)\,d\omega(\cdot )\otimes \,\delta_\zeta^{\;\mathbb{T}^q}
\,,\nonumber
\end{eqnarray} 
is, for any $t\geq 0$, injective.
\end{lem}
\begin{proof}  We have the obvious inclusion
\begin{eqnarray}
&&\left(\rm{Prob}\,(M,g),\,d_g^W\right)\times \left(\rm{Prob}\,(\mathbb{T}^q),\,d_{\mathbb{T}}^W\right)\,
\hookrightarrow \,\left(\rm{Prob}(N),\,d_h^W\right)\\
\nonumber\\
&&\;\;\;\;\;\;\;\left(p^{(\omega)} _t(\cdot \, , z)\,d\omega(\cdot ),\, \delta_\zeta^{\;\mathbb{T}^q}\right)\,\longmapsto \,
p^{(\omega)} _t(\cdot \, , z)\,d\omega(\cdot )\otimes \,\delta_\zeta^{\;\mathbb{T}^q}\;.\nonumber
\end{eqnarray} 
The map $\Upsilon  _{t}$ restricted to the torus fibers, \emph{i.e.}\; $\zeta\longmapsto \delta_\zeta^{\;\mathbb{T}^q}$, 
 is an isometric embedding of $\mathbb{T}^q$ into $(\rm{Prob}\,(\mathbb{T}^q),\,d_{\mathbb{T}}^W)$. By adaptating to the heat semigroup generated by $\triangle _\omega$ the analysis of the injectivity of the Laplace--Beltrami heat flow, (cf. Th. 2.3 and Proposition 5.16 of \cite{Carlo}), it follows that the restriction, $z\mapsto p^{(\omega)} _t(\cdot \, , z)\,d\omega(\cdot )$, of $\Upsilon  _{t}$ to $(M,g,\, d\omega)$ is an injective embedding of $(M,g,\,d\omega)$ into $(\rm{Prob}\,(M),\,d_g^W)$. Hence $\Upsilon  _{t}$ injects in 
$\left(\rm{Prob}(N),\,d_h^W\right)$. 
\end{proof}
To discuss the properties of the  map $\Upsilon  _{t}$ defined by (\ref{warpembed}), let us use coordinates $Z^a\,:=\,(z^i,\,\zeta ^\alpha)$,\,with $a=1,\ldots,n+q$,\; $i=1,\ldots,n$,\, and\;$\alpha=1,\ldots,q$,\;adapted to the product structure of the manifold $N:=\,M\times_{\omega}\mathbb{T}^q$. In the corresponding coordinate bases 
$\{\partial_i\}$,\;$\{\partial_\alpha\}$, let us consider vector fields $U_\perp ^i\partial _i\,\in C^{\infty }(M,TM)$ and $U_\parallel ^\alpha\partial _\alpha\,\in C^{\infty }(\mathbb{T}^q,T\,\mathbb{T}^q)$, and the associated vector field $U\in C^{\infty }(N,TN)$
 \begin{equation}
 \label{adapted}
 U\,=\,U^a(z,\zeta)\partial _a\,:=\,U_\perp  ^i(z)\partial _i\,+\,U_\parallel  ^\alpha(\zeta)\partial _\alpha\;.
 \end{equation}
 For any $t>0$, we can naturally associate to the vector field $U_\perp$ a corresponding element in $T_{p_t(d\omega)}\,\rm{Prob}(M, g)$, the tangent space to $\rm{Prob}(M, g)$ at $p_t(d\omega):=p^{(\omega)} _t(\cdot \,  , z)\,d\omega(\cdot )$.
 \begin{lem}
\label{tangVECT} 
Let $t\longmapsto p^{(\omega)} _t(\cdot \,  , z)\,d\omega(\cdot )$,\,$t\in(0,\infty)$, be the flow of probability measure in $\rm{Prob}_{ac}(M,g)$ defined by the weighted heat kernel $p^{(\omega)} _t(\cdot \, , z)$. Then, the map
\begin{eqnarray}
\label{mapU}
TM\times (0,\infty )&\,\longrightarrow&\, C^{\infty }(M,\mathbb{R})\\
\nonumber\\
(z, U_\perp(z);\,t)\,&\longmapsto& \,U_\perp ^i(z)\,\nabla_i^{(z)}\,\ln\,p^{(\omega)} _t(\cdot \,  , z)\;,\nonumber
\end{eqnarray}
defines, for each $t\in(0,\infty)$, a tangent vector in $T_{p_t(d\omega)}\,\rm{Prob}(M, g)$. 
\end{lem}
\noindent Note that the superscript ${}^{(z)}$ in (\ref{mapU})  signifies that the differentiation is applied to the indicated variable.
\begin{rem}
\label{notatremark}
In the notation for the tangent space $T_{p_t(d\omega)}\,\rm{Prob}(M, g)$   we  used the shorthand $p_t(d\omega)$ to denote the probability measure on $(M,g)$ associated with the heat kernel distribution $p^{(\omega)} _t(\cdot \,  , z)$ evaluated at time $t$ and with a fixed source $\delta_z$. This notation may become ambiguous when considering on $T_{p_t(d\omega)}\,\rm{Prob}(M, g)$ the associated $L^2$ inner product spaces, since these will functionally depend on the heat source. For instance, $\psi \in T_{p_t(d\omega)}\,\rm{Prob}(M, g)$, has a  natural point--dependent $L^2$ norm evaluated, at fixed heat source $\delta_z$, according to 
\begin{equation}
\parallel \psi\parallel ^2_{L^2(p_t(d\omega,z))}\,:=\,\int_M\,|\psi(y)|^2\,p^{(\omega)} _t(y \,  , z)\,d\omega(y)\;.
\end{equation}
We use the notation $L^2(p_t(d\omega,z))$ to emphasize, whenever necessary, the location of the fixed heat source, and the simpler  notation  $L^2(p_t(d\omega))$ if there is no danger of confusion.  
\end{rem}
\begin{proof} To prove lemma \ref{tangVECT}, let us observe that for $t>0$ the function $y\mapsto U_\perp ^i(z)\,\nabla_i^{(z)}\,\ln\,p^{(\omega)} _t(y, z)$ is smooth. By integrating over $(M,\,p^{(\omega)} _t\,d\omega)$ we get
\begin{eqnarray}
\label{fredzero}
&& \int_M\,\left(U_\perp ^i(z)\,\nabla_i^{(z)}\,\ln\,p^{(\omega)} _t(y , z)\right)\,p^{(\omega)} _t(y , z)\,d\omega(y)\\
\nonumber\\
&&=\,\int_M\,U_\perp ^i(z)\,\nabla_i^{(z)}\,p^{(\omega)} _t(y , z)\,d\omega(y)\nonumber\\
\nonumber\\
&&=\,U_\perp ^i(z)\,\nabla_i^{(z)}\,\int_M\,p^{(\omega)} _t(z , y)\,d\omega(y)\,=\,0\;,\nonumber
\end{eqnarray}
where, in the last passage, we have exploited the symmetry of the heat kernel, $p^{(\omega)} _t(z , y)=p^{(\omega)} _t(y , z)$, and $\int_M\,p^{(\omega)} _t(z , y)\,d\omega(y)\,=\,1$. It follows that, for each $t\in(0,\infty)$,\;(\ref{mapU}) defines, according to (\ref{HilbSPace}), a tangent vector in $T_{p_t(d\omega)}\,\rm{Prob}(M, g)$.
\end{proof}
\begin{rem}
By proceeding similarly, and using the map 
\begin{equation}
(z, U_\perp(z);\,t)\,\longmapsto \,U_\perp ^i(z)\,\nabla_i^{(z)}\,p^{(\omega)} _t(\cdot \,  , z)\,
\end{equation}
we can interpret $U_\perp ^i(z)\,\nabla_i^{(z)}\,p^{(\omega)} _t(\cdot \,  , z)$ as an element of the tangent space $T_{d\omega}\,\rm{Prob}(M, g)$.
\end{rem}
\begin{rem}
Roughly speaking, the vector field $U_\perp ^i(z)\,\nabla_i^{(z)}\,\ln\,p^{(\omega)} _t(\cdot \,  , z)$,  (or, equivalently, $U_\perp ^i(z)\,\nabla_i^{(z)}\,p^{(\omega)} _t(\cdot \,  , z)$), can be thought of as describing the perturbation in the probability measure $p^{(\omega)} _t(\cdot \,  , z)\,d\omega(\cdot )$  as we vary the heat source $\delta_z$ in the direction $U_\perp(z)$. 
\end{rem}
These remarks imply that we  can exploit Otto's parametrization 
(cf. (\ref{otto11})), and represent $U_\perp ^i(z)\,\nabla_i^{(z)}\,p^{(\omega)} _t(\cdot \,  , z)$, or $U_\perp ^i(z)\,\nabla_i^{(z)}\,\ln\,p^{(\omega)} _t(\cdot \,  , z)$, as (the gradient of) a scalar potential, $\widehat{\psi}_{(t,z,U_\perp )}\in C^{\infty }(M,\mathbb{R})$, according to the
\begin{prop}
\label{CNprop}
For each fixed $t>0$, and for any $U_\perp \in C^{\infty}(M, TM)$, the elliptic partial differential equation,  
\begin{equation}
{div}_{\omega}^{(y)}\,\left(p^{(\omega)} _t(y , z)\,\nabla^{(y)}\,\widehat{\psi}_{(t,z,U_\perp )}(y) \right)\, 
=\,-\,U_\perp(z)\cdot \,\nabla^{(z)}\,p^{(\omega)} _t(y , z)\;,
\label{carlo12H}
\end{equation} 
admits a unique solution $\widehat{\psi}_{(t,z,U_\perp )}$\,$\in\, C^{\infty }(M,\mathbb{R})$, with $\int_M\,\widehat{\psi}_{(t,z,U_\perp )}\,d\omega\,=0\,$,  smoothly depending on the data $t,z,\,U_\perp$,  and such that $\nabla^{(y)}\,\widehat{\psi}_{(t,z,U_\perp )}(y)\,\not\equiv \,0$\, if\, $U_\perp  \,\not= \,0$.
\end{prop}
\begin{proof} This is a rather obvious adaptation of a similar statement in \cite{Carlo}, Prop. 3.1. For its relevance in what follows, and for the convenience of the reader, we outline the proof in our case. According to the definition (\ref{weightdiv}) of the weighted divergence  and the positivity of the heat kernel we have   
\begin{eqnarray}
&&\triangle _g^{(y)}\,\widehat{\psi}_{(t,z,U_\perp )}(y)\,-\,
\nabla^{(y)} \left(f(y)-\ln\,p^{(\omega)} _t(y , z)\,\right)\cdot \nabla^{(y)}\,\widehat{\psi}_{(t,z,U_\perp )}(y)\nonumber\\
\nonumber\\
&&=\,-\,U_\perp(z)\cdot \,\nabla^{(z)}\,\ln\,p^{(\omega)} _t(y , z)\;.
\label{pderewr}
\end{eqnarray}
It follows that (\ref{carlo12H})  can be equivalently rewritten as  
\begin{equation}
\label{carlo12Hnew}
\bigtriangleup _{p _t\,(d\omega)}^{(y)}\,\widehat{\psi}_{(t,z,U_\perp )}(y)\,=
\,-\,U_\perp(z)\,\cdot \nabla^{(z)}\,\ln\,p^{(\omega)} _t(y , z)\;,
\end{equation}
where 
\begin{equation}
\label{heatweight}
\bigtriangleup _{p _t\,(d\omega)}^{(y)}\,:=\,
\bigtriangleup^{(y)}\,\,-\,\nabla ^{(y)}\,\left(f(y)-\ln\,p^{(\omega)} _t(y , z)\,\right)\,
\cdot \,\nabla ^{(y)}\,
\end{equation}
is the weighted Laplacian associated with the measure $p^{(\omega)} _t(y , z)d\omega(y)$. Solutions $\widehat{\psi}_{(t,z,U_\perp )}$ of (\ref{carlo12Hnew}) are naturally defined modulo an additive constant.  To remove this redundancy, we normalize $\widehat{\psi}_{(t,z,U_\perp )}$ by 
requiring
\begin{equation}
\int_M\,\widehat{\psi}_{(t,z,U_\perp )}(y)\,d\omega(y)=\,0\;.
\end{equation}
For any 
given $t>0$, denote by $\mathcal{H}^{1}(M,\mathbb{R};\,p_t(d\omega))$ the Sobolev space of  functions which together their gradients are square summable with respect to the heat kernel 
measure $p_t(d\omega)$, \,(see remark \ref{notatremark}).  Assume 
that  $\widehat{\psi}_{(t,z,U_\perp )}\in \mathcal{H}^{1}(M,\mathbb{R};\,p_t(d\omega))$,  and  write (\ref{carlo12Hnew}) distributionally as
\begin{eqnarray}
\label{Cdistrib}
&&\int_{M_y}\, \left(\nabla ^{(y)}_i\,\widehat{\psi}_{(t,z,U_\perp )}(y)\,\nabla^i\,\chi (y)\right)\;p^{(\omega)} _t(y , z)d\omega(y)\\
\nonumber\\
&&=\,\int_{M_y}\,\left( U_\perp(z)\,\cdot \nabla^{(z)}\,\ln\,p^{(\omega)} _t(y , z)\,\chi (y)\right)\;p^{(\omega)} _t(y , z)d\omega(y)\nonumber\;
\end{eqnarray}
for any test function $\chi\,\in\,C^\infty_0(M,\mathbb{R})\subset W_0^{1,2}(M,\mathbb{R})$, (by density). According to (\ref{fredzero}),\, $U_\perp(z)\,\cdot \nabla^{(z)}\,\ln\,p^{(\omega)} _t(y , z)$ is $L^2(M,\mathbb{R};\,p_t(d\omega))$--orthogonal to the constant functions and by the Fredholm alternative it follows that  (\ref{Cdistrib}), and hence (\ref{carlo12Hnew}), has a unique solution $\widehat{\psi}_{(t,z,U_\perp )}$ in $\mathcal{H}^{1}(M,\mathbb{R};\,p_t(d\omega))$. Standard elliptic regularity then implies that $\widehat{\psi}_{(t,z,U_\perp )}$\,$\in\, C^{\infty }(M,\mathbb{R})$, with a smooth dependence on the data $t,z,\,U_\perp$. In order to prove that if $U_\perp  \,\not= \,0$ then  $\nabla _{i}^{(y)}\,\widehat{\psi}_{(t,z,U_\perp )}(y)\,\not\equiv \,0$\, we exploit an induced heat equation associated to the elliptic problem  (\ref{carlo12H}). From 
\begin{eqnarray}
\frac{\partial}{\partial t}\,\left(U_\perp(z)\cdot \,\nabla^{(z)}\,p^{(\omega)} _t(y , z) \right)\,&=&\,U_\perp(z)\cdot \,\nabla^{(z)}\,\bigtriangleup _{\omega}^{(y)}\,p^{(\omega)} _t(y , z)\\
&=&\,\bigtriangleup _{\omega}^{(y)}\,\left(U_\perp(z)\cdot \,\nabla^{(z)}\,p^{(\omega)} _t(y , z) \right)\nonumber\;,
\end{eqnarray} 
and  (\ref{carlo12H}) we get that ${div}_{\omega}^{(y)}\,\left(p^{(\omega)} _t(y , z)\nabla^{(y)}\,\widehat{\psi}_{(t,z,U_\perp )}(y)\right)$ satisfies the heat equation
\\
\begin{equation}
\label{zlimit}
\left(\frac{\partial}{\partial t}\,-\,\bigtriangleup _{\omega}^{(y)} \right)\,\left[ 
{div}_{\omega}^{(y)}\,\left(p^{(\omega)} _t(y , z)\nabla^{(y)}\,\widehat{\psi}_{(t,z,U_\perp )}(y)\right)\right]\,=0\;,
\end{equation}
with
\begin{equation}
\lim_{t\searrow 0}\,{div}_{\omega}^{(y)}\,\left(p^{(\omega)} _t(y , z)\nabla^{(y)}\,\widehat{\psi}_{(t,z,U_\perp )}(y)\right)\,=\,{div}_{\omega}^{(z)}\,U_\perp(z)\;,\nonumber
\end{equation}
in the weak sense. 
Hence, if $\nabla^{(y)}\,\widehat{\psi}_{(t,z,U_\perp )}\,\equiv \,0$ then, by the (backward) uniqueness of the heat flow,  it follows that we must have ${div}_{\omega}^{(z)}\,U_\perp(z)\,\equiv \,0$,\;$\forall \,U_\perp \in C^{\infty}(M,TM)$. This necessarily implies $U_\perp\,\equiv \,0$.  
\end{proof}
If we denote by   
\begin{equation}
\label{HilbertTang}
\mathcal{H}_{t,\,z}(TM)\,:=\, \overline{\left\{\nabla \widehat{\psi}\,\in\,C^{\infty}(M,TM) \;:\;\widehat{\psi} \in C^{\infty}(M,\,\mathbb{R} )\right\}}^{\;L^{2}(p_t(d\omega,z))}\; ,
\end{equation}
the Hilbert space of gradient vector fields obtained by completion  with respect to the Otto  $L^{2}(p_t(d\omega, z))$ norm 
\begin{eqnarray}
\label{TGnorm}
\parallel\nabla \widehat{\psi}\parallel_{\mathcal{H}_{t,\,z}} ^2\,:=\,
\int_{M}\,\left|\nabla^{(y)}\,\widehat{\psi}\right|_{g(y)}^2
\,p^{(\omega)} _t(y ,\,z)\,d\omega(y)\;,
\end{eqnarray}
then we have
\begin{lem}
The map
\begin{eqnarray}
\label{tangmappsi}
T_z\,M\;&\longrightarrow&\,\;\;\; T_{p_t(d\omega)}\,\rm{Prob}(M, g)\,\overset{\simeq }{\longrightarrow }\,\;\;\;\;\;
\mathcal{H}_{t,\,z}(TM)\\
\nonumber\\
U_\perp(z)\,&\longmapsto&\;U_\perp(z)\,\cdot \nabla^{(z)}\,\ln\,p^{(\omega)} _t(y , z)\,\longmapsto \, \nabla \widehat{\psi}_{(t,z,U_\perp )}\;,\nonumber
\end{eqnarray}
is, for any $t\in (0,\infty)$, an injection. 
\end{lem}
\begin{proof}
This is an immediate consequence of Proposition \ref{CNprop}.
\end{proof}
The heat kernel  parametrization (\ref{tangmappsi}) can be applied to any adapted vector field $W=W_\perp +W_\parallel$, and we shall set 
\begin{equation}
\label{vectheinjmap}
W_t^a(y,\upsilon)\,:=\,\left(\nabla ^{i}_{(y)}\,\widehat{\psi}_{(t,z,W_\perp )}(y)\,,\; W^\alpha_\parallel(\upsilon) \right)\;,
\end{equation}
with an obvious notation.

\subsection{The induced G-M metric rescaling $(M,g)\mapsto (M,\,g_t^{(\omega)})$}
 According to (\ref{tangmappsi}), and in the spirit of Otto's formal Riemmanian calculus \cite{otto3},  we can interpret $\nabla \widehat{\psi}_{(t,z,U_\perp )}$ as the push--forward of $U_\perp \in T_z\,M$ to the tangent space $T_{p_t(d\omega)}\,\rm{Prob}(M,\,g)$, under the heat kernel embedding map (\ref{warpembed}). This remark 
motivated a basic observation by N. Gigli and C. Mantegazza which we extend to the weighted heat kernel embedding according to
\begin{prop}(cf. Def. 3.2 and Prop. 3.4 of \cite{Carlo}).
\label{CarlNicProp}
For any $t>0$,\, $z\in M$, and $U_\perp $,\,$W_\perp $\,$\in\,T_{z}\,M$, let us denote by $\nabla _{(y)}\,\widehat{\psi}_{(t,z,U_\perp )}$ and $\nabla_{(y)}\,\widehat{\psi}_{(t,z,W_\perp )}$ the corresponding vector fields defined by the weighted heat kernel injection map (\ref{tangmappsi}). Then, the symmetric bilinear form
\\
\begin{eqnarray}
\label{gt}
&&\;\;\;\;\;\;\;\;\;z\,\longmapsto \,g_t^{(\omega)}\,\left(U_\perp (z),W_\perp (z)\right)\\
\nonumber\\
&&\,:=\,\int_{M}\,g_{ik}(y)\,\nabla ^{i}_{(y)}\,\widehat{\psi}_{(t,z,U_\perp )}\,\nabla ^{k}_{(y)}\,\widehat{\psi}_{(t,z,W_\perp )}\,p^{(\omega)} _t(y , z)\,d\omega(y)\;,\nonumber
\end{eqnarray}
\\
\noindent
defines a scale--dependent metric tensor  on $M$, varying smoothly in $0\,<\,t\,<\, \infty$.
\end{prop}
\begin{proof}   We briefly outline the proof, an obvious adaptation of Prop. 3.4 of \cite{Carlo}. To begin with, let us observe that as a consequence of the uniqueness property of the defining pde (\ref{carlo12H}), the functions $\widehat{\psi}_{(t,z,U_\perp )}$ and $\widehat{\psi}_{(t,z,W_\perp )}$ depend linearly from the vectors $U_\perp$ and $W_\perp$. This implies that the expression (\ref{gt}) is bilinear, besides being manifestly symmetric and non--negative. Moreover, the smoothness of the weighted heat kernel $p^{(\omega)} _t\,d\omega$, and the smooth dependence of $\widehat{\psi}_{(t,z,U_\perp )}$ and $\widehat{\psi}_{(t,z,W_\perp )}$   from the data $(t,z,U_\perp, W_\perp)$, imply that  $g_t^{(\omega)}\,\left(U_\perp (z),W_\perp (z)\right)$ is smooth in its arguments. Since $p^{(\omega)} _t\,d\omega>0$, if we assume that $g_t^{(\omega)}\,\left(U_\perp (z),U_\perp (z)\right)\,=\,0$ then we necessarily get $\nabla_{(y)}\,\widehat{\psi}_{(t,z,U_\perp )}\,\equiv \,0$ and hence, according to proposition  \ref{CNprop},
$U_\perp (z)\,\equiv \,0$. It follows that the bilinear form (\ref{gt}) is a smooth metric on $M$. 
\end{proof}
 To characterize geometrically the metric (\ref{gt})  let us
 consider a one--parameter family of locally Lipschitz diffeomorphisms $\Theta_\lambda$ in $M$,
 \begin{eqnarray}
 \Theta _\lambda\,:\,[0,1)\times\,M\,\longrightarrow \,M\\
 (\lambda,\,x)\,\longmapsto\, c_\lambda(x)\;, \nonumber
 \end{eqnarray}
where $c_\lambda(x)$ is the point in $M$ reached at time $\lambda$ along the $\Theta_\lambda$--trajectory issued from $x$, and where $\Theta_0=id_M$. Let ${c\,'}_\lambda\,( x)$ be  the $\lambda$--dependent velocity field associated, for almost every $\lambda$, with the trajectories $(\lambda, x) \mapsto c_\lambda(x)$  of $\Theta_\lambda$. Locally Lipschitz diffeomorphisms map null sets to null sets, and the behavior of the heat kernel embedding along $\Theta _\lambda$ is described by the 

\begin{lem} (cf. Th.2.5 in \cite{Carlo}).
\label{movHeatSource}
  Let $\delta _{c_\lambda(z)}$ denote the heat source at the point $c_\lambda(z)\in M$ reached at time $\lambda$ along the $\Theta_\lambda$--trajectory issued from $z$. Let  $c_\lambda(z)\longmapsto p^{(\omega)} _t(\cdot \, , c_\lambda(z))d\omega(\cdot )$ be the corresponding heat embedding map.  Then, for almost every $\lambda\in [0,1]$, there exists, along the path $p^{(\omega)} _t(\cdot \, , c_\lambda(z))d\omega(\cdot )$,\,$\lambda\in [0,1]$, \,in $\rm{Prob}_{ac}(M,g)$, a tangent velocity field  $\mathfrak{v} _\lambda \in \mathcal{H}_{t,c_\lambda(z)}(TM)$ such that the relation   
\begin{equation}
\label{cpde}
c'_\lambda(z)\,\cdot \nabla _{c_\lambda(z)}\,p^{(\omega)} _t(y , c_\lambda(z))\,+\,div_\omega^{(y)}\,
\left( \mathfrak{v} _\lambda(y)\,p^{(\omega)} _t(y , c_\lambda(z))\right)\,=\,0\;,
\end{equation}
holds in the sense of distributions.
\end{lem}
\begin{proof}  
Since $\Theta _\lambda$ is locally Lipschitz, the curve of measures, (at fixed $t$), $\lambda\longmapsto p^{(\omega)} _t(\cdot \, , c_\lambda(z))d\omega(\cdot )$ is absolutely continuous with respect to the Wasserstein  distance. By a result of Ambrosio, Gigli and Savar\'e (cf. Th. 8.3.1 in \cite{savare} and Th. 13.8 in \cite{VillaniON}),  for almost every $\lambda$  the path $p^{(\omega)} _t(\cdot \, , c_\lambda(z))d\omega(\cdot )$,\,$\lambda\in [0,1]$, \,in $\rm{Prob}_{ac}(M,g)$ admits a tangent velocity field $y\longmapsto \mathfrak{v} _\lambda (y) \in \mathcal{H}_{t,c_\lambda(z)}(TM)$, where 
$\mathcal{H}_{t,c_\lambda(z)}(TM)$ is the Hilbert space (\ref{HilbertTang}) associated with  $c_\lambda(z)$. Moreover the continuity equation
\begin{equation}
\label{diffcont}
\mathcal{L}_{(\tfrac{\partial }{\partial \lambda}+\mathfrak{v} _\lambda)}\,p^{(\omega)} _t(y , c_\lambda(z))\,d\omega(y)\,d\lambda\,=\,0
\end{equation}
holds in the distributional sense on  $M\times[0,1]$,\emph{i.e.}
\begin{equation}
\label{Liecont2}
\int_M\,\int_{[0,1]}\,\left(\frac{\partial }{\partial \lambda}\,\varphi (\lambda,y)\,+\,
\mathfrak{v} _\lambda(y)\cdot \nabla^{(y)} \varphi (\lambda,y)  \right)\,p^{(\omega)} _t(y , c_\lambda(z))\,d\omega(y)\,d\lambda\,=\,0\;,
\end{equation}
for all $\varphi\in C^\infty_0({M}\times[0,1])$. Here we have denoted by $\mathcal{L}_{(\tfrac{\partial }{\partial \lambda}+\mathfrak{v} _\lambda)}$ the (weakly defined) Lie derivative in the direction of the $M\times[0,1]$--vector field $\tfrac{\partial }{\partial \lambda}+\mathfrak{v} _\lambda$. Hence, for almost every $\lambda$, we can write  
\begin{eqnarray}
\label{Liecont}
&&\mathcal{L}_{(\tfrac{\partial }{\partial \lambda}+\mathfrak{v} _\lambda)}\,p^{(\omega)} _t(y , c_\lambda(z))\,d\omega(y)\,d\lambda\\
\nonumber\\
&&=\,\frac{\partial }{\partial \lambda}\,p^{(\omega)} _t(y , c_\lambda(z))\,d\omega(y)\,d\lambda\,+
\,\mathcal{L}_{\mathfrak{v} _\lambda}\left(p^{(\omega)} _t(y , c_\lambda(z))\,d\omega(y)\,d\lambda\,  \right)\nonumber\\
\nonumber\\
&&=\,\left[\frac{\partial }{\partial \lambda}\,p^{(\omega)} _t(y , c_\lambda(z))\,+\,div_\omega^{(y)}\,
\left( \mathfrak{v} _\lambda\,p^{(\omega)} _t(y , c_\lambda(z))\right)\right]\,d\omega(y)\,d\lambda\,=\,0\;.\nonumber
\end{eqnarray}
  Along the curve $\lambda\mapsto c_\lambda(z)$,\, we have
\begin{equation}
\frac{\partial }{\partial \lambda}\,p^{(\omega)} _t(y , c_\lambda(z))=c'_\lambda(z)\,\cdot \nabla _{c_\lambda(z)}\,p^{(\omega)} _t(y , c_\lambda(z))\;,
\end{equation}
which inserted in (\ref{Liecont}), implies that the velocity vector  $\mathfrak{v} _\lambda$ satisfies (\ref{cpde}), as stated.
\end{proof}  
As a direct consequence of Lemma \ref{movHeatSource}  we have the following result (cf. Prop. 3.5 of \cite{Carlo}), which can be interpreted as a form of equivariance of the heat kernel embedding under (Lipschitzian) diffeomorphisms, 
\begin{lem}
\label{velident} 
For almost every $\lambda$ the velocity field $y\longmapsto \mathfrak{v} _\lambda(y)$ can be represented  as the $L^2([0,1]\times{M},\,\nu_\lambda\otimes {d\lambda})$ vector field 
$(\lambda,\,z)\,\mapsto \,\nabla^{(y)}\, \widehat{\psi}_{(t,c_\lambda(z), c'_\lambda)}(y)$  covering the 
curve $\lambda\mapsto \nu_\lambda:= p^{(\omega)} _t(\cdot \, , c_\lambda(z))d\omega(\cdot )$ in $\rm{Prob}_{ac}(M,g)$.
\end{lem}
\begin{proof}
If we compare (\ref{cpde})  with the elliptic PDE (\ref{carlo12H}), characterizing the scalar potential $\widehat{\psi}_{(t,c_\lambda(z),c'_\lambda )}$ associated with the vector $U_\perp(z)\,\equiv \,c'_\lambda(z)$, then, for almost every $\lambda$, we get
\begin{eqnarray}
&&{div}_{\omega}^{(y)}\,\left(p^{(\omega)} _t(y , c_\lambda(z))\,\nabla^{(y)}\,\widehat{\psi}_{(t,c_\lambda(z), c'_\lambda)}(y) \right)\\
\nonumber\\
&& =\,-\,c'_\lambda(z)\cdot \,\nabla^{(c_\lambda(z))}\,p^{(\omega)} _t(y , c_\lambda(z))\nonumber\\
\nonumber\\
&&=\,div_\omega^{(y)}\,
\left(\mathfrak{v} _\lambda(y)\,p^{(\omega)} _t(y , c_\lambda(z))\right)\;.\nonumber
\end{eqnarray}
Since the solution of (\ref{carlo12H}) is  unique, we have that for almost every $\lambda$  we can write $\widehat{\psi}_{(t,c_\lambda(z),c'_\lambda )}(y)\,=\,\mathfrak{v} _\lambda(y)$, as stated.
\end{proof}
The above results imply the following property that, in line with Lemma \ref{McClemma},  extends to $(M,g,\,d\omega)$ a basic observation of Gigli--Mantegazza.

\begin{prop} (cf. Prop. 3.5 of \cite{Carlo}).
\label{wassdinter}
The heat kernel induced  
metric (\ref{gt}) 
 can be identified with the (squared) norm of the Wasserstein metric speed of the absolutely continuous curves of measures $\lambda\,\longmapsto \,\nu_\lambda:= p^{(\omega)} _t(\cdot \, , c_\lambda(z))d\omega(\cdot )$\, $\in\left(\Upsilon_t(M),\,d_g^W  \right)$,\\
\begin{eqnarray}  
&& g_t^{(\omega)}\,\left(c'_\lambda (z),\,c'_\lambda (z)\right):=\,
\left\langle  \nabla\,\widehat{\psi}_{(t,c_\lambda(z),U_\perp )},\,\nabla\,\widehat{\psi}_{(t,c_\lambda(z),U_\perp )}  \right\rangle _{(g,p_t(d\omega))}\nonumber\\
\nonumber\\
&&=\left| \frac{d{\nu}_\lambda(z, t)}{d\lambda}  \right|^2\,:=\,\left[\lim_{\epsilon \rightarrow 0}\,\frac{d^W_g\left(\nu_{\lambda+\epsilon }(z, t),\,\nu_\lambda(z, t)  \right)}{\epsilon } \right]^2\;,
\label{meaning}
\end{eqnarray} 
where $\Upsilon_t(M)$ denotes the image of $(M,g)$ in $\left(\rm{Prob}_{ac}(M,g),\,d_g^W  \right)$ under the heat kernel embedding map (\ref{warpembed}).
\end{prop}
\begin{proof}
Let $K_f\in \mathbb{R}$ denote the lower bound of the Bakry--Emery Ricci curvature of $(M,g,\,d\omega)$,
\begin{equation}
\label{BKE}
Ric_g(v,v)\,+\,Hess\,f(v,v)\,\geq\,K_f\,g(v,v)\;,\;\;\;\;\forall \,v\in\,TM\;.
\end{equation}
According to a result  of D. Bakry, I. Gentil, and M. Ledoux, (\cite{bakry2} Corollary 4.2), we have, for any given $t>0$, 
\begin{eqnarray}
\label{gentil}
&&d_g^W\left(p^{(\omega)} _t(\cdot \, , c_{\lambda_1}(z)),\,p^{(\omega)} _t(\cdot \, , c_{\lambda_2}(z)) \right)\\
\nonumber\\
&&\leq\,e^{\,-\,K_f\,t}\,d_g^W(\delta_{c_{\lambda_1}(z)},\,\delta_{c_{\lambda_2}(z)})
\,=\,e^{\,-\,K_f\,t}\,d_g({c_{\lambda_1}(z)},\,{c_{\lambda_2}(z)})\;.\nonumber
\end{eqnarray}
Since the trajectories of $\Theta_\lambda$ are absolutely continuous, (the family of diffeomorphisms $\Theta_\lambda$ is locally Lipschitz),\, (\ref{gentil}) implies that the curves in $\rm{Prob}_{ac}(M,g)$ defined by
\begin{equation}
\lambda\,\longmapsto \,\nu_{\lambda}(z, t)\,:=\, p^{(\omega)} _t(\cdot \, , c_\lambda(z))d\omega(\cdot )
\end{equation}
are locally absolutely continuous in the metric topology induced by the Wasserstein distance $d^W_g$, and we can define, for almost every $\lambda\,\in\,[0,1]$, the metric derivative \cite{savare} of $\nu_{\lambda}$,
\begin{equation}
\label{metricvel}
\left| \frac{d{\nu}_\lambda(z, t)}{d\lambda}  \right|\,:=\,\lim_{\epsilon \rightarrow 0}\,\frac{d^W_g\left(\nu_{\lambda+\epsilon }(z, t),\,\nu_\lambda (z, t)  \right)}{\epsilon }\;.
\end{equation}
Absolute continuity of the $\Theta^{\,\star }_\lambda$ trajectories, also implies\footnote{By an elementary transposition of Th. 2.5 and Proposition 3.5 in \cite{Carlo}.} that the vector field $\nabla\, \widehat{\psi}_{(t,c(z), c'_\lambda)}$ is uniquely determined  by (\ref{Liecont2}), and its norm  in the formal  Riemannian structure defined on $T_{\nu_\lambda}\,\rm{Prob}_{ac}(M,g)$  by the $L^2(M, \nu_\lambda)$  Otto inner product is provided by  
\begin{eqnarray}
\label{velsquare}
&&\left| \frac{d{\nu}_\lambda(z, t)}{d\lambda}  \right|^2\,=\,\int_M\,\left|\nabla_{y}\, \widehat{\psi}_{(t,\,c(z),\, c'_\lambda)}(y)  \right|^2_g\,p^{(\omega)} _t(y,\, c_\lambda(z))\,d\omega(y)\\
\nonumber\\
&&=\,g_t^{(\omega)}\,\left(c'_\lambda (z),\,c'_\lambda (z)\right)\;,
\end{eqnarray}
as stated.
\end{proof} 
Let us consider the curves of  measures
\begin{equation}
\lambda\,\longmapsto\,\nu_\lambda(c,t)\,:=\, p^{(\omega)} _t(\cdot \, , c_\lambda)d\omega(\cdot )\;,
\end{equation}
associated to  the set of all (absolutely continuous) 
curves, $\{c_{(z,x)}\}:=\,\{[0,1]\ni \lambda\mapsto c_\lambda\in M\}$, \; $c_0=z$, \, $c_1=x$, connecting the points $z$ and $x$ in $(M, g)$. If we let 
\begin{equation}
\label{riemdist}
 d_{g_t^{(\omega)}}(x,z)\,:=\,\underset{\{c_{(z,x)}\}}{inf}\;\int_0^1\,\sqrt{g_t^{(\omega)}\,\left(c'_\lambda,\,c'_\lambda\right)}\,d\lambda 
\end{equation}
denote the Riemannian distance in $(M,g_t^{(\omega)})$, then (\ref{meaning}) 
 implies
\begin{equation}
\label{distanT}
d_{g_t^{(\omega)}}(x,z)\,:=\,\underset{\{c_{(z,x)}\}}{inf}\;\int_0^1\,\left| \frac{d{\nu}_\lambda(c, t)}{d\lambda}  \right|\,d\lambda\;,
\end{equation}
which characterizes metrically  the manifold $(M,\,g_t^{(\omega)})$.\\
\\
\noindent 
To prove that $(t, g)\longmapsto{g_t^{(\omega)}}$ is a variation (and possibly a geometrical deformation) of the original Riemannian structure we need to show that $\lim_{\,t\searrow 0}\,g_t^{(\omega)}\,=\,g$. This is a singular limit and requires  some care. For the standard heat kernel embedding the argument is tersely presented in \cite{Carlo}, and can be easily adapted to our particular case according to
\begin{lem} (cf. Prop. 3.7 of \cite{Carlo}).
\label{metricatzero} 
\begin{equation}
\label{limitmetric}
\lim_{t\searrow  0}\,g_t^{(\omega)}\,\left(c'_\lambda,\,c'_\lambda\right)\,=\,
g\,\left(c'_\lambda,\,c'_\lambda\right)\;.
\end{equation}
\end{lem}
\begin{proof}
From (\ref{gentil}) 
and (\ref{metricvel}) we get
\begin{eqnarray}
\label{}
&&\left| \frac{d{\nu}_\lambda(c, t)}{d\lambda}  \right|\,:=\,\lim_{\epsilon \rightarrow 0}\,\frac{d^W_g\left(\nu_{\lambda+\epsilon }(c, t),\,\nu_\lambda(c, t)  \right)}{\epsilon }\\
\nonumber\\
&&\, \leq\,e^{\,-\,K_f\,t}\,
\lim_{\epsilon \rightarrow 0}\,\frac{d_g({c_{\lambda+\epsilon}},\,{c_{\lambda}})}{\epsilon }\;,
\end{eqnarray}
which, according to (\ref{distanT}), implies
\begin{equation}
\label{upmetric}
d_{g_t^{(\omega)}}(x,z)\,\leq\,e^{\,-\,K_f\,t}\,d_{g}(x,z)\;.
\end{equation}
Let us observe that the right member of (\ref{distanT}) provides the intrinsic Wasserstein distance between $p^{(\omega)} _t(\cdot\,,z)d\omega$  and $p^{(\omega)} _t(\cdot\,,x)d\omega$ on the image $\Upsilon _t(M)\subset \rm{Prob}_{ac}(M,g)$ of $(M,g,\,d\omega)$. In general, this is larger than the actual Wasserstein distance between $p^{(\omega)} _t(\cdot\,,z)d\omega$  and
$p^{(\omega)} _t(\cdot\,,x)d\omega$ in $\rm{Prob}_{ac}(M,g)$ which is defined by
\begin{equation}
d^W_g\left(p^{(\omega)} _t(\cdot\,,z)d\omega,\,p^{(\omega)} _t(\cdot\,,x)d\omega\right)\,:=\,\underset{\{\widehat{c}_{(z,x)}\}}{inf}\;\int_0^1\,\left| \frac{d\, \widehat{c}_{(z,x)}}{d\lambda}  \right|\,d\lambda\;,
\end{equation}
where the $inf$ is over all absolutely continuous curves of probability measure, $\lambda\longmapsto \widehat{c}_{(z,x)}(\gamma) \in \rm{Prob}_{ac}(M,g)$,  connecting
 $p^{(\omega)} _t(\cdot\,,z)d\omega$ to
$p^{(\omega)} _t(\cdot\,,x)d\omega$. Hence,    
from  (\ref{distanT}) and (\ref{upmetric}), we get
\begin{equation}
\label{uplowmetric}
d^W_g\left(p^{(\omega)} _t(\cdot\,,z)d\omega,\,p^{(\omega)} _t(\cdot\,,x)d\omega\right)\,\,\leq\,d_{g_t^{(\omega)}}(x,z)\,\leq\,e^{\,-\,K_f\,t}\,d_{g}(x,z)\;.
\end{equation}
 The upper bound in (\ref{uplowmetric}) immediately  provides $\limsup_{t\rightarrow  0}\,g_t^{(\omega)}\,\leq\,g$. A similar control 
on $\liminf_{t\rightarrow  0}\,g_t^{(\omega)}$ is less direct since the lower bound  in (\ref{uplowmetric}) is expressed in terms of the \emph{secant} Wasserstein distance between $p^{(\omega)} _t(\cdot\,,z)d\omega$  and $p^{(\omega)} _t(\cdot\,,x)d\omega$. To circumvent this, one equivalently characterizes 
the metric $g_t^{(\omega)}$,  
\begin{equation}
g_t^{(\omega)}\,\left(c'_\lambda,\,c'_\lambda\right)\,=\,
\,\int_M\,\left|\nabla_{y}\, \widehat{\psi}_{(t,\,c_\lambda,\, c'_\lambda)}(y)  \right|^2_g\,p^{(\omega)} _t(y,\, c_\lambda)\,d\omega(y)\;,
\end{equation}
by going to its variational description defined by \cite{Carlo}
\begin{eqnarray}
&&\;\;\;\;\;\;\;\;\;g_t^{(\omega)}\,\left(c'_\lambda,\,c'_\lambda\right)\\
\nonumber\\
&&\,=\,\sup_{\varphi \in C_0^\infty(M,\mathbb{R})}\,\left\{2\,\,\int_M\,\nabla_y\varphi\cdot \nabla_{y}\, \widehat{\psi}_{(t,\,c_\lambda,\, c'_\lambda)}(y)\,p^{(\omega)} _t(y,\, c_\lambda)\,d\omega(y)\right.\nonumber\\
\nonumber\\
&&\left.-
\,\int_M\,\left|\nabla_{y}\, {\varphi}(y)  \right|^2_g\,p^{(\omega)} _t(y,\, c_\lambda)\,
d\omega(y)\right\}\;,\nonumber
\end{eqnarray}
which implies  
\begin{eqnarray}
&&g_t^{(\omega)}\,\left(c'_\lambda,\,c'_\lambda\right)\,\geq\,2\,\,\int_M\,\nabla_y\varphi\cdot \nabla_{y}\, \widehat{\psi}_{(t,\,c_\lambda,\, c'_\lambda)}(y)\,p^{(\omega)} _t(y,\, c_\lambda)\,d\omega(y)\nonumber\\
\nonumber\\
&&-
\,\int_M\,\left|\nabla_{y}\, {\varphi}(y)  \right|^2_g\,p^{(\omega)} _t(y,\, c_\lambda)\,
d\omega(y)\;,
\label{eqnA}
\end{eqnarray}
for any $\varphi \in C_0^\infty(M,\mathbb{R})$.
Writing, as usual, $\widehat{\psi}_{(t)}:=\widehat{\psi}_{(t,\,c_\lambda,\, c'_\lambda)}$ for ease of notation we have, 
by an obvious transposition of the argument provided in \cite{Carlo}, Prop. 3.7,
\begin{eqnarray}
\label{unaltr}
&&\int_M\,\nabla_y\varphi\cdot \nabla_{y}\, 
\widehat{\psi}_{(t)}(y)\,p^{(\omega)} _t(y,\, c_\lambda)\,d\omega(y)\\
&&=\,
\int_M\,\varphi(y)\,div^{(y)}_\omega \left(\nabla_{y}\, \widehat{\psi}_{(t)}(y)\right)\,p^{(\omega)} _t(y,\, c_\lambda)\,d\omega(y)\nonumber\\
&&=\, 
\int_M\,\varphi(y)\,c'_\lambda\cdot\,\nabla_{c_\lambda}\,p^{(\omega)} _t(y,\, c_\lambda)\,d\omega(y)\;,
\nonumber
\end{eqnarray}
where, in the last line, we exploited  the pde (\ref{carlo12H})  
defining the potential $\widehat{\psi}_{(t,\,c_\lambda,\, c'_\lambda)}$ corresponding to the tangent vector $c'_\lambda$.
The relation (\ref{unaltr}), can be equivalently rewritten as
\begin{eqnarray}
\label{eqnB}
&&\;\;\;\;\;\;\int_M\,\nabla_y\varphi\cdot \nabla_{y}\, 
\widehat{\psi}_{(t)}(y)\,p^{(\omega)} _t(y,\, c_\lambda)\,d\omega(y)\\
&&=\,
c'_\lambda\cdot\,\nabla_{c_\lambda}\,\int_M\,\varphi(y)\,
p^{(\omega)} _t(y,\, c_\lambda)\,d\omega(y)\,=\,c'_\lambda\cdot\,\nabla_{c_\lambda}\,\varphi_t(c_\lambda)\;,\nonumber
\end{eqnarray}
where $\varphi_t(c_\lambda)$ is the solution of the heat 
equation $(\tfrac{\partial}{\partial t}-\bigtriangleup^{(c_\lambda)} _\omega)\,\varphi_t=0$ 
with $\lim_{t\searrow 0}\,\varphi_t=\,\,\varphi$. From (\ref{eqnA}) and (\ref{eqnB}) we get 
\begin{eqnarray}
\label{reveqlt}
&&g_t^{(\omega)}\,\left(c'_\lambda,\,c'_\lambda\right)\,\geq\,2\,c'_\lambda\cdot\,\nabla_{c_\lambda}\,\varphi_t(c_\lambda)\\
&&-\,\int_M\,\left|\nabla_{y}\, {\varphi}(y)  \right|^2_g\,p^{(\omega)} _t(y,\, c_\lambda)\,
d\omega(y)\;,\nonumber
\end{eqnarray} 
which, by choosing a $\varphi\in C_0^\infty(M,\mathbb{R})$ with $\nabla\varphi\,=\,c'_\lambda$, 
yields
\begin{equation}
\liminf_{t\rightarrow  0}\,g_t^{(\omega)}(c'_\lambda,\,c'_\lambda)\,\geq\,g(c'_\lambda,\,c'_\lambda)\;.
\end{equation}
Together with (\ref{upmetric}),  this implies 
\begin{equation}
\label{limitmetric11}
\lim_{t\searrow  0}\,g_t^{(\omega)}\,\left(c'_\lambda,\,c'_\lambda\right)\,=\,
g\,\left(c'_\lambda,\,c'_\lambda\right)\;,
\end{equation}
as required.
\end{proof}

\subsection{Warping $(M,\,g_t^{(\omega)})$ on $M\times \mathbb{T}^q$}
 To extend the Gigli--Mantegazza construction to the warped manifold $M\times \mathbb{T}^q$, let us consider the solution $(t,f)\longmapsto \exp[-\,\tfrac{2f^{\,(\omega)}_t}{q}]$ of the heat equation 
\begin{equation}
\label{heatF} 
\left(\tfrac{\partial}{\partial t}-\bigtriangleup^{(z)} _\omega \right)\,
e^{-\,\tfrac{2f^{\,(\omega)}_t(z)}{q}}\,=0\;,\;\;\;\;\;\;\;t\in (0,\infty )\;,
\end{equation}
associated to the warping factor $e^{-\,\tfrac{2f}{q}}$ in the metric (\ref{newmetr0}). Note that we can
equivalently write (\ref{heatF}) as 
\begin{eqnarray}
\label{parabolicf}
&&\frac{\partial}{\partial t}\,f^{\,(\omega)}_t\,=\,\bigtriangleup^{(z)} _\omega\,f^{\,(\omega)}_t\,-\,
\tfrac{2}{q}\,|\nabla\,f^{\,(\omega)}_t|^2_g\\
\nonumber\\
&&\;\;\;\;\;\;\;\;\;\;\;\;\;\;=\,\bigtriangleup^{(z)} _g\,f^{\,(\omega)}_t\,-\,\,\nabla ^i\,f\,\nabla_i\,f^{\,(\omega)}_t
\,-\tfrac{2}{q}\,|\nabla\,f^{\,(\omega)}_t|^2_g
\;,\nonumber\\
\nonumber\\
&&\lim_{t\searrow 0}f^{\,(\omega)}_{t}(z)\,=\,f(y)\nonumber\;.
\end{eqnarray}
We have 
\begin{prop}
For any $t>0$,\, $(z,\zeta)\in M\times \mathbb{T}^q$ and $U$,\,$W$\,$\in\,T_{(z,\zeta)}\,M\times \mathbb{T}^q$, let $U_t(y,\upsilon)$ and $W_t(y,\upsilon)$ denote the vector fields defined by the  map (\ref{vectheinjmap}).  
Then
\begin{eqnarray}
&&\;\;\;\;\;\;h_t\left(U(z,\zeta ),W(z,\zeta )\right)\\
\nonumber\\
&&\,=\,g_t^{(\omega)}\,\left(U_\perp (z),W_\perp (z)\right)\,+\,e^{-\,\tfrac{2}{q}\,f^{\,(\omega)}_{t}(z)}\,\delta_{\alpha\beta}\,U_\parallel ^\alpha(\zeta)\,W_\parallel ^\beta(\zeta)\;,\nonumber
\end{eqnarray}
provides a scale--dependent metric tensor  on $M\times \mathbb{T}^q$, varying smoothly with $t\in (0,\infty)$.
\end{prop}
\begin{proof} In analogy with (\ref{gt}), let us consider the symmetric bilinear form
\begin{eqnarray}
\label{CarloMetricgf}
&&\;\;\;\;\;\;h_t\left(U(z,\zeta ),W(z,\zeta )\right)\\
\nonumber\\
&&\,:=\,\int_{M\times\mathbb{T}^q}\,h_{ab}^{(q)}(y,\upsilon)\,U^a_t(y,\upsilon)
W^b_t(y,\upsilon)\;p^{(\omega)} _t(y , z)\,d\omega(y)\otimes \,\delta_\zeta^{\;\mathbb{T}^q}\nonumber\;,
\end{eqnarray}
which can be seen as the pull--back, under the heat kernel embedding (\ref{warpembed}), of the Otto inner product in $T_{p_t(d\omega)\otimes \delta }\,\rm{Prob}(M\times\mathbb{T}^q,\,h)$,\, (see (\ref{inner})). 
 If we decompose the warped metric $h(\Phi_{(q)})$ according to (\ref{newmetr0}), then we can write 
\begin{eqnarray}
\label{firstht}
\\
&&h_t\left(U(z,\zeta ),W(z,\zeta )\right)\nonumber\\
\nonumber\\
&&=\,\int_{M\times\mathbb{T}^q}\,g_{ik}(y)\,\nabla ^{i}_{(y)}\,\widehat{\psi}_{(t,z,U_\perp )}\,\nabla ^{k}_{(y)}\,\widehat{\psi}_{(t,z,W_\perp )}\,p^{(\omega)} _t(y , z)\,d\omega(y)\,\otimes\, \delta_\zeta^{\;\mathbb{T}^q}\nonumber\\
\nonumber\\
&&+\,\int_{M\times\mathbb{T}^q}\,e^{-\,\tfrac{2f(y)}{q}}\,\delta_{\alpha\beta}\,U_\parallel ^\alpha(\upsilon)\,W_\parallel ^\beta(\upsilon)\,p^{(\omega)} _t(y , z)\,d\omega(y)\,\otimes \,\delta_\zeta^{\;\mathbb{T}^q}\nonumber\\
\nonumber\\
&&=\,\int_{M}\,g_{ik}(y)\,\nabla ^{i}_{(y)}\,\widehat{\psi}_{(t,z,U_\perp )}\,\nabla ^{k}_{(y)}\,\widehat{\psi}_{(t,z,W_\perp )}\,p^{(\omega)} _t(y , z)\,d\omega(y)\nonumber\\
\nonumber\\
&&+\,\delta_{\alpha\beta}\,U_\parallel ^\alpha(\zeta)\,W_\parallel ^\beta(\zeta)\,\int_{M}\,e^{-\,\tfrac{2f(y)}{q}}\,\,p^{(\omega)} _t(y , z)\,d\omega(y)\nonumber
\end{eqnarray}
\\
\noindent where we have exploited $\int_{\mathbb{T}^q}\,\delta_\zeta^{\;\mathbb{T}^q}\,=\,1$,\, and $\int_{\mathbb{T}^q}\,\delta_{\alpha\beta}\,U_\parallel ^\alpha(\upsilon)\,W_\parallel ^\beta(\upsilon)\,\delta_\zeta^{\;\mathbb{T}^q}\,=\,\delta_{\alpha\beta}\,U_\parallel ^\alpha(\zeta)\,W_\parallel ^\beta(\zeta)$. We immediately recognize in the last line the $t$--dependent metric tensor $g_t^{(\omega)}$ defined by (\ref{gt}).  
To cast  in a more explicit from also the term quadratic in $U_\parallel ^\alpha$ and $W_\parallel ^\beta(\zeta)$, let us exploit the symmetry of the heat 
kernel $p^{(\omega)} _t(y , z)\,=\,p^{(\omega)} _t(z , y)$ to compute, for any $t\,>0$,
\begin{eqnarray}
\label{backf}
&&\int_{M}\,e^{-\,\tfrac{2f(y)}{q}}\,\,p^{(\omega)} _t(y , z)\,d\omega(y)
\,=\,
\int_{M}\,e^{-\,\tfrac{2f(y)}{q}}\,p^{(\omega)} _t(z , y)\,d\omega(y)\\
\nonumber\\
&&\,=\,e^{t\,\bigtriangleup _\omega\,}\, \left( e^{-\,\tfrac{2f}{q}} \right)\,:=
e^{-\,\tfrac{2f^{\,(\omega)}_t(z)}{q}}\;,\nonumber
\end{eqnarray}
 where  $f^{\,(\omega)}_t(z)$ is the solution of the heat 
equation (\ref{heatF}). It follows that the symmetric bilinear form (\ref{firstht}) can be written  as
\begin{eqnarray}
&&\;\;\;\;\;\;\;h_t\left(U(z,\zeta ),W(z,\zeta )\right)\\
\label{hTmetric}\\
&&\,=\,g_t^{(\omega)}\,\left(U_\perp (z),W_\perp (z)\right)\,+\,e^{-\,\tfrac{2}{q}\,f^{\,(\omega)}_{t}(z)}\,\delta_{\alpha\beta}\,U_\parallel ^\alpha(\zeta)\,W_\parallel ^\beta(\zeta)\;,\nonumber
\end{eqnarray}
as stated. Finally, as in proposition \ref{CarlNicProp}, the properties of the weighted heat kernel (\ref{heatfloweight}), and the smooth dependence of the potential $\widehat{\psi}_{(t,z,U_\perp )}$ from the data $(t,z,U_\perp )$, imply that  $h_t$ is indeed a metric on $M\times\mathbb{T}^q$, which according to (\ref{limitmetric}) varies smoothly  with the scale $t\in [0, \infty)$. 
\end{proof}

\subsection{Harmonic energy rescaling and dilatonic action flow}
The deformation induced in $(M, g)$ and in $M\times\mathbb{T}^q$  by the heat kernel embedding generates a corresponding geometric rescaling of the Harmonic energy functionals (\ref{Msigmamod0}), (\ref{qharmwarp})  and of the associated dilatonic NL$\sigma$M action (\ref{eqnAct}). To describe the nature of such rescalling we shall consider explicitly  the functional $E[\phi,\,g]$.  The  
extension to the warped map $\Phi_{(q)}:(\Sigma, \gamma) \longrightarrow (M\times\mathbb{T}^q,\,h^{(q)})$  is a straightforward, and we limit ourselves to state how the relevant results generalize to  $E[\Phi_{(q)},\,h^{(q)}]$.\\
\\
\noindent
Let us consider the Hilbert bundle  $(p_t^{(\omega)}\circ \phi)^{-1} T_{p_t(d\omega)}\,\rm{Prob}(M, g)$,  covering the   map 
\begin{eqnarray}
\label{compmap}
\Sigma\,\simeq \mathbb{T}^2\,&\longrightarrow&\, (M,g,\,d\omega)\,\longrightarrow \, \left(\rm{Prob}_{ac}(M,g),\,d_g^W\right)\\
\nonumber\\
x\,&\longmapsto&\, \;\;\;\;\; \phi(x)\,\;\;\;\;\longmapsto\,\;\;\;\; p_t^{(\omega)}(\cdot ,\phi(x))\,d\omega(\cdot )\nonumber \;,
\end{eqnarray} 
and whose fiber over $x\in\Sigma$ is given by the Hilbert space (\ref{HilbertTang}) evaluated at $z=\phi(x)$,
\begin{equation}
\label{Hilbertfiber}
\left.(p_t^{(\omega)}\circ \phi)^{-1} T_{p_t(d\omega)}\,\rm{Prob}(M, g)\right|_x\,\simeq \,\mathcal{H}_{t,\,\phi(x)}(TM)\;.
\end{equation}
We have
\begin{lem}
For any given $x\in\Sigma$, we can associate to the vector  $v\,=\,v^i\tfrac{\partial }{\partial \phi^i(x)}\,\in\,\left(\phi^{-1}TM\right)_x$  the vector field over $M$ given by 
\begin{eqnarray}
\widehat{\Psi}_{t,\phi}\,:\,
\mathbb{R}_{\,\geq 0}\times \left(\phi^{-1}\,TM\right)_x\,\;&\longrightarrow&\,\;\;\;\;\;\;\;\;\mathcal{H}_{t,\,\phi(x)}(TM)\\
\nonumber\\
\left(t,\;v^i\tfrac{\partial }{\partial \phi^i(x)}\right)\,\;\;\;\;\;\;&\longmapsto&\,
\widehat{\Psi}_{t,\phi}(v)\,:=\,\nabla_{(y)} ^i 
\widehat{\psi}_{\left(t,\, \phi(x),\,v\right)}\,(y)\,\frac{\partial }{\partial y^i}\;,\nonumber
\label{ivector}
\end{eqnarray}
where $y\,\mapsto\, \widehat{\psi}_{(t,\, \phi(x),\,v )}\,(y)$ denotes the Otto potential 
associated to $v\,=\,v^i\,\tfrac{\partial }{\partial \phi^i(x)}$,\;(cf. (\ref{carlo12H})).
\end{lem}
\begin{proof} This immediately follows by recalling that the map
\begin{equation}
T_{\phi(x)}M\,\longrightarrow\, T_{p_t(d\omega)}\,\rm{Prob}(M, g)\,\simeq \,\mathcal{H}_{t,\,\phi(x)}(TM)
\end{equation}
is, according to (\ref{tangmappsi}), an injection.
\end{proof}
To handle the situation 
where $x$ varies smoothly over $\Sigma$, we denote 
by $\phi^{-1}TM \,\boxtimes\,TM$ the vector bundle over $\Sigma\,\times\,M$ whose fiber over $(x,y)\in\,\Sigma\,\times\,M$ is given by
\begin{equation}
\label{bibundle}
\left.\phi^{-1}TM \,\boxtimes\,TM\right|_{(x,y)}\,=\,
\left.\phi^{-1}TM\right|_{x}\,\otimes \,TM_y\;.
\end{equation}
With this notation along the way, let us consider the section of $\phi^{-1}TM \,\boxtimes\,TM$ associated to the push--forward $\phi_*\partial_{\alpha}\,=\,\tfrac{\partial \phi^i}{\partial x^{\alpha}}\,\tfrac{\partial }{\partial \phi^i}$ of the $x^\alpha$--basis vector,
\\
\begin{eqnarray}
\label{Psialpha}
\Sigma\,\times \,M\,\;\;\;\;\;\;\;\;&\longrightarrow &\,\;\;\;\;\;\;\;\;\;\;\;\;\phi^{-1}TM \,\boxtimes\,TM\\
(x,\,y)\,\;\;\;\;\;\;\;\;\;&\longmapsto &\,\;\;\;\;\;\;\left(\phi_*\partial_\alpha,\;\widehat{\Psi}_{t,\phi}(\phi_*\partial_\alpha)\right)
\;,\;\;\;\;\;\alpha\,=\,1,2\,.\nonumber
\end{eqnarray}
\\
\noindent
Since is clear from the notation adopted which $\phi^{-1}TM$--vector we are talking about, 
  we shall write (\ref{Psialpha})  as $\widehat{\Psi}_{t,\phi}(\phi_*\partial_\alpha)$.
\begin{rem}
 If we fix the  index $\alpha\,=1,2$,\,  and let $x^\alpha\,\in[0,1]$, then we can consider $\phi_*\partial_{\alpha}$ as the tangent vector covering the coordinate curve  $\phi_\alpha:\,[0,1]\longrightarrow M$,\; $x^\alpha\,\longmapsto \,\phi(x^\alpha,\,x^\beta=0)$, \,$\beta\not=\alpha$, issued from $\phi(0)$. In particular,  we can think of  $\widehat{\Psi}_{t,\phi}(\phi_*\partial _\alpha )$ as the tangent vector in $(\rm{Prob}_{ac}(M,g),\,d_g^w)$  to the absolutely continuous curve of probability measures
\begin{equation}
\label{coordcurve}
[0,1]\,\ni \,x^\alpha\,\longmapsto \, \nu(x^\alpha)\,:=\,p_t^{\,(\omega)}(\cdot ,\,\phi_{\alpha}(x))\,d\omega(\cdot )\;.
\end{equation} 
\end{rem}
\noindent
According to (\ref{ivector}), for any $t\,>0$, we can associate to  $d\phi=dx^{\alpha}\otimes\tfrac{\partial\phi^i}{\partial x^{\alpha}} \tfrac{\partial }{\partial \phi^i}$ the 
heat kernel deformed differential  
\begin{eqnarray}
\label{barlocphidef}
\\
&&
\Sigma\,\times\,M\,\longrightarrow \,\;\;\;\;\;\;\;\;\;\;\;\;T^*\Sigma\otimes \,
\left(\phi^{-1}TM\,\boxtimes\,TM\right)\nonumber\\
\nonumber\\
&&(x,\,y))\,\longmapsto \,d\widehat{\Psi}_{t,\phi}\,:=\,
dx^{\alpha}\otimes\widehat{\Psi}_{t,\phi}(\phi_*\partial_{\alpha})\nonumber\\
\nonumber\\
&& \,=\,dx^\alpha\,\otimes \,\nabla_{(y)} ^i 
\widehat{\psi}_{\left(t,\, \phi(x),\,\phi_*\partial_\alpha\right)}\,\frac{\partial }{\partial y^i}\nonumber\;.
\end{eqnarray}
The corresponding space of sections $\left\{d\widehat{\Psi}_{t,\phi}\,\in\, T^*\Sigma\otimes \,\left(\phi^{-1}TM
 \,\boxtimes\,TM \right) \right\}$
is endowed with the $T^*\Sigma\otimes\,L^2(p_t(d\omega,\,\phi(x))$ inner product
\begin{eqnarray}
\label{bundlemetrcc}
&&\bigg(d\widehat{\Psi}_{t,\phi}(u)   ,\,  d\widehat{\Psi}_{t,\phi}(v)  \bigg)_{p_t(d\omega)}(x):=\,\gamma^{-1}(dx^\alpha ,\,dx^\beta )(x)\\
\nonumber\\
&&\,\otimes\,
\int_{M}\,g_{km}(y)\,\nabla ^{k}_{y}\,\widehat{\psi}_{(t,\, \phi(x),\,u(x))}\,\nabla ^{m}_{y}\,
\widehat{\psi}_{(t,\, \phi(x),\,v(x))}\,p^{(\omega)} _t(y , \phi(x))\,d\omega(y)\;,\nonumber
\end{eqnarray}
 in terms of which we can define the pre-Hilbertian $L^2(p _t(d\omega)\otimes d\mu_\gamma)$ norm 
\\
\begin{eqnarray}
\label{prehilnorm}
\\
&&\parallel d\widehat{\Psi}_{t,\phi} \parallel_{p_t(d\omega)\otimes d\mu_\gamma}^2\,=\, \left\langle d\widehat{\Psi}_{t,\phi},\,d\widehat{\Psi}_{t,\phi} \right\rangle_{p_t(d\omega)\otimes d\mu_\gamma}\,
\nonumber\\
\nonumber\\
&&=\,\int_{\Sigma}d\mu_{\gamma}(x)\,\,\gamma^{\alpha\beta}(x)\,\int_{M}\,g\left(\widehat{\Psi}_{t,\phi}(\phi_*\partial_\alpha),\,\widehat{\Psi}_{t,\phi}(\phi_*\partial_\beta) \right)\,p^{(\omega)} _t(y , \phi(x))\,d\omega(y)\nonumber\\
\nonumber\\
&&=\,\int_{\Sigma}d\mu_{\gamma}(x)\,\gamma^{\alpha\beta}(x)\nonumber\\
\nonumber\\
&&\,
\times\,\int_{M}\,g_{km}(y)\,\nabla ^{k}_{y}\,\widehat{\psi}_{(t,\, \phi(x),\,\phi_*\partial_\alpha )}\,\nabla ^{m}_{y}\,\widehat{\psi}_{(t,\, \phi(x),\,\phi_*\partial_\beta )}\,p^{(\omega)} _t(y , \phi(x))\,d\omega(y)\;.
\nonumber
\end{eqnarray}
(see (\ref{Psialpha})). We shall often write $\parallel d\widehat{\Psi}_{t,\phi} \parallel_{t}^2\,$ if there is no danger of confusion. Let 
\\ 
\begin{equation}
\label{Hilbertsec}
\mathcal{H}_{t,\,\phi}(\Sigma,\,M)\,:=\, \overline{\left\{C^{\infty}\left(\Sigma\,\times\,M,\;
T^*\Sigma\otimes \,\left(\phi^{-1}TM\,\boxtimes\,TM\right)\right)\right\}}^{\;L^{2}(p_t(d\omega)\otimes d\mu_\gamma)}\; ,
\end{equation}
\\
\noindent
 be the Hilbert space of sections of $T^*\Sigma\otimes \,\left(\phi^{-1}TM \,\boxtimes\,TM\right)$ obtained by completion with respect to the norm (\ref{prehilnorm}), 
\, and define the associated energy functional
\begin{eqnarray}
\label{penergy}
\mathcal{E}\,:\, \mathcal{H}_{t,\,\phi}(\Sigma, M)\,&\longrightarrow& \,   \mathbb{R}\;,\\
\nonumber\\
d\widehat{\Psi}_{t,\phi}\,\;\;\;\;\;\;\;\;\;\;&\longmapsto &\,\mathcal{E}[\widehat{\Psi}_{t,\phi}]\,:=\,\frac{1}{2}\,\parallel d\widehat{\Psi}_{t,\phi} \parallel_{t}^2\;.\nonumber
\end{eqnarray}
The geometrical meaning of $\mathcal{E}[\widehat{\Psi}_{t,\phi}]$ is provided by 
\begin{prop}
\label{defharmenerg}
The energy functional $\mathcal{E}[\widehat{\Psi}_{t,\phi}]$ is a generalized harmonic map functional over  the Wasserstein metric space  
\begin{equation}
\Upsilon_t\left((M,g,\,d\omega)  \right)\cap (\rm{Prob}_{ac}(M,g),\,d_g^w)\;,
\end{equation}
where $\Upsilon_t$ is the weighted heat kernel injection map (\ref{warpembed}). $\mathcal{E}[\widehat{\Psi}_{t,\phi}]$ can be identified with the smooth deformation 
of $E[\phi,\,g]_{(\Sigma,\, M)}$ generated by the flow of metrics $(t,\,g)\longmapsto g_t^{(\omega)}$. In particular, for $t\in (0,\infty)$, we can write
\\ 
\begin{eqnarray}
\label{energ1}
&&\mathcal{E}[\widehat{\Psi}_{t,\phi}]\,:=\,\frac{1}{2}\,
\parallel d\widehat{\Psi}_{t,\phi} \parallel_{t}^2\\
\nonumber\\
&&=\,\frac{1}{2}\,\int_\Sigma\,\gamma^{\mu\nu}\,\frac{\partial \phi^{i}}{\partial x^{\mu}}
\frac{\partial\phi^{j}}{\partial x^{\nu}}\,(g_t^{(\omega)})_{ij}(\phi)\,d\mu_{\gamma}\,:=\,
E[\phi,\,g_t^{(\omega)}]_{(\Sigma,\, M)}\;,\nonumber
\end{eqnarray}
with
\begin{equation}
\label{energ2}
\lim_{t\,\searrow 0}\,\mathcal{E}[\widehat{\Psi}_{t,\phi}]\,=\,E[\phi,\,g]_{(\Sigma,\, M)}\;.
\end{equation}
\end{prop}
\begin{proof}
Since $\mathcal{E}[\widehat{\Psi}_{t,\phi}]$ is conformally invariant and $\Sigma\simeq \mathbb{T}^2$, we can assume that $\gamma$ is the flat metric $\delta^{\alpha\beta}$ and write (\ref{penergy})  as  
\begin{equation}
\label{penergy2}
\mathcal{E}[\widehat{\Psi}_{t,\phi}]\,=\,
\frac{1}{2}\,\sum_{\alpha=1,2}\,\int_{\Sigma}\,d\mu_\gamma(x)\,
\int_{M}\,\left|\nabla_{y}\,\widehat{\psi}_{(t,\, \phi(x),\,\phi_*\partial_\alpha )}\right|^2_g\,p^{(\omega)} _t(y , \phi(x))\,d\omega(y)\;.
\end{equation} 
According to (\ref{velsquare}), we have,  for almost every $x^\alpha\in[0,1]$, 
\begin{equation}
\int_{M}\,\left|\nabla_{y}\,\widehat{\psi}_{(t,\, \phi(x),\,\phi_*\partial_\alpha )}\right|^2_g\,p^{(\omega)} _t(y , \phi(x))\,d\omega(y)\,=\,\left| \frac{d{\nu}(x^\alpha)}{d x^{\alpha}} \right|^2\;,
\end{equation}
where
\begin{equation}
\left| \frac{d{\nu}(x^\alpha)}{d x^{\alpha}} \right|\,=\,
\,\lim_{\epsilon \rightarrow 0}\,\frac{d^W_g\left(\nu(x^\alpha+\epsilon),\,\nu(x^\alpha)  \right)}{\epsilon }\;
\end{equation}
is the metric speed, in  $(\rm{Prob}_{ac}(M,g),\,d_g^w)$, of the curves (\ref{coordcurve}), (see (\ref{metricvel})). 
By absolute continuity of $x^\alpha\longmapsto \nu(x^\alpha)$ we can write
\begin{eqnarray}
\label{penergy3}
\mathcal{E}[\widehat{\Psi}_{t,\phi}]\,&=&\,
\frac{1}{2}\,\sum_{\alpha=1,2}\,\int_{\Sigma}\,d\mu_\gamma(x)\,\left| \frac{d{\nu}_\alpha(\phi(x))}{d x^{\alpha}} \right|^2\\
\nonumber\\
&=&\,
\lim_{\epsilon \rightarrow 0}\,
\frac{1}{2}\,\sum_{\alpha=1,2}\,\int_{\Sigma}\,d\mu_\gamma(x)\,\frac{\left[d^W_g\left(\nu_{\alpha+\epsilon },\,\nu_\alpha  \right)\right]^2}{\epsilon^2 }\;,
\end{eqnarray}
from which it follows that $\mathcal{E}[\widehat{\Psi}_{t,\phi}]$ has the structure of a generalized harmonic map functional over the  Wasserstein metric space  $\Upsilon_t\left((M,g,\,d\omega)  \right)\cap (\rm{Prob}_{ac}(M,g),\,d_g^w)$, \;(see Section \ref{digression}). On the other hand, according to the definition (\ref{gt}) of the deformed metric $g_t^{(\omega)}$ and of the vector fields $\widehat{\psi}_{(t,\, \phi(x),\,\phi_*\partial_\alpha )}$, we can write the inner product (\ref{bundlemetrcc})  as
\\    
\begin{eqnarray}
\label{}
\left( d\widehat{\Psi}_{t,\phi},\,d\widehat{\Psi}_{t,\phi} \right)_{p_t(d\omega)}(x)\,
&=&\,\gamma^{\alpha\beta}\,g_t^{(\omega)}\,\left(\tfrac{\partial \phi^i}{\partial x^{\alpha}}\,\tfrac{\partial }{\partial \phi^i},\,
\tfrac{\partial \phi^j}{\partial x^{\beta}}\,\tfrac{\partial }{\partial \phi^j} \right)(x)\,\\
\nonumber\\
&=&\,
\gamma^{\alpha\beta}(x)\,\tfrac{\partial \phi^i}{\partial x^{\alpha}}\,\tfrac{\partial \phi^j}{\partial x^{\beta}}\,
\left(g_t^{(\omega)}  \right)_{ij}(\phi(x))\;,\nonumber
\end{eqnarray}
which implies that (\ref{penergy}) can be rewritten as the harmonic energy functional $E[\phi,\,g_t^{(\omega)}]_{(\Sigma,\, M)}$ associated with the deformed $(M, g_t^{(\omega)})$,\,
\\ 
\begin{eqnarray}
\label{energ1bis}
&&\mathcal{E}[\widehat{\Psi}_{t,\phi}]\,:=\,\frac{1}{2}\,
\parallel d\widehat{\Psi}_{t,\phi} \parallel_{t}^2\\
\nonumber\\
&&=\,\frac{1}{2}\,\int_\Sigma\,\gamma^{\mu\nu}\,\frac{\partial \phi^{i}}{\partial x^{\mu}}
\frac{\partial\phi^{j}}{\partial x^{\nu}}\,(g_t^{(\omega)})_{ij}(\phi)\,d\mu_{\gamma}\,=\,
E[\phi,\,g_t^{(\omega)}]_{(\Sigma,\, M)}\;,\nonumber
\end{eqnarray}
as stated. Moreover, from Lemma \ref{metricatzero}  we get
\begin{equation}
\label{energ2bis}
\lim_{t\,\searrow 0}\,E[\phi,\,g_t^{(\omega)}]_{(\Sigma,\, M)}\,=\,E[\phi,\,g]_{(\Sigma,\, M)}\;,
\end{equation}
from which it follows that $\mathcal{E}[\widehat{\Psi}_{t,\phi}]$ can be pulled--back to
 $\rm{Map}(\Sigma,M)$ as a (smooth) deformation of the harmonic map functional $E[\phi,\,g]_{(\Sigma,\, M)}$.
\end{proof}
\noindent It is straightforward to extend  the above analysis to the warped  map  $\Phi_{(q)}\in Map\left(M^n\times_{(f)} \mathbb{T}^q \right)$, (see (\ref{newMap0})), and characterize a  deformation of the harmonic energy functional (\ref{EPhiq}) according to
\begin{lem}
\label{figrande}
The heat kernel embedding (\ref{warpembed})
generates the scale--dependent harmonic energy functional on the warped manifold $M^n\times_{(f)} \mathbb{T}^q$ given by
\\
\begin{eqnarray}
\label{EPhiqtime}
\\
&&E[\Phi_{(q)},\,h_t^{(q)}]_{(\Sigma, N^{n+q})}\,:=\,
\frac{1}{2}\,\int_\Sigma\,\gamma^{\mu\nu}\,\frac{\partial \Phi_{(q)}^{a}(x,\xi )}{\partial x^{\mu}}
\frac{\partial\Phi_{(q)}^{b}(x,\xi )}{\partial x^{\nu}}\,(h_t)_{ab}(\phi)\,d\mu_{\gamma}\nonumber\\
\nonumber\\
&&\,=\,\frac{1}{2}\,\int_\Sigma\,\gamma^{\mu\nu}\,\frac{\partial \phi^{i}}{\partial x^{\mu}}
\frac{\partial\phi^{j}}{\partial x^{\nu}}\,(g_t^{(\omega)})_{ij}(\phi)\,d\mu_{\gamma}\,\nonumber\\
\nonumber\\
&&+\,\frac{1}{8}\,\sum_{k=1}^q\,d_{g}^2(\phi_{cm},\phi_{(k)})\,
\int_{\Sigma}\,\left|d f_t^{(\omega)}(\phi(x);q)\right|^2_{\gamma}\,d\mu_{\gamma}\;,\nonumber
\end{eqnarray}
where $(t,h)\longmapsto h_t$, \;$t\in (0,\infty)$,\, is the flow of metrics defined by (\ref{hTmetric}).  
\end{lem}
\noindent This directly implies the
\begin{prop}
If $t\mapsto (\gamma_t)_{\mu\nu}=e^{f_t^{(\omega)}(\phi)}\, \delta_{\mu\nu}$,\;\;$t\in (0,\infty)$, denotes the family of conformally flat metrics on $\Sigma\simeq \mathbb{T}^2$ associated with $(t,f)\mapsto f_t^{(\omega)}$, then in the conformal gauge  $(\Sigma,\gamma_t)$ the  harmonic energy functional (\ref{EPhiqtime}) provides a scale--dependent family of dilatonic actions 
\\
\begin{eqnarray}
\label{eqnAct100time}
&&{S}_M\left[\gamma_t,\phi;\,\,F(\phi_{cm};q),\,{f_t^{(\omega)}},\,g_t^{(\omega)}\right]\\
\nonumber\\
&&\;\underset{(f_t)}{:=}\; 
\frac{2}{F(\phi_{cm};q)}\;
E[\Phi_{(q)},\,h_t^{(q)}]_{(\Sigma, N^{n+q})}\;,\nonumber
\end{eqnarray}
such that
\begin{eqnarray}
&&\lim_{t\searrow \,0}\,{S}_M\left[\gamma_t,\phi;\,\,F(\phi_{cm};q),\,{f_t^{(\omega)}},\,g_t^{(\omega)}\right]\\
\nonumber\\
&&=\,{S}_M\left[\gamma,\phi;\,\,F(\phi_{cm};q),\,{f},\,g\right]\;.\nonumber
\end{eqnarray}
\end{prop}
\begin{proof}
In the  conformal gauge, (see (\ref{psiconf00})),
 \begin{equation}
\left(\Sigma,(\gamma_t)_{\mu\nu}=e^{f_t^{(\omega)}(\phi)}\, \delta_{\mu\nu}\right)\;,\;\;\;\;t\in [0,\infty)\;, 
\label{psiconf00time}
\end{equation}
the harmonic energy functional (\ref{EPhiqtime}) can be rewritten as 
\begin{eqnarray}
\label{Erelation00time}
&&E[\Phi_{(q)},\,h_t^{(q)}]_{(\Sigma, N^{n+q})}\, \underset{(f_t)}{=}\,
E[\phi,\,g_t^{(\omega)}]_{(\Sigma,\, M)}\\
\nonumber\\
&&\,+
\,\frac{F(\phi_{cm};q)}{2}\, \int_{\Sigma}\,f_t^{(\omega)}(\phi;q)\,\mathcal{K}_{f_t^{(\omega)}}\,
d\mu _{\gamma_t}\;.\nonumber
\end{eqnarray}
If we define  ${S}_M\left[\gamma_t,\phi;\,\,F(\phi_{cm};q),\,{f_t^{(\omega)}},\,g_t^{(\omega)}\right]$ 
according to (\ref{eqnAct100time}), then the statement follows from (\ref{energ2bis})   and the definition (\ref{eqnAct100}) 
of the dilatonic action ${S}_M\left[\gamma,\phi;\,\,F(\phi_{cm};q),\,{f},\,g\right]$.
\end{proof}

\section{The heat kernel embedding and Renormalization Group}
\label{renGroup}
The results of the previous section imply that along the heat kernel embedding $\Upsilon _t$ we get the induced flow 
\begin{equation}
\label{flucactionBIS}
[0,\infty )\,\ni \,t\,\longmapsto\, {S}_M\left[\gamma_t,\phi;\,\,F(\phi_{cm};q),\,{f_t^{(\omega)}},\,g_t^{(\omega)}\right]\;,
\end{equation}
deforming the dilatonic  action ${S}_M\left[\gamma,\phi;\,\,F(\phi_{cm};q),\,{f},\,g\right]$ in the direction of
the non--trivial geometric rescaling $(t,\,g,\,f)\,\longmapsto g^{\,(\omega)}_t,\,f_t^{(\omega)}$ 
of the \emph{coupling} $(M,g,\,d\omega)$.
This strongly suggests a connection between heat kernel embedding and the circle of ideas and techniques related to renormalization group. As discussed in the introductory section \ref{selfcont}, the strategy\footnote{See  \cite{gawedzki} for a general review, and \cite{mauro1} for a renormalization group analysis in the Ricci flow setting.}of the renormalization group analysis of the non--linear $\sigma$ model \cite{DanPRL}, \cite{DanAnnPhys}  is to discuss the scaling behavior of the (quantum) fluctuations of the maps $\phi:\Sigma\to M$ around the background  average field $\phi_{cm}$,  defined by the distribution of the center of mass of a large ($q\to\infty$) number of randomly distributed independent copies $\{\phi_{(\mathfrak{j})}\}_{\mathfrak{j}=1}^q$ of $\phi$ itself. This \emph{background field technique} can be seen as a natural extension of the constant map localization described in Section \ref{lacmadw}. To formulate it within our framework,  we must provide, in a suitable abstract Wiener space, a Borel functional measure whose properties reflect the heat kernel--induced flow for the harmonic map functional $E[\Phi_{(q)},\,h_t^{(q)}]$.\\
\\
\noindent  For technical reasons, (existence of a unique center of mass), we assume that the maps $\{\phi_{(\mathfrak{j})}\}_{\mathfrak{j}=1}^q$ all take values in a convex ball $B(p,r)$ 
with  $r\,<\,\min\left\{\tfrac{1}{3}\,inj\,(M),\,\tfrac{\pi }{6\,\sqrt{\kappa }}  \right\}$,\; (see (\ref{errezero1}), and section \ref{lacmadw} for notation).  Under these hypotheses, for any given $x\in\Sigma$ there is a unique center of mass in $B(p,2r)$ for the corresponding  collection of points $\{\phi_{(\mathfrak{j})}(x)\}_{\mathfrak{j}=1}^q\,\in\,B(p,r)$, (see section \ref{lacmadw}). Thus,  
\begin{eqnarray}
\label{cmmappa}
\phi_{cm}\,:\,\Sigma\,&\longrightarrow & \,\;\;\;\;\;\;\;\;M\\
x\,&\longmapsto &\, \phi_{cm}(x)\,:=\,\inf_{y\in\,B(p,2r)}\,\left[\frac{1}{2}\,
\sum_{\mathfrak{j}=1}^q\,d^2_g(y,\,\phi_{(\mathfrak{j})}(x) )\right]\;,\nonumber
\end{eqnarray}
is well--defined and provides the  center of mass map generated by the set of maps $\{\phi_{(\mathfrak{j})}\}_{\mathfrak{j}=1}^q$. 
Let  $\phi^{-1}_{cm}\,TM$ be the corresponding pull--back bundle over $\Sigma$.
Since $B(p,2r)$ is assumed convex, we can introduce $q$ sections 
$v_{(\mathfrak{j})}:\Sigma\to \phi^{-1}_{cm}\,TM$, \;
$x\,\mapsto \,v_{(\mathfrak{j})}(x)\,=\,v_{(\mathfrak{j})}^k(x)\tfrac{\partial }{\partial \phi_{cm}^k(x)}$, and parametrize the maps $\{\phi_{(\mathfrak{j})}\}_{\mathfrak{j}=1}^q$ according to  
\begin{equation}
\label{kmap}
\phi_{(\mathfrak{j})}(x) \,=\, \exp_{cm(x)}(v_{(\mathfrak{j})}(x))\;,
\end{equation}
where $\exp_{cm(x)}: T_{\phi_{cm}(x)}M\to B(p,2r)$ denotes the exponential map based at $\phi_{cm}(x)$, and where the center of mass constraint takes the form  $\sum_{\mathfrak{j}=1}^q\,v_{(\mathfrak{j})}(x)=0$. More generally, for $v\in \phi^{-1}_{cm}\,TM$ we shall write $\phi_{(v)}(x) \,=\, \exp_{cm(x)}(v(x))$. \\
\\
\noindent As recalled above, the strategy in the (perturbative) renormalization group analysis  of non--linear $\sigma$ model is to let the sections $v_{(\mathfrak{j})}$ fluctuate around a classical background defined by the center of mass map $\phi_{cm}$. The fluctuations are generated by
assuming that  the fields $v_{(\mathfrak{j})}$, subject to the constraint $\sum_{\mathfrak{j}=1}^q\,v_{(\mathfrak{j})}(x)=0$, are vector valued random variables distributed on $Map(\Sigma, \phi_{cm}^{-1}\,TM)$ according to a (non--existing) infinite-dimensional probability measure $\mathbb{P}[D\,v_{(\mathfrak{j})}]$. It is customary to formally write
\begin{equation}
\mathbb{P}[D\,v_{(\mathfrak{j})}]\,:=\,Z^{-1}\,e^{-S_{cm}[v_{(\mathfrak{j})};\,a]}\,D\,v_{(\mathfrak{j})}\;,
\end{equation}
where $S_{cm}[v_{(\mathfrak{j})};\,a]\,:=\,\tfrac{2}{a}\,E[\phi_{(\mathfrak{j})},\,g]$ is the non--linear $\sigma$ model action associated to the harmonic map functional $E[\phi_{(\mathfrak{j})},\,g]$, 
and where the normalization factor
\begin{equation}
Z\,:=\,\int_{Map(\Sigma, \phi_{cm}^{-1}\,TM)}\,
e^{-S_{cm}[v_{(\mathfrak{j})};\,a]}\,D\,v_{(\mathfrak{j})}
\end{equation}
is the partition function of the theory. In such a setting, the completion of  $Map(\Sigma, \phi_{cm}^{-1}\,TM)$ under the energy norm associated to $S_{cm}[v_{(\mathfrak{j})};\,a]$, is considered to be the physical Hilbert space of choice. It is well--known that, from the point of view of ($\infty $--dimensional) geometrical analysis, the existence of the distribution measure $\mathbb{P}[D\,v_{(\mathfrak{j})}]$ on such a space is obstructed by the Borel--Cantelli lemma, according to which $S_{cm}[v_{(\mathfrak{j})};\,a]$ is $\mathbb{P}$--almost surely divergent. The way out from this impasse  is to extract from $S_{cm}[v_{(\mathfrak{j})};\,a]$, by perturbative techniques, a (dimensionally regularized) Gaussian measure with respect to which  $\prod_{\mathfrak{j}=1}^{q} e^{-S_{cm}[v_{(\mathfrak{j})};\,a]}\,D\,v_{(\mathfrak{j})}$
 can be formally interpreted  as generating the fluctuations of $v_{(\mathfrak{j})}$ in the 
large deviations sense \cite{gawedzki}.
\begin{rem}
Typically, the Gaussian measure in question is generated by a rough Laplace-Beltrami operator $\triangle_\Sigma$ associated to the pulled back Levi--Civita connection on $\phi_{cm}^{-1}\,TM$. Since $\Sigma$ is two-dimensional, the Green function of $\triangle_\Sigma$ is logarithmically divergent and needs to be regularized in order to be used as the covariance function of a Gaussian measure. Dimensional regularization is often the choice.
\end{rem}
\noindent
The regularization procedure, necessary for defining the reference Gaussian measure, introduces a running length scale  in the theory. By renormalization group techniques, this allows to control the behavior of the measure  $e^{-S_{cm}[v_{(\mathfrak{j})};\,a]}\,D\,v_{(\mathfrak{j})}$ under fluctuations, and gives rise to a perturbative  rescaling of the geometrical couplings $(M,g,\,d\omega)$ which, as recalled in section \ref{selfcont}, connects to the Ricci flow \cite{mauro1}, \cite{DanPRL}, \cite{DanAnnPhys}, \cite{gawedzki}.   \\
 \\
 \noindent To provide a mathematical well--defined functional representation of the heat kernel embedding of the non--linear $\sigma$ model, which to some extent conveys the ideas of the physics path integral and of renormalization group, we need to relax on considering the energy norm completion of $Map(\Sigma, \phi_{cm}^{-1}\,TM)$ as the physical Hilbert space. This can be done by constructing a Gaussian measure on an abstract Wiener space associated with the heat kernel embedding. The resulting picture gives rise to a Gaussian renormalization group action, a toy model of NL$\sigma $M renormalization which nonetheless conveys many of the relevant features of the physical RG flow.

\subsection{A Wiener space associated to the heat kernel embedding}
 Let
\begin{equation}
d\phi_{(\mathfrak{j})}(x)\,=\,dx^\alpha\otimes(\phi_{(\mathfrak{j})})_*\partial_\alpha
\;,\;\;\;\;\;v_{(\mathfrak{j})}\,\in \,\phi_{cm}^{-1}TM\;,
\end{equation}
denote the differential  of the map $\phi_{(\mathfrak{j})}(x) = \exp_{cm(x)}(v_{(\mathfrak{j})}(x))$, (see (\ref{kmap})). According to (\ref{barlocphidef}), we define the corresponding heat kernel deformed section by  
\begin{eqnarray}
\label{Hilbsection}
&&
\Sigma\,\times\,M\,\longrightarrow \,\;\;\;\;\;\;\;\;\;\;\;\;T^*\Sigma\otimes \,\left(\phi_{cm}^{-1}TM \,\boxtimes\,TM\right)\\
\nonumber\\
&&(x,\,y))\,\longmapsto \,d \widehat{\Psi}_{t,(\mathfrak{j})}\,:=\,dx^\alpha\otimes 
\widehat{\Psi}_{t,\phi_{cm}}((\phi_{(\mathfrak{j})})_*\partial_\alpha)\;,\nonumber
\end{eqnarray}
 where
\begin{equation}
\label{vectfields}
\widehat{\Psi}_{t,\phi_{cm}}((\phi_{(\mathfrak{j})})_*\partial_\alpha)\,=
\, \nabla^k_{(y)} \widehat{\psi}_{\left(t,\, \phi_{cm}(x),\,(\phi_{(\mathfrak{j})})_*\partial_\alpha\right)}\,
\frac{\partial }{\partial y^k}\;.
\end{equation}
\begin{rem}
Note that the section $dx^\alpha\otimes 
\widehat{\Psi}_{t,\phi_{cm}}((\phi_{cm})_*\partial_\alpha)$, associated to the differential $d\phi_{cm}$ of the center of mass map $\phi_{cm}$, corresponds to $\sum_{\mathfrak{j}=1}^q\, v_{(\mathfrak{j})}=0$. Hence we set
\begin{equation}
\label{cmpsi}
d \widehat{\Psi}_{t,cm}\,:=\,dx^\alpha\otimes \,\widehat{\Psi}_{t,cm}((\phi_{cm})_*\partial_\alpha)\,:=\,\sum_{\mathfrak{j}=1}^q\,d \widehat{\Psi}_{t,(\mathfrak{j})}\;.   
\end{equation}
\end{rem}
\noindent The natural framework for the discussing the energetics of the fluctuations of the maps $d \widehat{\Psi}_{t,(\mathfrak{j})}$ around their background $d \widehat{\Psi}_{t,cm}$,\; is the Hilbert space of sections  
\begin{eqnarray}
\label{Hilbertsec bis}  
&&\mathcal{H}_{t,\, \phi_{cm}}(\Sigma,M)\\
\nonumber\\
&&\,:=\, \overline{\left\{C^{\infty}\left(\Sigma\,\times\,M,\;
T^*\Sigma\otimes \,\left(\phi^{-1}_{cm}TM\,\boxtimes\,TM\right)\right)\right\}}^{\;L^{2}(p_t(d\omega)\otimes d\mu_\gamma)}\; ,\nonumber
\end{eqnarray}
\\
\noindent
obtained by completion with respect to the $L^2({p_t(d\omega)\otimes d\mu_\gamma})$--norm 
\begin{equation}
\label{prehilnorm bis}
\parallel d\widehat{\Psi}_{t,\phi_{cm}} \parallel_{p_t(d\omega)\otimes d\mu_\gamma}^2\,=\, \left\langle d\widehat{\Psi}_{t,\phi_{cm}},\,d\widehat{\Psi}_{t,\phi_{cm}} \right\rangle_{p_t(d\omega)\otimes d\mu_\gamma}\;,
\end{equation}
defined by (\ref{prehilnorm}) and (\ref{Hilbertsec}), (for $\phi=\phi_{cm}$).\\
\\
\noindent
As hinted in the introductory remarks, $\mathcal{H}_{t,\, \phi_{cm}}(\Sigma,M)$ is too small to support a Gaussian measure randomizing the distribution of the $d\widehat{\Psi}_{t,\phi_{(j)}}$ and giving them finite (harmonic) energy norm. The mathematical strategy to circumvent this difficulty is due to L. Gross \cite{Gross}, and is to modify the Hilbertian norm so as to characterize a Banach space large enough to host the desired Borel measure, (various geometrical aspects  of Gross theory  are discussed in \cite{Driver},  \cite{Leandre}, \cite{Stroock}, \cite{Taubes}, \cite{Weitsman}). In our case
the required extension is naturally suggested by the properties of the heat kernel embedding.
  For any fixed $t>0$, and $\eta\,\in\, \mathbb{R}$,  let us consider the inner product  defined, 
on  sections of the bundle (\ref{Hilbsection}),  by
\begin{eqnarray}
\label{innereta}
&&\left\langle d \widehat{\Psi}_{t,u},\,d \widehat{\Psi}_{t,v} \right\rangle_{(t,\,-\eta)}\\
\nonumber\\
&&\,:=\,a^{-1}\,\int_\Sigma\,d\mu_\gamma(x)\,\gamma^{\alpha\beta}\int_{M}\,\left[g_{km}(y)\,\widehat{\Psi}^k_{t,\phi_{cm}}((\phi_{(u)})_*\partial_\alpha)\right.\nonumber\\
\nonumber\\
&&\left.\times \,\left(1-\,a\,\Delta _{p_t(d\omega)} \right)^{-\,\eta}\,\widehat{\Psi}^m_{t,\phi_{cm}}((\phi_{(v)})_*\partial_\beta)\right]\,p^{(\omega)} _t(y , \phi_{cm}(x))\,d\omega(y)\;,
\nonumber
\end{eqnarray}
where $\Delta _{p_t(d\omega)}$ denotes the weighted (rough) Laplacian (\ref{heatweight}) acting on the vector fields 
$M\ni y\mapsto  \nabla_{(y)} \widehat{\psi}_{\left(t,\, \phi_{cm}(x),\,(\phi_{(v)})_*\partial_\alpha\right)}\in TM$, and where, for later convenience, we have  inserted the dimensional coupling constant $a$, 
(see (\ref{eqnAct})), to make (\ref{innereta}) explicitly dimensionless. 
\begin{rem}
For $\eta=0$, (\ref{innereta}) induces the energy norm on $\mathcal{H}_{t,\, \phi_{cm}}(\Sigma,M)$,
\begin{eqnarray}
&&\left.\left\langle d \widehat{\Psi}_{t,u},\,d \widehat{\Psi}_{t,u} \right\rangle_{(t,\,-\eta)}\right|_{\eta=0}\,
=\,\frac{2}{a}\,\mathcal{E}[\widehat{\Psi}_{t,\phi_{cm}}]\\
\nonumber\\
&&=\,
\sum_{\alpha=1,2}\,\int_{\Sigma}\,d\mu_\gamma(x)\,
\int_{M}\,\left|\nabla_{y}\,\widehat{\psi}_{(t,\, \phi_{cm},\,(\phi_{(u)})_*\partial_\alpha )}\right|^2_g\,
\,p^{(\omega)} _t(y , \phi(x))\,d\omega(y)\;.\nonumber
\end{eqnarray} 
Similarly, from the relation
\begin{eqnarray}
&&\int_{M}\,\left(\nabla^i\,\widehat{\psi}_{(t)}\,
\Delta _{p_t(d\omega)}\,\nabla_i\,\widehat{\psi}_{(t)}\right)
\,p^{(\omega)} _t(y , \phi(x))\,d\omega(y)\\
\nonumber\\
&&=\,-\,\int_{M}\,\left(\nabla^k\,\nabla^i\,\widehat{\psi}_{(t)}\,
\nabla_k\,\nabla_i\,\widehat{\psi}_{(t)}\right)
\,p^{(\omega)} _t(y , \phi(x))\,d\omega(y)\;, \nonumber
\end{eqnarray}
which holds for any $t>0$ and  smooth $\widehat{\psi}_{(t)}\equiv \widehat{\psi}_{(t,\, \phi_{cm},\,(\phi_{(u)})_*\partial_\alpha )}$,  we get, for $\eta=\,-1$,  
\begin{eqnarray}
\label{}
&&\left\langle d \widehat{\Psi}_{t,u},\,d \widehat{\Psi}_{t,u} \right\rangle_{(t,\,-\eta=1)}\,\\
\nonumber\\
&&\,=\,
a^{-1}\,\sum_{\alpha=1,2}\,\int_{\Sigma}\,d\mu_\gamma(x)\,
\int_{M}\,\left[\;\left|\nabla_{y}\,\widehat{\psi}_{(t,\, \phi_{cm},\,(\phi_{(u)})_*\partial_\alpha )}\right|^2_g\,\right.\nonumber\\
\nonumber\\
&&+\left.a\, \left|Hess_{y}\,\widehat{\psi}_{(t,\, \phi_{cm},\,(\phi_{(u)})_*\partial_\alpha )}\right|^2_g\;\, \right]
\,p^{(\omega)} _t(y , \phi(x))\,d\omega(y)\;.\nonumber
\end{eqnarray}
\end{rem}
\noindent We denote by
\begin{equation}
\label{etaHilbert}
\mathcal{H}_{t,\, \phi_{cm}}^{-\,\eta}(\Sigma,M)\,:=\, \overline{
\left\{\mathcal{H}_{t,\, \phi_{cm}}(\Sigma,M)\right\}}^{\;\;{-\,\eta}}\; 
\end{equation} 
the separable Hilbert  space obtained by completion of $\mathcal{H}_{t,\, \phi_{cm}}(\Sigma,M)$ with respect 
to the inner product (\ref{innereta}). As a direct consequence of the definition  (\ref{etaHilbert}) and of  (\ref{innereta})   we have
\begin{lem}
\label{denseincl}
If $\mathcal{H}_{t,\, \phi_{cm}}^{\eta}(\Sigma,M)$ denotes the topological dual of  $\mathcal{H}_{t,\, \phi_{cm}}^{\,-\eta}(\Sigma,M)$ then, for $\eta>0$,\,
\begin{equation}
\label{Hilbincl}
\mathcal{H}_{t,\, \phi_{cm}}^{\eta}(\Sigma,M)\,\hookrightarrow \,\mathcal{H}_{t,\, \phi_{cm}}(\Sigma,M)\,
\hookrightarrow \,\mathcal{H}_{t,\, \phi_{cm}}^{-\eta}(\Sigma,M)\;,
\end{equation}
are continuous inclusions of dense subsets.
\end{lem}
\begin{proof}
The duality between the distribution--valued differentials  $d\phi_{(u)}^{(-\eta)}= dx^{\beta}\otimes (\phi_{(u)}^{(-\eta)})_*\partial_\beta \in \mathcal{H}_{t,\, \phi_{cm}}^{-\eta}(\Sigma,M)$ and the sections $d \widehat{\Psi}_{t,v}^{(\eta)}\in \mathcal{H}_{t,\, \phi_{cm}}^{\eta}(\Sigma,M)$, is defined by the pairing
\begin{eqnarray}
\label{pairing}
\\
&&\left\langle d \widehat{\Psi}_{t,v}^{(\eta)},\, d\phi_{(u)}^{(-\eta)} \right\rangle_{(t,\,\eta)}\,:=\,
a^{-1}\,\int_\Sigma\,d\mu_\gamma(x)\,\gamma^{\alpha\beta}(x)\nonumber\\   
\nonumber\\
&&\times\int_{M}\left(\widehat{\Psi}_{t,\phi_{cm}}^{(\eta)}((\phi_{(v)})_*\partial_\alpha)\right)\,\circ  _{g(y)}
\left(1-a\,\Delta _{p_t(d\omega)} \right)^{\eta}\,(\phi_{(u)}^{(-\eta)})_*\partial_\beta
\,\nonumber\\
\nonumber\\
&&p^{(\omega)} _t(y , \phi_{cm}(x))\,d\omega(y)\;,
\nonumber
\end{eqnarray}
where $\circ  _{g(y)}$ denotes the $(M,g)$ inner product. Since $\mathcal{H}_{t,\, \phi_{cm}}(\Sigma,M)$ is embedded as a dense subset into  $\mathcal{H}_{t,\, \phi_{cm}}^{-\eta}(\Sigma,M)$, it follows by topological  duality that $\mathcal{H}_{t,\, \phi_{cm}}^{\eta}(\Sigma,M)$ is densely embedded 
into $\mathcal{H}_{t,\, \phi_{cm}}(\Sigma,M)$. In particular, we can inject $\mathcal{H}_{t,\, \phi_{cm}}^\eta(\Sigma,M)$ in $\mathcal{H}_{t,\, \phi_{cm}}(\Sigma,M)$ by identifying $d \widehat{\Psi}_{t,v}^{(\eta)}$ with its representative $d \widehat{\Psi}_{t,v}$ in $\mathcal{H}_{t,\, \phi_{cm}}(\Sigma,M)$, via the Riesz representation theorem 
\begin{equation}
\left\langle d \widehat{\Psi}_{t,v}^{(\eta)},\,d\Xi_t  \right\rangle_{(t,\,\eta)}\,=\,\left\langle
 d \widehat{\Psi}_{t,v},\,d\Xi_t \right\rangle_{p_t(d\omega)\otimes d\mu_\gamma}\;,
\;\;\;\;\;\;\;\;      \forall d\Xi_t\,\in\,\mathcal{H}_{t}\cap
 \mathcal{H}_{t}^{\,\eta}\;,
\end{equation}
where we have set $\mathcal{H}_{t}^{\eta}\equiv \mathcal{H}_{t,\, \phi_{cm}}^{\eta}(\Sigma,M)$ and $\mathcal{H}_{t}\equiv \mathcal{H}_{t,\, \phi_{cm}}(\Sigma,M)$.
 \end{proof}
\noindent
The norm (\ref{innereta})  is concocted in such a way that the dual space $\mathcal{H}_{t,\, \phi_{cm}}^{-\eta}(\Sigma,M)$ is large enough to support
a Gaussian measure $\mathcal{Q}_{\mathcal{H}_t}$ for which the triple
\begin{equation} 
\left(\mathcal{H}_{t,\, \phi_{cm}}(\Sigma,M),\,\mathcal{H}_{t,\, \phi_{cm}}^{-\eta}(\Sigma,M),\,\mathcal{Q}_{\mathcal{H}_t} \right)
\end{equation} 
is an \emph{abstract Wiener space}. Explicitly, we have
\begin{prop}
Set $\mathcal{H}_{t}^{-\eta}\equiv \mathcal{H}_{t,\, \phi_{cm}}^{-\eta}(\Sigma,M)$ and let $\mathcal{B}_{t}^{-\eta}$
be the Borel 
field  for $\mathcal{H}_{t}^{-\eta}$, defined as the smallest $\sigma$--algebra with respect to which the maps 
\begin{equation}
d\phi_{(u)}^{(-\eta)}\,\longmapsto \,\left\langle d \widehat{\Psi}_{t,v}^{(\eta)},\, d\phi_{(u)}^{(-\eta)} \right\rangle_{(t,\,\eta)}\,
\end{equation}
are measurable. Then, for $\eta\,>\,\tfrac{dim\,\Sigma\times M}{2}$, \;   $\mathcal{H}_{t,\, \phi_{cm}}^{-\eta}(\Sigma,M)$ is endowed with a finite Borel measure $\mathcal{Q}_{\mathcal{H}_t}$, characterized by its Fourier transform according to
\\
\begin{eqnarray}
\label{functmeas}
&&\int_{\mathcal{H}_{t}^{-\eta}}\,e^{\sqrt{-1}\,
\left\langle d \widehat{\Psi}_{t,v}^{(\eta)},\, d\phi_{(u)}^{(-\eta)} \right\rangle_{(t,\,\eta)}\,}\;\,
\mathcal{Q}_{\mathcal{H}_t}\left[d\phi_{(u)}^{(-\eta)}\right]\\
\nonumber\\
&&\,=\, e^{-\,\tfrac{1}{2}\,\left\langle d \widehat{\Psi}_{t,v},\,
d \widehat{\Psi}_{t,v} \right\rangle_{p_t(d\omega)\otimes d\mu_\gamma}}\,=\,
e^{-\,\tfrac{1}{a}\,\mathcal{E}[\widehat{\Psi}_{t,v}]}\;,\nonumber
\end{eqnarray} 
\\
\noindent
with 
$\mathcal{Q}_{\mathcal{H}_t}\left[\mathcal{H}_{t,\, \phi_{cm}}^{-\eta}(\Sigma,M)\right]\,=\,1$, and 
where $\mathcal{E}[\widehat{\Psi}_{t,v}]$ is the energy norm (\ref{energ1}). 
\end{prop}
\begin{proof}
By the Sobolev embedding theorem, if $\eta\,>\,\tfrac{dim\,\Sigma\times M}{2}$, the space $\mathcal{H}_{t}^{\eta}$ is continuously embedded into the set  of continuous and (bounded) 
sections $C^0_{\Sigma\times M}:=\,C^{0}\left(\Sigma\,\times\,M,\;
T^*\Sigma\otimes \,\left(\phi^{-1}_{cm}TM\,\boxtimes\,TM\right)\right)$. Let
\begin{equation}
\mathcal{CYL}\left(\mathcal{H}_{t}^{\eta}\right)\,:=\,\left\{\left.T\in\mathbb{L}(\mathcal{H}_{t}^{\eta},H)\right|\,
dim\,H\,<\infty,\,T(\mathcal{H}_{t}^{\eta})=H \right\}
\end{equation}
denote the set of linear maps from $\mathcal{H}_{t}^{\eta}$ onto finite dimensional Hilbert spaces $H\subset \mathcal{H}_{t}$. Let $d\lambda _{T(\mathcal{H}_{t}^{\eta})}$ be the Lebesgue measure on $T(\mathcal{H}_{t}^{\eta})$, and denote by $\langle \cdot,\,\cdot  \rangle_{T(\mathcal{H}_{t}^{\eta})}$ the restriction to $T(\mathcal{H}_{t}^{\eta})$ of the $\mathcal{H}_{t}^{\eta}$ inner product. With this notation along the way, it is easily checked that the family of  Gaussian measures
\begin{eqnarray}
&&\;\;\;\;\;\;\;\;\;\mathcal{Q}_{T(\mathcal{H}_{t}^{\eta})}\\
\nonumber\\
&&:=(2\pi)^{-\,\tfrac{dim(T(\mathcal{H}_{t}^{\eta}))}{2}}\;e^{-\,\tfrac{1}{2}\,\left\langle d \widehat{\Psi}_{t,v}^{(\eta)},\,
d \widehat{\Psi}_{t,v}^{(\eta)} \right\rangle_{T(\mathcal{H}_{t}^{\eta})} }\,d\lambda _{T(\mathcal{H}_{t}^{\eta})}\;,\;\;\;\;\;T\in \mathcal{CYL}\left(\mathcal{H}_{t}^{\eta}\right)\;,\nonumber
\end{eqnarray}
defines a cylinder set measure on $\mathcal{H}_{t}^{\eta}\cap C^{0}_{\Sigma\times M}$. The linear injective map with dense range $\mathcal{H}_{t,\, \phi_{cm}}(\Sigma,M)\,
\hookrightarrow \,\mathcal{H}_{t,\, \phi_{cm}}^{-\eta}(\Sigma,M)$ defined by the embedding (\ref{Hilbincl}), pushes forward $\mathcal{Q}_{T(\mathcal{H}_{t}^{\eta})}$ to a Gaussian, \emph{white noise}, measure $\mathcal{Q}_{\mathcal{H}_t}$ on $\mathcal{H}_{t,\, \phi_{cm}}^{-\eta}(\Sigma,M)$. By exploiting the extended  Bochner--Minlos theorem \cite{Driver}, \cite{Stroock}, \cite{Weitsman} for the triple of Hilbert spaces (\ref{Hilbincl}),  the  measure $\mathcal{Q}_{\mathcal{H}_t}$ can be characterized via its Fourier transform according to (\ref{functmeas}).
\end{proof}
\begin{rem}
Note that from the general theory of Gaussian measures on Banach spaces, (cf. \cite{7}, Chap.III), we have
\begin{equation}
\mathcal{E}[\widehat{\Psi}_{t,v}]\,=\, \frac{1}{2}\,
\int_{\mathcal{H}_{t}^{-\eta}}\,\left\langle d \widehat{\Psi}_{t,v}^{(\eta)},\,
d\phi_{(u)}^{(-\eta)} \right\rangle_{(t,\,\eta)}^2\;\,
\mathcal{Q}_{\mathcal{H}_t}\left[d\phi_{(u)}\right]\;.
\end{equation}
Moreover, according to Fernique theorem \cite{Fernique},  applied to the measurable and sub--additive function $\left\| \cdot  \right\| _{(t,\,\eta)}$, there exists $\rho \in (0,\infty)$ such that 
\begin{equation}
\label{fernique}
\int_{\mathcal{H}_{t}^{-\eta}}\,e^{\left[\rho\,\left\| d\phi_{(u)}^{(-\eta)} \right\| _{(t,\,\eta)}^2\right]}\;\,
\mathcal{Q}_{\mathcal{H}_t}\left[d\phi_{(u)}\right]\,<\,\infty \;.
\end{equation} 
\end{rem}
\subsection{Deformed harmonic energy as a large deviation functional}

Let us consider (\ref{functmeas}) for the center of mass section
\begin{equation}
\label{cmpsi}
d \widehat{\Psi}_{t,cm}\,:=\,\sum_{\mathfrak{j}=1}^q\,d \widehat{\Psi}_{t,(\mathfrak{j})}\;,   
\end{equation}
associated to the fields (\ref{kmap}). We  write the collection of vector fields $\{ v_{(\mathfrak{j})} \}_{\mathfrak{j}=1}^q$, generating the maps $\phi_{(\mathfrak{j})}$, as  $v_{(\mathfrak{j})}\,=\,\exp_{cm}^{-1}\,\phi_{(\mathfrak{j})}$, and  consider them, for $q\rightarrow \infty $, as a sequence of independent, identically distributed random variables on  
$\left(\mathcal{H}_{t}^{-\eta},\,\mathcal{B}_{t}^{-\eta},\,\mathcal{Q}_t \right)$.
 The associated normalized partial sums $\frac{1}{q}\,\sum_{\mathfrak{j}=1}^{q}\,v_{(\mathfrak{j})}
\left( \phi_{(\mathfrak{j})} \right)$ are distributed under $\prod_{\mathfrak{j}=1}^{q}\,\mathcal{Q}_t\left[d\phi_{(\mathfrak{j})}\right]$. Denote by
\begin{equation}
\mathcal{Q}_{(t,q)}\left(d\phi_{cm}\right) 
\end{equation}
the distribution of
\begin{equation}
\left(\phi_{(\mathfrak{1})},\ldots,\phi_{(\mathfrak{q})}\right)\longmapsto \frac{1}{q}\,\sum_{\mathfrak{j}=1}^{q}\,\exp_{cm}^{-1}\,\phi_{(\mathfrak{j})}
\end{equation}
under $\prod_{\mathfrak{j}=1}^{q}\,\mathcal{Q}_t\left[d\phi_{(\mathfrak{j})}\right]$. According to the strong law of large numbers we have that, as $q\nearrow \infty $,
\begin{equation}
\frac{1}{q}\,\sum_{\mathfrak{j}=1}^{q}\,\exp_{cm}^{-1}\,\phi_{(\mathfrak{j})}\,\longrightarrow \,0\;,
\end{equation} 
or, equivalently
\begin{equation}
\frac{1}{q}\,\sum_{\mathfrak{j}=1}^{q}\,
d\phi_{(\mathfrak{j})}^{(-\eta)}\,\longrightarrow \,d\phi_{cm}^{(-\eta)}\;,
\end{equation} 
for $\mathcal{Q}_{(t,q)}\left(d\phi_{cm}\right)$--almost every $d\phi_{(\mathfrak{j})}$. The deviant behavior of  $\tfrac{1}{q}\,\sum_{\mathfrak{j}=1}^{q}\,
d\phi_{(\mathfrak{j})}^{(-\eta)}$, with respect to an exponential rate of convergence to the center of mass $d\phi_{cm}^{(-\eta)}$, is governed by the large deviations of the family of distributions 
\begin{equation}
\left\{\mathcal{Q}_{(t,q)}\left(d\phi_{cm}\right)\,|\,q\,>\,1\right\}
\end{equation}
according to the following result
\begin{prop}
\label{devbehavior}
Let $J: \mathcal{H}_t\hookrightarrow \mathcal{H}_t^{-\eta}$ denote the inclusion map defined by Lemma \ref{denseincl}. For $d\phi_{cm}^{(-\eta)}\in J\left(\mathcal{H}_t \right)\cap \mathcal{H}_t^{-\eta}$ let 
$d \widehat{\Psi}_{t,cm}\,=\,J^{-1}\left(d\phi_{cm}^{(-\eta)} \right)$, then
\begin{eqnarray}
\Gamma_t(d\phi_{cm}^{(-\eta)})\,&:=&\,\tfrac{1}{a}\,\mathcal{E}[\widehat{\Psi}_{t,cm}]\;\;\;if\;\;d\phi_{cm}^{(-\eta)}\in J\left(\mathcal{H}_t \right)\cap \mathcal{H}_t^{-\eta}\\
&:=&\,\;\;\;\;\infty \;\;\;\;\;\;\;\;\;if\;\;d\phi_{cm}^{(-\eta)}\notin J\left(\mathcal{H}_t \right)\cap \mathcal{H}_t^{-\eta}\nonumber
\end{eqnarray}
is the rate functional governing the large deviations for the family of distributions $\left\{\mathcal{Q}_{(t,q)}\left(d\phi_{cm}\right)\,|\,q\,>\,1\right\}$.
\end{prop} 
\begin{proof}
Let us assume that $d\phi_{cm}^{(-\eta)}\,=\,J(d \widehat{\Psi}_{t,cm})$, then the rate functional governing the large deviations of $\left\{\mathcal{Q}_{(t,q)}\left(d\phi_{cm}\right)\,|\,q\,>\,1\right\}$ is provided by the Legendre transform \footnote{Large deviation theory for Gaussian measure over a generalized Wiener space is nicely discussed in Chap. 3 of \cite{7}.}\\  
\begin{eqnarray}
\label{largedev}
&&\Gamma_t\left(d\phi_{cm}^{(-\eta)}\right)\\
\nonumber\\
&&:=
\sup\left\{\left\langle
\zeta ^{(\eta)},J\left(d \widehat{\Psi}_{t,cm}\right) \right\rangle_{(t,\eta)}-\left.
\tfrac{1}{2}\parallel  J^*\left(\zeta ^{(\eta)}\right)\parallel_{\mathcal{H}_t}^2\,\right|\, \zeta ^{(\eta)}\in \mathcal{H}_{t}^{\eta} \right\}\;,\nonumber
\end{eqnarray} 
\\
\noindent
where $J^*: \mathcal{H}_{t}^{\eta}\longrightarrow \mathcal{H}_{t}$ is the adjoint map to $J$, and where, (see (\ref{functmeas}) for notation),
\\
\begin{equation}
\parallel  J^*\left(\zeta ^{(\eta)}\right)\parallel_{\mathcal{H}_t}^2:= \left\langle J^*\left(\zeta ^{(\eta)}\right),\,
J^*\left(\zeta ^{(\eta)}\right) \right\rangle_{p_t(d\omega)\otimes d\mu_\gamma}\;.
\end{equation}
 According to Lemma \ref{denseincl}, $J^*\left(\mathcal{H}_{t}^{\eta}\right)$ is dense in $\mathcal{H}_{t}$, hence, there is $d \widetilde{\Psi}_{t} \in \mathcal{H}_{t}$ such that we can write
\begin{eqnarray}
\left\langle
\zeta ^{(\eta)},J\left(d \widehat{\Psi}_{t,cm}\right) \right\rangle_{(t,\eta)}&=&
\left\langle
J^*\left(\zeta ^{(\eta)}\right),d \widehat{\Psi}_{t,cm} \right\rangle_{(t,\eta)}\\
\nonumber\\
&=&
\left\langle
d \widetilde{\Psi}_{t},d \widehat{\Psi}_{t,cm} \right\rangle_{(t,\eta)}\nonumber\;,
\end{eqnarray}
and (\ref{largedev})  reduces to
\begin{eqnarray}
\label{largedev1}
&&\Gamma_t\left(d\phi_{cm}^{(-\eta)}\right)\\
\nonumber\\
&&=
\sup\left\{\left\langle
d \widetilde{\Psi}_{t},d \widehat{\Psi}_{t,cm} \right\rangle_{(t,\eta)}-\left.
\tfrac{1}{2}\parallel  d \widetilde{\Psi}_{t}\parallel_{\mathcal{H}_t}^2\,\right|\,d \widetilde{\Psi}_{t}\in \mathcal{H}_{t} \right\}\nonumber\\
\nonumber\\
&&=\,\tfrac{1}{2}\parallel d \widehat{\Psi}_{t,cm}\parallel_{\mathcal{H}_t}^2\,=\,
\tfrac{1}{a}\,\mathcal{E}[\widehat{\Psi}_{t,cm}]\;.\nonumber
\end{eqnarray} 
To complete the proof we need to show that the assumption $\Gamma_t(d\phi_{cm}^{(-\eta)})\,<\,\infty $ necessarily implies that $d\phi_{cm}^{(-\eta)}\,=\,J(d \widehat{\Psi}_{t,cm})$. Set $\Gamma_t(d\phi_{cm}^{(-\eta)})\,\leq C^2/2$ for some $C>0$. Since on the unit ball $\parallel  J^*\left(\zeta ^{(\eta)}\right)\parallel_{\mathcal{H}_t}^2=1$, we have
\begin{equation}
\Gamma_t\left(d\phi_{cm}^{(-\eta)}\right)\,\geq\,
\frac{1}{2}\,\left\langle
\zeta ^{(\eta)},d\phi_{cm}^{(-\eta)}\right\rangle_{(t,\eta)}^2\;,\;\;\;\;\;\zeta ^{(\eta)}\in \mathcal{H}_{t}^{\eta}\;,
\end{equation}
and we can bound $\langle
\zeta ^{(\eta)},d\phi_{cm}^{(-\eta)}\rangle_{(t,\eta)}$ according to
\begin{equation}
\left| \left\langle
\zeta ^{(\eta)},d\phi_{cm}^{(-\eta)}\right\rangle_{(t,\eta)} \right|\,\leq\,C\,\parallel  J^*\left(\zeta ^{(\eta)}\right)\parallel_{\mathcal{H}_t}\;.
\end{equation}
The density  of $J^*\left(\mathcal{H}_{t}^{\eta}\right)$  in $\mathcal{H}_{t}$ and Riesz representation theorem then easily imply that there is a unique $d \widehat{\Psi}_{t,cm} \in \mathcal{H}_{t}$ such that $d\phi_{cm}^{(-\eta)}\,=\,J(d \widehat{\Psi}_{t,cm})$.
\end{proof}
This result characterizes the deformed  harmonic energy $\tfrac{1}{a}\,\mathcal{E}[\widehat{\Psi}_{t,cm}]$ as the large deviation functional  governing the  $\mathcal{O}(q)$ fluctuations of the distribution of the maps $d\phi_{(\mathfrak{j})}$, around the center of mass $d\phi_{cm}$, as compared to the $\mathcal{O}(\sqrt{q})$ Gaussian fluctuations sampled by the central limit theorem. Note that 
 in the path integral approach $\Gamma_t(d\phi_{cm}^{(-\eta)})$ plays the role of the \emph{effective action} of the non-linear $\sigma$ model formally derived, from the \emph{physical} Hilbert space measure $e^{-S_{cm}[v_{(\mathfrak{j})};\,a]}\,D\,[v_{(\mathfrak{j})}]$, by expressing the asymptotic series of associated Feynman amplitudes in terms of 1--Particle Irriducible diagrams \cite{gawedzki}.\\  
\\
\noindent 
According to  (\ref{energ1}), and (\ref{energ2}) we have that for $t\searrow 0$   
$d \widehat{\Psi}_{t,v}\longrightarrow\,d\phi_{(v)}$,\; 
$\mathcal{E}[\widehat{\Psi}_{t,v}]\rightarrow {E}[{\phi}_{(v)}]$, and (\ref{functmeas}) correspondingly extends to $t=0$\, as 
\begin{equation}
\label{tzero}
\int_{\mathcal{H}_{t=0}^{-\eta}}\,e^{\sqrt{-1}\,\left\langle\,d \widehat{\Psi}_{t=0,v}^{(\eta)}\,d\phi_{(u)}^{(-\eta)} \right\rangle_{(t=0,\,\eta)}}\,\mathcal{Q}_{t=0}\left[
  d\phi_{(u)}\right]\,=\, 
e^{-\,\tfrac{1}{a}\,{E}\left[{\phi}_{(u)}\right]}\;.
\end{equation}
If we consider the non--decreasing family of sub $\sigma$--algebras  $\{\mathcal{B}_{t}^{-\eta}\}_{t\geq0}$ generated by requiring that $\left\langle d \widehat{\Psi}_{t,v}^{(\eta)},\,d\phi_{(u)}^{(-\eta)} \right\rangle_{(t,\,\eta)}$ are $\{\mathcal{B}_{t}^{-\eta}\;:\;t\in[0,\infty )\}$--progressively measurable maps, then we can interpret the curve  of measures
\begin{equation}
\label{filtration}
t\,\longmapsto \,\left(\mathcal{H}_{t,},\,\mathcal{H}_{t,}^{-\eta},\,\mathcal{Q}_{\mathcal{H}_t} \right)\;,\;\;\;\;\;\;\;\; t\,\geq \,0\;,
\end{equation}
as describing the $\mathcal{Q}_{\mathcal{H}_t}$--distribution of the  fields $d\phi^{(-\eta)}$ as we deform the geometry of $(M,g,\,d\omega)$ along the heat kernel embedding.  In this sense (\ref{filtration}) provides a natural framework for a renormalization group analysis of the heat kernel embedding of the non--linear $\sigma$ model. The filtration $\left(\mathcal{H}_{t,},\,\mathcal{H}_{t,}^{-\eta},\,\mathcal{Q}_{\mathcal{H}_t} \right)$,\; ${t\geq0}$, is  described by the family  $t\longmapsto \mathcal{E}[\widehat{\Psi}_{t,v}]$ of associated characteristic functions. According to (\ref{energ1}) and Lemma \ref{figrande}, $t\longmapsto \mathcal{E}[\widehat{\Psi}_{t,v}]$ is  characterized by the running metric $t\mapsto h^{(q)}_t$,\,\emph{i.e.} by the coupling between the running metric  $g^{(\omega)}_t$ and the running dilaton $f^{(\omega)}_t$ 
describing the behavior of the warped harmonic energy functional ${E}[{\Phi}_{(q)},\,h^{(q)}_t]$ as the distances in $(M^q\times \times _{(f)}\mathbb{T}^q)$  are rescaled by the heat kernel embedding. The associated tangent vector, the so called \emph{beta}--function,
\begin{equation}
\label{betaflow1h}
\beta\left(h^{(q)}_t\right)\,:=\,\frac{d }{d t}\,h^{(q)}_t\;,
\end{equation}
plays a major role in renormalization group theory. Typically, in field theory one is able to compute the beta function only  up to a few leading order terms in the perturbative expansion of the effective action. Notwithstanding this limitation, the resulting truncated flow, associated with (\ref{betaflow1h}), can be exploited to study both the validity of the perturbative expansion and the nature of the possible fixed points of the renormalization group action \cite{gawedzki}. In our case, we have the exact expression for the \emph{effective action} $\mathcal{E}[\widehat{\Psi}_{t,v}]$, (or, equivalently for ${E}[{\Phi}_{(q)},\,h^{(q)}_t]$), however  $\left(\mathcal{H}_{t,},\,\mathcal{H}_{t,}^{-\eta},\,\mathcal{Q}_{\mathcal{H}_t} \right)_{t\geq0}$ is not the   
 natural physical renormalization group filtration of the dilatonic non--linear $\sigma$ model. Nonetheless, the beta function (\ref{betaflow1h}) computed along the heat kernel embedding turn out to be remarkably similar to the ones obtained, to leading order, by the formal perturbative path integral approach.

\section{Heat kernel embedding and Ricci flow}
\label{Bfunct}
According to (\ref{betaflow1h}),  the beta function  associated with the heat kernel embedding is provided by a time derivative of the flow of metrics  $(h,\,t)\mapsto h^{(q)}_t$, \emph{i.e.}
 \begin{equation}
\label{hTmetricflow}
\frac{d }{d t}\,h_t\,=\,\frac{d }{d t}\,g_t^{(\omega)}\,-
\,\frac{2}{q}\,e^{-\,\tfrac{2f_t^{(\omega)}}{q}}\,\delta_{\mathbb{T}^q}\,\frac{d }{d t}\,f_t^{(\omega)}\;.
\end{equation}
Since the whole geometry of $(M\times\mathbb{T}^q,\,h_t)$ is controlled by $(M, g_t^{(\omega)})$, we start with a detailed computation of the Riemannian curvatures associated with $(M, g_t^{(\omega)})$ in the smooth case, induced by the heat kernel embedding 
\begin{eqnarray}
\label{heatsmooth}
(M,d_g)\times \mathbb{R}_{>0}\,&\hookrightarrow&\, \left(\rm{Prob}_{ac}(M,g),\; d_{\,g}^{\,W}\, \right)\\
(z,\,t)\,\;\;\;\;\;&\longmapsto&\, \;\;\;\;\;p_t^{(\omega)}(\cdot,z)d\omega(\cdot )\;,\;\;\;\;t>0\;.\nonumber
\end{eqnarray}

\subsection{Geodesics and the Ricci curvature of $(M,\,g_t^{(\omega)})$}

Let us consider, for $t>0$, the curve of measures
\begin{equation}
\lambda\,\longmapsto\,\nu(c(\lambda),t)\,:=\, p^{(\omega)} _t(\cdot \, , c(\lambda))d\omega(\cdot )\;,\;\;\;t>0\;,
\end{equation}
associated with an absolutely continuous curve 
\begin{equation} 
c\,:\,[-\epsilon ,\,\epsilon ]\ni \lambda\mapsto c(\lambda)\in M\,,\;\;\;\; c(\lambda=0)=z,\;\;\;\epsilon >0\;,
\end{equation}
passing through the point $z\in M$. We have the following result locally characterizing the geodesics of $(M,\,g_t^{(\omega)})$. 
 \begin{prop}
 \label{hamjac}
 For any given $t>0$, let $\widetilde{\nabla}^{(t)}$ denote the Levi--Civita connection of $(M, g_t^{(\omega)})$. If the curve $c$ is a geodesic for $\widetilde{\nabla}^{(t)}$ then the the vector field $\widehat{\psi}_{(t,c(\lambda),c'(\lambda))}$
$\in T_{\nu(\lambda)}\rm{Prob}_{ac}(M,g)$, covering the 
curve $\lambda\mapsto \nu(\lambda):= p^{(\omega)} _t(\cdot \, , c(\lambda))d\omega(\cdot )$, satisfies the Hamilton--Jacobi equation
\begin{equation}
\label{H-J}
\frac{\partial }{\partial \lambda}\,\widehat{\psi}_{(t,c(\lambda),c'(\lambda) )}(y)\,+\,
\frac{1}{2}\,\left|\nabla^{(y)}\,  \widehat{\psi}_{(t,c(\lambda),c'(\lambda) )}(y)\right|^2_g\,=\,0\;,
\end{equation}
where $|\cdot |_g$ and $\nabla $ respectively denote 
the Riemannian norm and the Levi--Civita connection in the original $(M,g)$. 
\end{prop}
  
\begin{proof} If $c$ is a geodesic for $\widetilde{\nabla}^{(t)}$ then, according to (\ref{gt}) and Lemma \ref{velident},  we have
\begin{eqnarray}
\label{firstgeo}
 &&\left.\frac{d }{d \lambda}\right|_{c(\lambda)}\,
 g_t^{(\omega)}\left({c'}(\lambda),\,{c'}(\lambda) \right)\,=\,2\,g_t^{(\omega)}
\left(\widetilde{\nabla}_{{c'}(\lambda)}{c'}(\lambda),\,{c'}(\lambda) \right)\,=\,0\\
\nonumber\\
&&=\,\frac{d }{d \lambda}\,\int_M\,\left|\nabla^{(y)}\,  \widehat{\psi}_{(\lambda,t)}(y)\right|^2_g\,
p^{(\omega)} _t(y , c(\lambda))\,d\omega(y)\nonumber\\
\nonumber\\
&&=\,2\,\int_M\,\left(\nabla^i\widehat{\psi}_{(\lambda, t)}\,\tfrac{\partial  }{\partial  \lambda}\,\nabla _{i}\,\widehat{\psi}_{(\lambda,t)}\right)\,p^{(\omega)} _t(y , c(\lambda))\,d\omega(y)\nonumber\\
\nonumber\\
&&+\,\int_M\,\left|\nabla^{(y)}\,  \widehat{\psi}_{(\lambda, t)}(y)\right|^2_g\,\tfrac{\partial\,  }{\partial \lambda}\,
p^{(\omega)} _t(y , c(\lambda))
\,d\omega(y)\;,\nonumber
\end{eqnarray}
where we set for notational simplicity $\widehat{\psi}_{(\lambda, t)}\:=\,\widehat{\psi}_{(t,c(\lambda),c'(\lambda) )}$.
According to Lemma \ref{movHeatSource} and \ref{velident}, along the curve of measures $
\lambda\,\longmapsto\, p^{(\omega)} _t(\cdot \, , c(\lambda))d\omega(\cdot )$ we can write, for almost every $\lambda\in (-\epsilon,\,\epsilon)$,
\begin{equation}
\label{cpde2}\frac{\partial}
{\partial \lambda}\,
p^{(\omega)} _t(y , z)\,=\,-\,div_\omega^{(y)}\,
\left(\nabla \,\widehat{\psi}_{(\lambda, t)}(y)\,p^{(\omega)} _t(y , c(\lambda))\right)\;.
\end{equation}
By exploiting this relation, which we interpret in the distributional sense,  and integrating by parts (\ref{firstgeo}), we get
\begin{eqnarray}
\label{viv}
&&\int_M\nabla^i\widehat{\psi}_{(\lambda, t)}\left(\frac{\partial  }{\partial  \lambda}\nabla _{i}\widehat{\psi}_{(\lambda, t)}+\frac{1}{2}
\nabla_i\left|\nabla  \widehat{\psi}_{(\lambda,t)}\right|^2_g\right)p^{(\omega)} _t(y , c(\lambda))d\omega(y)\\
\nonumber\\
&&=\,\int_M\nabla^i\widehat{\psi}_{(\lambda, t)}\nabla _{i}\left(\frac{\partial  }{\partial  \lambda}\widehat{\psi}_{(\lambda, t)}+\frac{1}{2}
\left|\nabla  \widehat{\psi}_{(\lambda,t)}\right|^2_g\right)p^{(\omega)} _t(y , c(\lambda))d\omega(y)\nonumber\\
\nonumber\\
&&=-\int_M\,\bigtriangleup _{p_t(d\omega)}\widehat{\psi}_{(\lambda, t)}\,\left(\frac{\partial  }{\partial  \lambda}\widehat{\psi}_{(\lambda, t)}+\frac{1}{2}
\left|\nabla \widehat{\psi}_{(\lambda, t)}\right|^2_g\right)p^{(\omega)} _t(y , c(\lambda))d\omega
=0\;.\nonumber
\end{eqnarray}
According to (\ref{carlo12Hnew}) we have
\begin{equation}
\bigtriangleup _{p _t\,(d\omega)}^{(y)}\,\widehat{\psi}_{(t,c(\lambda),c'(\lambda) )}\,=\,-\,c'(\lambda)\,\cdot \nabla^{(c(\lambda))}\,\ln\,p^{(\omega)} _t(y , c(\lambda))\;,
\end{equation}
and the last line of  (\ref{viv}) becomes
\begin{equation}
\int_M\,\left(\frac{\partial  }{\partial  \lambda}\widehat{\psi}_{(\lambda, t)}+\frac{1}{2}
\left|\nabla \widehat{\psi}_{(\lambda, t)}\right|^2_g\right)
c'(\lambda)\,\cdot \nabla^{(c(\lambda))}\,p^{(\omega)} _t(y , c(\lambda))
d\omega =0\;.
\end{equation}
This relation must hold for any (germ of) geodesic $(-\epsilon , \epsilon )\ni \lambda\longmapsto c(\lambda)$ in $(M,\,g_t^{(\omega)})$ passing through $z$, \emph{i.e.} for any $c'(\lambda)\not\equiv 0$. According to Proposition \ref{CNprop} this implies
\begin{equation}
\tfrac{\partial  }{\partial  \lambda}\,\widehat{\psi}_{(\lambda, t)}+\frac{1}{2}
\left|\nabla\,  \widehat{\psi}_{(\lambda, t)}\right|^2_g\,=\,0\;.
\end{equation}
\end{proof}
As an immediate consequence of this result and of the characterization of the geodesics of $(\rm{Prob}_{ac}(M,g),\,d^W_g)$ with respect of the Otto inner product  we have 
\begin{prop}
\label{exchconn}
 For any fixed $t>0$, let $\Upsilon _t\left((M,\,g) \right)$ be the  image of $(M,\,g)$ in the smooth probability space $(\rm{Prob}_{ac}(M,g),\,d^W_g)$ under the heat kernel embedding.  If we endow $\rm{Prob}_{ac}(M,g)$ with the the Otto Riemannian connection $\overline{\nabla }$ associated with the inner product $\langle \cdot, \cdot  \rangle_{p_t(d\omega)}$, then the geodesics of $(M,\,g_t^{(\omega)})$ can be identified with the geodesics of  $[\Upsilon _t((M,\,g)),\,\overline{\nabla }]$. Hence for $t>0$,\; $[\Upsilon _t((M,\,g)),\,\overline{\nabla }]$ is a totally geodesic submanifold of $\rm{Prob}_{ac}(M,g)$.
\end{prop} 
\begin{proof} The Otto inner product(\ref{inner}) associated to the heat kernel embedding is the metric tensor $g_t^{(\omega)}\left(\cdot ,\,\cdot \right)$ evaluated at $z\in M$,
\begin{eqnarray}
&&\,g_t^{(\omega)}\left(\vee _{\varphi} ,\,\vee _{\zeta} \right)\,=\,
\langle \nabla \varphi,\,\nabla \zeta \rangle_{p_t(d\omega)}\\
\nonumber\\
&&\,:=\int_M\,g^{ik}(y)\nabla_i^{(y)} \varphi\,
\nabla_k^{(y)} \zeta\;p^{(\omega)} _t(y, z)\,d\omega(y)\;,\nonumber
\end{eqnarray}
where
$\vee _{\varphi},\,\vee _{\zeta}\in T_{p_t(d\omega)}{\rm Prob}_{ac}(M,g)$ are the vector fields on ${\rm Prob}_{ac}(M,g)$ defined by the potentials $\varphi,\,\zeta\in C^{\infty }(M ,\mathbb{R})/\mathbb{R}$ under the isomorphism (\ref{veciso}).
According to Theorem \ref{lotth1}, the corresponding connection $\overline{\nabla}$  is provided by
\begin{equation}
\label{conn2}
\langle \overline{\bigtriangledown }_{\vee _{\varphi}}\,\vee _{\chi},\,\vee _{\zeta} \rangle_{p_t(d\omega)}\,=\,\int_M\,\nabla_i\varphi(y)\,\nabla^i \nabla^k\chi(y)\,\nabla _k\zeta(y)\,\,p^{(\omega)} _t(y, z)\,d\omega(y)\;,  
\end{equation}
  The lemma directly follows from the result proved by J. Lott  (cf. Prop. 4 of \cite{lott1b}) and by F. Otto and C. Villani \cite{otto2}, according to which if a vector field $\widehat{\psi}_{(t,c(\lambda),c'(\lambda))}$
$\in T_{\nu(\lambda)}\rm{Prob}_{ac}(M,g)$, tangent to a smooth curve of measures $\lambda\,\longmapsto\, p^{(\omega)} _t(\cdot \, , c(\lambda))d\omega(\cdot )$, satisfies the Hamilton--Jacobi equation (\ref{H-J}) then the curve in question  is a geodesic in the smooth probability space  $({\rm Prob}_{ac}(M,g), \overline{\nabla })$. Since in the smooth setting minimizing geodesics are unique \cite{mccan}, (cf. Lemma \ref{McClemma}),  the stated result  follows.  
\end{proof}

According to Lemma \ref{exchconn} the Levi--Civita connection $\widetilde{\nabla }$ of  $(M, g_t^{(\omega)})$ can be identified with the induced $\overline{\nabla }$ connection on  $\Upsilon _t\left((M,\,g) \right)$, in particular 
we can compute the curvature of $(M, g_t^{(\omega)})$ by restricting to $\Upsilon _t\left((M,\,g) \right)$ the Riemannian curvature of the Wasserstein space $({\rm Prob}_{ac}(M,g),\,\langle\cdot ,\cdot \rangle_{p_t(d\omega)})$. As a preliminary result, let us consider the tensor field $T_{\psi_{(U)}\psi_{(V)}}$ defined by (\ref{Tvector}) and  associated to the measure $p_T^{(\omega)}\,d\omega$. Since $\Upsilon _t\left((M,\,g) \right)\subset\rm{Prob}(M, g)$ is a totally geodesic submanifold we have
\begin{lem}
Let $U_\perp , V_\perp\,\in T_z(M)$ and let $\widehat{\psi}_{(U)}:=\widehat{\psi}_{t, z,U(z)}$ and    
$\widehat{\psi}_{(V)}:=\widehat{\psi}_{t, z,V(z)}$ be the corresponding potentials solutions of (\ref{carlo12H}). Then, for any $t>0$\, and\, $\forall\; U_\perp , V_\perp\,\in T_z(M)$,
\begin{equation}
\label{Tvector20}
T_{\psi_{(U)}\psi_{(V)}}^b\,:=\,(I-\Pi_{p_t(d\omega)})\left(\nabla_a\widehat{\psi}_{(U)}\,\nabla^a \nabla^b\widehat{\psi}_{(V)}\,\right)\,=\,0\;,
\end{equation}
where $\Pi_{p_t(d\omega)}$ is  the projection operator (\ref{proj2}) associated with $p_t(d\omega)$.
\end{lem} 
\begin{proof}
Let us consider $U_\perp , V_\perp\,, W_\perp \in T_z(M)$ as  elements of $T_e\,\mathcal{D}iff(M)$, the tangent space to $\mathcal{D}iff(M)$  at the identity map, then  we can introduce the connection  $\breve{\nabla }$ associated to the weak Riemannian structure induced on $\mathcal{D}iff(M)$ by the $L^2(p_t(d\omega))$ inner product according to 
\begin{equation}
\label{}
\left\langle \breve{\nabla }_{U_\perp }\,V_\perp ,\,W_\perp \right\rangle_{p_t(\omega)}
=\,\int_M\,g\left(\nabla ^{(g)}_{U_\perp }V_\perp ,\,W_\perp \right)\,p^{(\omega)} _t(y, z)\,d\omega(y)\;,
\end{equation}
(see (\ref{diffconn})). By considering the restriction of $\breve{\nabla }$ to $\rm{Prob}(M, g)$ we have
\begin{eqnarray}
U_\perp&\longmapsto& \nabla\widehat{\psi}_{(U)}\\
V_\perp&\longmapsto& \nabla\widehat{\psi}_{(V)}\nonumber\\
\breve{\nabla }_{U_\perp }\,V_\perp&\longmapsto& \breve{\nabla }_{\nabla\widehat{\psi}_{(U)}}\,\nabla\widehat{\psi}_{(V)}\nonumber\\
\Pi_{p_t(d\omega)}\left(\nabla_a\widehat{\psi}_{(U)}\,\nabla^a \nabla^b\widehat{\psi}_{(V)}\,\right)
&\longmapsto& \overline{\nabla }_{\nabla\widehat{\psi}_{(U)}}\,\nabla\widehat{\psi}_{(V)}\nonumber;,
\end{eqnarray}
in the weak sense, hence we can formally  
rewrite (\ref{Tvector20}) as 
\begin{equation}
T_{\psi_{(U)}\psi_{(V)}}\,=\,\breve{\nabla }_{\nabla\widehat{\psi}_{(U)}}\,\nabla\widehat{\psi}_{(V)} \,-\,\overline{\nabla }_{\nabla\widehat{\psi}_{(U)}}\,\nabla\widehat{\psi}_{(V)}\;.
\end{equation}
Thus, $T_{\psi_{(U)}\psi_{(V)}}$ can be interpreted as the second fundamental form of the immersion 
$[\Upsilon _t((M,\,g)),\,\overline{\nabla }]\hookrightarrow \rm{Prob}_{ac}(M,g)\subset \mathcal{D}iff(M)$. Since for any $t>0$  the immersion is totally geodesic,  $T_{\psi_{(U)}\psi_{(V)}}=0$ follows. On less formal ground, we can prove the stated result by introducing  normal geodesic coordinates with respect  to $(M,\,g_t^{(\omega)})$, centered at $z\in M$. To this end, let $\{e_{(i)}(z)\}$ denote a basis in $T_z(M)$ and 
let us consider a set of curves in  $(M,\,g_t^{(\omega)})$,  issuing from $z\in M$ with tangent vector $\{c'_{(i)}(0)\}\,:=\,\{e_{(i)}(z)\}$,
\begin{equation} 
\label{igeod}
c_{(i)}\,:\,[-\epsilon ,\,\epsilon ]\ni \lambda\mapsto c_{(i)}(\lambda)\in M\,,\;\;\;\; c_{(i)}(0)=z,\;\;\;
 i=\,1,\ldots,n\;.
\end{equation}
According to Proposition \ref{hamjac}, Lemma \ref{wassdinter} and \ref{velident},\;  these curves are geodesics of $(M,\,g_t^{(\omega)})$, for a given $t>0$,  iff the  corresponding curves of probability measures and potentials 
\begin{eqnarray}
\lambda\,&\longmapsto &\, p^{(\omega)} _t(y, c_{(i)}(\lambda))\,d\omega(y)\nonumber\\
\\
\lambda\,&\longmapsto &\,\widehat{\psi}^{(i)}_{(t, \lambda)}\,:=\, \widehat{\psi}^{(i)}_{(t, c_{(i)}(\lambda), c'_{(i)}(\lambda))}\nonumber
\end{eqnarray}
evolve according to the relations, (defining the optimal transport map),
\begin{eqnarray}
&&\frac{\partial }{\partial \lambda}\,p^{(\omega)} _t(y , c_{(i)}(\lambda))\,+\,div_\omega^{(y)}\,
\left(p^{(\omega)} _t(y , c_{(i)}(\lambda))\,\nabla ^{(y)}\widehat{\psi}^{(i)}_{(t, \lambda)})\right)\,=\,0\nonumber\\
\label{optmap}\\
&&\frac{\partial }{\partial \lambda}\,\widehat{\psi}_{(t, \lambda)}^{(i)}(y)\,+\,
\frac{1}{2}\,\left|\nabla^{(y)}\,  \widehat{\psi}_{(t, \lambda)}^{(i)}(y)\right|^2_g\,=\,0\;,\nonumber
\end{eqnarray}
(see (\ref{Liecont}) and (\ref{H-J})), with the initial conditions $p^{(\omega)} _t(y , c_{(i)}(\lambda=0))=p^{(\omega)} _t(y , z)$ 
and $\widehat{\psi}^{(i)}_{(t, \lambda=0)}\,:=\, \widehat{\psi}^{(i)}_{(t, z, c'_{(i)}(0))}$. In such a framework, normal geodetic coordinates at $z$ can be characterized by: \emph{(i)}\, Choosing the initial $\widehat{\psi}^{(i)}_{(t, \lambda=0)}$, (hence the $c'_{(i)}(0)$), in such a way that the connection coefficients 
(\ref{conn2})  vanish at $z$;\,\emph{(ii)}\, Propagating the chosen $\widehat{\psi}^{(i)}_{(t, \lambda=0)}$ according to the Hamilton--Jacobi equation in (\ref{optmap}), (see (\ref{H-J})), so as to obtain the curve of potentials 
$\lambda\mapsto\widehat{\psi}^{(i)}_{(t, \lambda)}$; \emph{(i)}\,Use the potentials so obtained to evolve the initial density $p^{(\omega)} _t(y , z)$ by means of  the continuity equation in (\ref{optmap}).  Since 
\begin{equation}
\frac{\partial }{\partial \lambda}\,p^{(\omega)} _t(y , c_{(i)}(\lambda))=c'_{(i)}(\lambda)\,\cdot 
\nabla _{c_{(i)}(\lambda)}\,p^{(\omega)} _t(y , c_{(i)}(\lambda))\;,
\end{equation}
the procedure described implies that along (\ref{optmap}) we recover the relation
\begin{equation}
div_\omega^{(y)}\,
\left(p^{(\omega)} _t(y , c_{(i)}(\lambda))\,\nabla ^{(y)}\widehat{\psi}^{(i)}_{(t, \lambda)})\right)\,=\,
-\,c'_{(i)}(\lambda)\,\cdot 
\nabla _{c_{(i)}(\lambda)}\,p^{(\omega)} _t(y , c_{(i)}(\lambda))\;,
\end{equation}
connecting the tangent vector $c'_{(i)}(\lambda)$ to the potential $\widehat{\psi}^{(i)}_{(t, \lambda)}$, and hence identifying them with $\widehat{\psi}^{(i)}_{(t, c_{(i)}(\lambda), c'_{(i)}(\lambda))}$.
For our purposes it is sufficient to consider the step \emph{(i)}. To this end we set $\widehat{\psi}^{\;(i)}:=\widehat{\psi}^{\;(i)}_{(t, \lambda=0)}$ for notational simplicity, and require that the 
$\{\widehat{\psi}^{\;(i)}\}_{i=1,\ldots, n}$ are such that 
\begin{eqnarray}
\label{LCcoord}
&&\langle \overline{\bigtriangledown }_{\vee _{\psi_{(i)}}}\,\vee _{\psi_{(k)}},\,\vee _{\zeta} \rangle_{p_t(d\omega)}(z)\\
\nonumber\\
&&\,=\,\int_M\,\nabla_a\widehat{\psi}^{\;(i)}(y)\,\nabla^a \nabla^b\widehat{\psi}^{\;(k)}(y)\,\nabla _b\zeta(y)\,\,p^{(\omega)} _t(y, z)\,d\omega(y)\nonumber\\
\nonumber\\
&&=\int_M\Pi _{p_t(d\omega)}\left(\nabla_a\widehat{\psi}^{\,(i)}(y)\nabla^a \nabla^b\widehat{\psi}^{\,(k)}(y)\right)\nabla _b\zeta(y)\,p^{(\omega)} _t(y, z)d\omega(y)
\,=\,0\;, \nonumber 
\end{eqnarray}
for any choice of $\zeta\in C^{\infty}(M, \mathbb{R})/\mathbb{R}$. 
Since  $\overline{\bigtriangledown }$ is the Levi--Civita connection for $(M,\,g_t^{(\omega)})$, (\ref{LCcoord}) is symmetric under the exchange $\vee _{\psi_{(i)}}\longleftrightarrow \vee _{\psi_{(k)}}$ and we can  write
\begin{eqnarray}
\label{}
&&\langle \overline{\bigtriangledown }_{\vee _{\psi_{(i)}}}\,\vee _{\psi_{(k)}},\,\vee _{\zeta} \rangle_{p_t(d\omega)}(z)\\
\nonumber\\
&&\,=\frac{1}{2}\,\int_M\,\Pi_{p_t(d\omega)}\left(\nabla^b\left(\nabla_a\widehat{\psi}^{\;(i)}\,\nabla^a \widehat{\psi}^{\;(k)}\right)\right)\nabla _b\zeta\,\,p^{(\omega)} _t(y, z)\,d\omega(y) \nonumber\\
\nonumber\\
&&\,=\frac{1}{2}\,\int_M\,\nabla^b\left(\nabla_a\widehat{\psi}^{\;(i)}\,\nabla^a \widehat{\psi}^{\;(k)}\right)\nabla _b\zeta\,\,p^{(\omega)} _t(y, z)\,d\omega(y)\;, \nonumber  
\end{eqnarray}
where we have exploited the fact that $\Pi_{p_t(d\omega)}$ projects onto the gradient vector fields.  
Integrating by parts with respect to $p^{(\omega)} _t(y, z)\,d\omega(y)$ we get
\begin{eqnarray}
\label{}
&&\langle \overline{\bigtriangledown }_{\vee _{\psi_{(i)}}}\,\vee _{\psi_{(k)}},\,\vee _{\zeta} \rangle_{p_t(d\omega)}(z)\\
\nonumber\\
&&\,=-\frac{1}{2}\,\int_M\,\bigtriangleup _{p_t(d\omega)}\left(\nabla_a\widehat{\psi}^{\;(i)}\,
\nabla^a \widehat{\psi}^{\;(k)}\right)\,\zeta(y)\,\,p^{(\omega)} _t(y, z)\,d\omega(y)\;, \nonumber 
\end{eqnarray} 
so that the condition (\ref{LCcoord}) reduces to
\begin{equation}
\label{ellgeod}
\bigtriangleup _{p_t(d\omega)}\,g\left(\nabla\widehat{\psi}^{\;(i)},\,\nabla\widehat{\psi}^{\;(k)}\right)\,=\,0
\end{equation}  
in the distributional sense. Thus, almost everywhere   $g\left(\nabla\widehat{\psi}^{\;(i)},\,\nabla\widehat{\psi}^{\;(k)}\right)= C_{ik}$,  for some constant $C_{ik}$. We normalize $C_{ik}$ and choose the geodesic coordinate initial data $\widehat{\psi}^{\;(i)}:=\widehat{\psi}^{(i)}_{(t, \lambda=0)}$ according to
\begin{equation}
\label{gpsiorth}
g\left(\nabla\widehat{\psi}^{\;(i)},\,\nabla\widehat{\psi}^{\;(k)}\right)= \delta_{ik}\;.
\end{equation}
In particular we get that, in a suitable weak sense, each $\widehat{\psi}^{\;(i)}$ should be a solution of the eikonal equation  
\begin{equation}
\label{eikonal}
\left|\nabla\widehat{\psi}^{\;(i)} \right|_g^2\,=\,1\;,
\end{equation} 
on $(M,g)$.   An explicit characterization of the solutions of the eikonal equation  can be obtained in terms of the distance function of $(M, g)$, (see \emph{e.g.} \cite{petersen}, and \cite{nirenberg}, \cite{mantegazzaHJ} for a detailed analysis). To apply the theory in our case, let $\mathcal{U}_z$ be a star--shaped open set of $0\in T_zM$ and let $\exp_z^{(g)}\,:\mathcal{U}_z\cap T_zM\longrightarrow (M, g)$ denote the exponential map of $(M, g)$ at $z\in M$. Let $E_{(i)}(z)$,\,$i=1,\ldots,n$,\, be a orthonormal basis of $(T_z(M),\,g)$. For $y\in M$ within the cut locus $Cut(z,g)$ of $z$ in $(M,g)$, let $\{z^i(y):=g(E_{(i)},\,(\exp^{(g)}_z)^{-1}(y))\}$ be its normal geodesic coordinates in  $T_zM$. For $y\in M\backslash (Cut(z,g)\cup \{z\})$ the distance function $d_g(z, y)$ is $C^\infty$, and if we denote by $d_g(0, (\exp^{(g)}_z)^{-1}(y))$ its pull--back in $(T_z(M),\,g)$ we have
\begin{eqnarray}
d_g^2\left(0, (\exp^{(g)}_z)^{-1}(y)\right)\,&=&\,\sum_{i=1}^n\,d_g^2\left(0, 
g\left(E_{(i)},\,(\exp^{(g)}_z)^{-1}(y)\right)\right)\nonumber\\
&=&\,\sum_{i=1}^n\,(z^i(y))^2\;.
\end{eqnarray}
Hence we compute   $\tfrac{\partial }{\partial z^k }\,d_g\left(0, 
g\left(E_{(i)},\,(\exp^{(g)}_z)^{-1}(y)\right)\right)\,=\,\delta_{(i)}^k$ and we can write $\left|\nabla ^{(y)}\, d_g\left(0, 
g\left(E_{(i)},\,(\exp^{(g)}_z)^{-1}(y)\right)\right)   \right|_g^2\,=\,1$,\,in $M\backslash Cut(z,g)\cup \{z\}$. Hence $\widehat{\psi}^{\;(i)}(y):=d_g\left(0, 
g\left(E_{(i)},\,(\exp^{(g)}_z)^{-1}(y)\right)\right)$ naturally appears as a candidate solution of the eikonal equation in the open set $M\backslash Cut(z,g)\cup \{z\}$. To handle the presence of the singular domain $Cut(z,g)\cup \{z\}$ where the distance function is not differentiable we can interpret the solution of (\ref{eikonal}) in the viscosity sense \cite{nirenberg}, \cite{mantegazzaHJ}. Thus,  if we normalize ${\psi}^{\;(i)}$ to a vanishing $p_t^{(\omega)}\,d\omega$ average, we have that
for any $t>0$ the functions
\begin{eqnarray}
&&\widehat{\psi}^{\;(i)}(y)\,=\,d_g\left(0, 
g\left(E_{(i)},\,(\exp^{(g)}_z)^{-1}(y)\right)\right)\nonumber\\
\label{eiksol}\\
&&\,-\,\int_{(\exp^{(g)}_z)^{-1}(M)}
d_g\left(0, 
g\left(E_{(i)},\,(\exp^{(g)}_z)^{-1}(x)\right)\right)\,(\exp^{(g)}_z)^*(p_t^{(\omega)}\,d\omega)
\nonumber
\end{eqnarray}
 generate normal geodesic coordinates  
\begin{eqnarray}
&&g_t^{(\omega)}\left(c'_{(i)}(0),\,c'_{(k)}(0)\right)(z)\,=\,\delta_{ik}\nonumber\\
\\
&&\langle \overline{\bigtriangledown }_{\vee _{\psi_{(i)}}}\,\vee _{\psi_{(k)}},\,\vee _{\zeta} \rangle_{p_t(d\omega)}(z)\,=\,0\nonumber\;,
\end{eqnarray}
in $(M\backslash (Cut(z,g)\cup \{z\}),\,g_t^{(\omega)})$ via the steps \emph{(ii)} and \emph{(iii)} described above. In such normal coordinates  we compute  $\nabla^a \nabla^b\widehat{\psi}^{\;(j)}=0$ almost 
everywhere in $M\backslash (Cut(z,g)\cup \{z\})$, hence
\begin{equation}
\label{STvector2}
T_{\psi_{(i)}\psi_{(j)}}^b\,=\,(I-\Pi _{(p_t(d\omega))})\left(\nabla_a\widehat{\psi}^{\;(i)}\,\nabla^a \nabla^b\widehat{\psi}^{\;(j)}\,\right)\,=\,0\;,
\end{equation}
in the weak sense, and the stated lemma follows.
\end{proof}

 With these preliminary results along the way a direct application of Theorem \ref{LottCurv} provides  
\begin{prop}  
\label{Riemcurv}
For any fixed $t>0$,\, $z\in M$, and $U_\perp,\,V_\perp,\,W_\perp,\,Z_\perp \in T_z(M)$, let us denote by $\nabla \widehat{\psi}_{(t,U)} :=\nabla _{(y)}\,\widehat{\psi}_{(t,z,U_\perp )}$, $\nabla \widehat{\psi}_{(t,V)} :=\nabla_{(y)}\,\widehat{\psi}_{(t,z,V_\perp )}$, $\nabla \widehat{\psi}_{(t,W)} :=\nabla_{(y)}\,\widehat{\psi}_{(t,z,W_\perp )}$, and $\nabla \widehat{\psi}_{(t,Z)} :=\nabla_{(y)}\,\widehat{\psi}_{(t,z,Z_\perp )}$ the corresponding vector fields defined by the weighted heat kernel injection map (\ref{tangmappsi}). Then the Riemannian curvature operator $\widetilde{R}m^{(t)}\left(U_\perp ,\,V_\perp \right)\,W_\perp $\; of $(M,\,g_t^{(\omega)})$ at the  point $z\in M$ is  given by\\
\begin{eqnarray}
\label{Riemop}
&&g_t^{(\omega)}\left(\widetilde{R}m^{(t)}\left(U_\perp,\,V_\perp \right)\,W_\perp,\,
 Z_\perp\right)(z)\\
\nonumber\\
&&\,=\,\int_M\,g\left(Rm\left(\nabla\widehat{\psi}_{(t,U)},\,\nabla\widehat{\psi}_{(t,V)}\right)\,\nabla\widehat{\psi}_{(t,W)},\,
\nabla\widehat{\psi}_{(t,Z)}  \right)(y)\,p^{(\omega)} _t(y, z)\,d\omega(y)\nonumber\;,
\end{eqnarray}
where $Rm\left(\circ  ,\,\circ \right)\,\circ $ denotes the Riemann tensor of $(M,g)$.
\end{prop}

By exploiting this result and the above characterization of the normal coordinates in $(M,\,g_t^{(\omega)})$ we can readily compute the Ricci tensor  associated to $(M,\,g_t^{(\omega)})$ according to
\begin{prop}
\label{smRiccT} 
For any $t>0$,\, $z\in M$, and $U_\perp $,\,$W_\perp $\,$\in\,T_{z}\,M$, the Ricci curvature of $(M,\,g_t^{(\omega)})$ in the direction of the $2$--plane spanned by the vectors $U_\perp, \,W_\perp \in T_zM$ is provided by
\begin{eqnarray}
&& \widetilde{R}ic^{(t)}(U_\perp, W_\perp)(z)\,:=\, {trace}\,\left(e_{(i)}\longmapsto \widetilde{R}m^{(t)}\left(e_{(i)},\,U_\perp\right)\,W_\perp \right)(z)\nonumber\\
\label{Riccicompt}\\
&&\,=\,\int_M\,Ric\left(\nabla\widehat{\psi}_{(t,U)},\,
\nabla\widehat{\psi}_{(t,W)}\right)(y)\,p^{(\omega)} _t(y, z)\,d\omega(y)
\nonumber\;, 
\end{eqnarray} 
where $Ric$ denotes the Ricci tensor of $(M,g)$ and where $\{e_{(i)}(z):=c'_{(i)}(0)\}$ denotes the ($g_t^{(\omega)}$--orthonormal) basis of $T_zM$ associated with the potentials $\{\psi_{(j)}\}$ defining  normal coordinates in $(M,\,g_t^{(\omega)})$, centered at $z\in M$.
\end{prop}
\begin{proof}
By tracing the expression (\ref{Riemop}) for the Riemann curvature operator $\widetilde{R}m^{(t)}$ and by exploiting the orthonormality conditions $g_t^{(\omega)}\left(e_{(i)},\,e_{(k)}\right)(z)\,=\,\delta_{ik}$, and 
$g\left(\nabla\widehat{\psi}^{\;(i)},\,\nabla\widehat{\psi}^{\;(k)}\right)= \delta_{ik}$, (see (\ref{gpsiorth}), which hold in normal geodesic coordinates
one immediately gets (\ref{Riccicompt}).    
\end{proof}

\subsection{Heat kernel induced Ricci flow}
 
We have set the stage for discussing the scaling dependence of the metric $(M,\,g_t^{(\omega)})$ for $t\in (0, \infty)$ and compute the associated beta function. To this end, for  $t>0$,\, $z\in M$ and $U_\perp \in T_z(M)$, let $\widehat{\psi}_{(t,U)} :=\widehat{\psi}_{(t,z,U_\perp )}$ be the corresponding potential defined by the weighted heat kernel injection map (\ref{tangmappsi}).
Let us consider for $t>0$, and for $U_\perp(z)$ varying in $T_z(M)$, the smooth functional
\begin{equation}
U_\perp \longmapsto \mathcal{P}_{(U)}(z,t):= \int_M\,\widehat{\psi}_{(t,U)}(y)\,p^{(\omega)} _t(y, z)d\omega(y)\;,
\end{equation}
providing the coordinate representation in $\left(C^{\infty }(M, \mathbb{R}) \right)^*$ of $p^{(\omega)} _t(\cdot , z)d\omega(\cdot )$ at fixed heat source $\delta_z$. We define a time--dependent tangent vector $X _{(t,z)}\in T_zM$ by 
\begin{equation}
\label{Xvectfield}
g_t^{(\omega)}\left(X _{(t,z)},\,U_\perp(z) \right)\,=\,-\,\frac{d}{d t}\,\mathcal{P}_{(z)}^{\varphi}(t)\;. 
\end{equation}
Explicitly one computes
\begin{eqnarray}
&&\frac{d}{d t}\int_M\,\widehat{\psi}_{(t,U)}(y)\;p^{(\omega)} _t(y, z)d\omega(y)\\
\nonumber\\
&&=\int_M\,\left(\widehat{\psi}_{(t,U)}(y)\,\frac{\partial}{\partial t}\ln p^{(\omega)} _t(y, z)\right)\;p^{(\omega)} _t(y, z)d\omega(y)\nonumber\\
\nonumber\\
&&=\int_M\,\left(\widehat{\psi}_{(t,U)}(y)\,\bigtriangleup _{p_t(d\omega)}\ln p^{(\omega)} _t(y, z)\right)\;p^{(\omega)} _t(y, z)d\omega(y)\nonumber\\
\nonumber\\
&&=\,-\,\int_M\,\left(\nabla \ln p^{(\omega)} _t(y, z)\cdot \nabla \widehat{\psi}_{(t,U)}(y)\right)\;p^{(\omega)} _t(y, z)d\omega(y),\;\;\;\;\;\;\forall U_\perp \in T_zM\nonumber\;,
\end{eqnarray}
where we exploited the relation $(p^{(\omega)} _t)^{-1}\,\bigtriangleup _\omega  p^{(\omega)} _t=\bigtriangleup _{p_t(d\omega)}\ln p^{(\omega)} _t $, which holds pointwise for $t>0$. Hence $X _{(t,z)}$ is provided by 
\begin{eqnarray}
\label{veetz}
&&g_t^{(\omega)}\left(X _{(t,z)},\,U_\perp(z) \right)\\
\nonumber\\
&&=\,\int_M\,\left(\nabla \ln p^{(\omega)} _t(y, z)\cdot \nabla \widehat{\psi}_{(t,U)}(y)\right)\,p^{(\omega)} _t(y, z)d\omega(y)\;.\nonumber
\end{eqnarray} 
We have the following
\begin{lem}
\label{Lagr}
The Lie derivative of the metric tensor $g_t^{(\omega)}$ along the $t$--dependent vector 
field $z\longmapsto X _{(t,z)}$ is provided by 
\begin{eqnarray}
 \label{hessianlieder}
 &&{\mathcal{L}}_{X _{(t,z)}}\,g_t^{(\omega)}(U_\perp , U_\perp )\\
\nonumber\\
&&=\,
2\,\int_M\,\left(\nabla^a\psi_{(t,U)}\,Hess_{ab}\,\ln p^{(\omega)} _t(y , z)\,\nabla ^b\psi_{(t,U)}\right)\,p^{(\omega)} _t(y , z)\,d\omega(y)\;.\nonumber
\end{eqnarray}
\end{lem}
\begin{proof}  
Let $\vee _{\psi_{(t)}}$ denote the vector $\in T_{p_t(d\omega)}\rm{Prob}_{ac}(M,g)$ associated with $X _{(t,z)}\in T_zM$. According to the characterization of the Lie derivative $\overline{\mathcal{L}}$ on $\rm{Prob}(M,\,g)$, (cf. (\ref{Lieder})), we have 
\begin{eqnarray}
 \label{}
 &&\overline{\mathcal{L}}_{\vee _{(t,z)}}\,g_t^{(\omega)}(U_\perp , U_\perp )\\
\nonumber\\
&&:=\, \langle \overline{\bigtriangledown }_{\vee _{\psi_{(t)}}}\,\vee _{(t,z)},\,\vee _{\psi_{(t)}} 
\rangle_{p_t(d\omega )}\,+\,\langle \,\vee _{\psi_{(t)}},\,\overline{\bigtriangledown }_{\vee _{\psi_{(t)}}}\vee _{(t,z)} \rangle_{p_t(d\omega )}\nonumber\\
\nonumber\\
&&=\,
2\,\int_M\,\left(\nabla^a\psi_{(t,U)}\,Hess_{ab}\,\ln p^{(\omega)} _t(y , z)\,\nabla ^b\psi_{(t,U)}\right)\,p^{(\omega)} _t(y , z)\,d\omega(y)\;,
\nonumber
 \end{eqnarray}
 where we have exploited the definition (\ref{conn}) of the Riemannian connection  on $\rm{Prob}(M,\,g)$   
and $\vee _{\varphi }\,=\,\nabla ^{(y)}\ln p^{(\omega)} _t(y , z)$, (see  (\ref{veetz})), to write, $\forall \varphi\in C^{\infty }(M,\mathbb{R})$,  
\begin{eqnarray}
\label{almhess}
&&\left\langle \overline{\bigtriangledown }_{\vee _{\varphi }}\,\vee _{(t,z)},\,\vee _{\varphi } \right\rangle _{(g,p_t(d\omega))}\\
\nonumber\\
&&=\,\int_M\,\left(\nabla \varphi (y)\cdot \nabla \nabla \ln p^{(\omega)} _t(y, z)\cdot \nabla \varphi (y)\right)\,p^{(\omega)} _t(y, z)d\omega(y)\;.\nonumber
\end{eqnarray} 
Since,  according to Lemma \ref{exchconn}, the Levi--Civita connection $\widetilde{\nabla }$ of  $(M, g_t^{(\omega)})$ can be identified with the induced $\overline{\nabla }$ connection on  $\Upsilon _t\left((M,\,g) \right)$, we have that the Lie derivative $\overline{\mathcal{L}}_{\vee _{(t,z)}}$ on ${\rm Prob}_{ac}(M,g)$ 
when restricted to $\Upsilon _t\left((M,\,g) \right)$ reduces to  ${\mathcal{L}}_{X _{(t,z)}}$ and the Lemma follows.
\end{proof}
\begin{rem}
Note that the Lie derivative  (\ref{hessianlieder}) captures heat concentration phenomena related to the heuristics of \emph{heat propagation along geodesics} expressed by the asymptotics (\ref{parkeras}),\,(cf. Th. \ref{parkth}). In line with (\ref{asympt5}) and  
Varadhan's large deviation formula (\ref{largedevweigh}), for $y$ in compact subset $M\setminus(Cut(z)\cup \{z\})$, we have the uniform limit \cite{malliavin} 
\begin{equation}
\label{}
-\,4\,\lim_{t\searrow 0^+}\,{t}\,Hess\,\left(\ln\,\,p^{(\omega)}_{t}(y, z)\right)\,=\,
Hess\;{d_g^2(y, z)}\;,
\end{equation}
whereas for $y\in Cut(z)$, \, the hessian $Hess\,\left(\ln\,\,p^{(\omega)}_{t}(y, z)\right)$ diverges to $-\,\infty $ faster than $t^{-1}$ as $t\searrow 0$. This implies that  $-t\ln\,\,p^{(\omega)}_{t}(y, z)$ plays the role of a smooth mollifier \cite{malliavin} of the (squared) distance function $d_g^2(y, z)$, and the limit of $-t\,Hess\,(\ln\,p^{(\omega)} _t(y , z))$ for $t\searrow 0$ computes $Hess\;{d_g^2(y, z)}$ in the sense of distributions \cite{malliavin}, \cite{neel}, \cite{neel2}. In particular it can be shown \cite{nirenberg}, \cite{mantegazzadis} that  the singular part of $Hess\;{d_g^2(y, z)}$ is concentrated, for $dim\,M\geq2$,  on $Cut(z)$ and it is absolutely continuous with respect to the $(n-1)$--dimensional Hausdorff measure $\mathcal{H}^{n-1}$ on $Cut(z)$. 
\end{rem}

\noindent The explicit  computation of the beta function  $\frac{d}{d t} g_t^{(\omega)}(U_\perp , U_\perp )$  can be carried out along the lines of \cite{Carlo}, with some obvious modifications due to the presence of the weighted Laplacians $\bigtriangleup _\omega$ and $\bigtriangleup _{p_t(d\omega)}$. We start with some preliminary remarks. Let us recall that the standard Bochner-Weitzenb\"ock formula for the Laplace--Beltrami operator $\bigtriangleup_g$ on $(M,g)$  reads
\begin{equation}
\label{BW}
2\nabla^i\varphi \nabla_i\bigtriangleup _g  \varphi \,=\,\bigtriangleup _g \left|\nabla \varphi  \right|^2_g-
2\left|Hess\,\varphi  \right|_g^2 -2\,R^{ik}\,\nabla_i\varphi\nabla_k\varphi\;,
\end{equation}
for any $\varphi\in C^{\infty }(M, \mathbb{R})$.  
From (\ref{BW}) and
\begin{equation}
2\,\nabla^i\varphi\, \nabla_i\bigtriangleup _\omega\,\varphi\,=\,2\,\nabla^i\varphi\, \nabla_i
(\bigtriangleup _g \varphi-\nabla_kf\,\nabla^k\varphi)\;,
\end{equation}
\begin{equation}
2\,\nabla^i\varphi\, \nabla_i\bigtriangleup _{p _t\,(d\omega)} \,\varphi\,=\,2\,\nabla^i\varphi\, \nabla_i
\left(\bigtriangleup _g \varphi-\nabla_k \left(f-\ln\,p^{(\omega)} _t\right)\,\nabla^k\varphi\right)\;,
\end{equation}
a direct computation extends  (\ref{BW}) to $(M,g,\,d\omega)$ and to $(M,g,\,p^{(\omega)} _t\,d\omega)$ according to 
\begin{eqnarray}
\label{wWB}
&&2\nabla^i\varphi \nabla_i\bigtriangleup _\omega \varphi=\bigtriangleup _\omega \left|\nabla \varphi  \right|^2_g -
2\left|Hess\,\varphi  \right|_g^2\\
&&\, -2 \left(R_{ik}+ Hess_{ik} f\right) \nabla^i\varphi \nabla^k\varphi\;,\nonumber
\end{eqnarray} 
and
\begin{eqnarray}
\label{wWBheat}
&&2\nabla^i\varphi \nabla_i\bigtriangleup _{p _t\,(d\omega)} \varphi\,=\,\bigtriangleup _{p _t\,(d\omega)} \left|\nabla \varphi  \right|^2_g -
2\left|Hess\,\varphi  \right|_g^2\\
 &&\,-\,2 \left(R_{ik}+ Hess_{ik}\left(f-\ln\,p^{(\omega)} _t(y , z)\right)\right) \nabla^i\varphi \nabla^k\varphi\nonumber\;,
\end{eqnarray} 
respectively.  Let us observe that
 from  (\ref{carlo12Hnew}) it easily follows that (\ref{zlimit}) can be rewritten as
\begin{equation}
\label{zlimit2}
\left(\frac{\partial}{\partial t}\,-\,\bigtriangleup _{\omega}^{(y)} \right)\,\left[ 
p^{(\omega)} _t(y , z)\,\bigtriangleup _{p _t\,(d\omega)}^{(y)}\,\widehat{\psi}_{(t,z,U_\perp )}(y)\right]\,=0\;,
\end{equation}
or equivalently as 
\begin{equation}
\label{zlimit3}
\left(\frac{\partial}{\partial t}-\bigtriangleup _{p _t\,(d\omega)}^{(y)}-\nabla^{(y)}\ln\,p^{(\omega)} _t\cdot\nabla^{(y)}  \right) 
\bigtriangleup _{p _t\,(d\omega)}^{(y)}\,\widehat{\psi}_{(t,z,U_\perp )}\,=\,0\,.
\end{equation}
Also note that by differentiating both members of (\ref{carlo12H}) with respect to $t$ we to get, for any $t>0$,
(cf. \cite{Carlo} Prop. 4.1 for a similar computation for the standard heat kernel),
\\
\begin{eqnarray}
&&{div}_{\omega}^{(y)}\,\left[\left(\tfrac{\partial }{\partial t}p^{(\omega)} _t(y , z)\right)\,\nabla^{(y)}\,\widehat{\psi}_{(t,z,U_\perp )}(y)+p^{(\omega)} _t(y , z)\,\nabla^{(y)}\,\tfrac{\partial }{\partial t}\widehat{\psi}_{(t,z,U_\perp )}(y)\right]\,\nonumber \\
\nonumber\\
&&=\,-\,U_\perp(z)\cdot \,\nabla^{(z)}\,\tfrac{\partial }{\partial t}p^{(\omega)} _t(y , z)\;,
\end{eqnarray}
\\
\noindent
Since  $\nabla^{(z)}\,\tfrac{\partial }{\partial t}p^{(\omega)} _t(y , z)=
\nabla^{(z)}\,\bigtriangleup _\omega^{(y)}\,p^{(\omega)} _t(y , z)$, and the weighted laplacian $\bigtriangleup _\omega^{(y)}$ acts with respect to the variable point $y$, we can write, after an obvious rearrangement of terms,
\begin{eqnarray}
&&{div}_{\omega}^{(y)}\,\left[p^{(\omega)} _t(y , z)\,\nabla^{(y)}\,\tfrac{\partial }{\partial t}\widehat{\psi}_{(t,z,U_\perp )}(y)\right]\\
\nonumber\\
&&=\,-\,\bigtriangleup _\omega^{(y)}\,\left[U_\perp(z)\cdot \,\nabla^{(z)}p^{(\omega)} _t(y , z)  \right]\,\nonumber\\
\nonumber\\
&&-\,{div}_{\omega}^{(y)}\,
\left[\nabla^{(y)}\,\widehat{\psi}_{(t,z,U_\perp )}(y)\;\bigtriangleup _\omega^{(y)}\,p^{(\omega)} _t(y , z)\right]\;.\nonumber
\end{eqnarray}
\\
\noindent
By inserting the defining pde  (\ref{carlo12H}) in $\bigtriangleup _\omega^{(y)}\,\left[U_\perp(z)\cdot \,\nabla^{(z)}p^{(\omega)} _t(y , z)  \right]$, we eventually get
\begin{eqnarray}
\label{firstrel}
&&{div}_{\omega}\,\left[p^{(\omega)} _t(y , z)\,\nabla\,\tfrac{\partial }{\partial t}\widehat{\psi}_{(t,z,U_\perp )}(y)\right]\\
\nonumber\\
&&=\,\bigtriangleup _\omega\,\left[ {div}_{\omega}\,\left(p^{(\omega)} _t(y , z)\,\nabla\,\widehat{\psi}_{(t,z,U_\perp )}(y)\right) \right]\,\nonumber\\
\nonumber\\
&&-\,{div}_{\omega}\,
\left[\nabla\,\widehat{\psi}_{(t,z,U_\perp )}(y)\;\bigtriangleup _\omega\,p^{(\omega)} _t(y , z)\right]\;,\nonumber
\end{eqnarray}
where we dropped the superscript $(y)$ since the operators $\nabla$,\, ${div}_{\omega}$,\, and $\bigtriangleup _\omega$ all act with respect to $y$.\\
\\
\noindent
With these preliminary remarks along the way we get   
\begin{lem} (cf. Prop. 4.1 in \cite{Carlo}).
\label{MGlemma}
For $t\in (0,\infty)$ we have
\begin{eqnarray}
\label{genbetafunct}
&&\frac{d }{d t}\,g_t^{(\omega)}(U_\perp , U_\perp )\,=\,-\,2\,\int_{M}\,\left[\left|Hess\,\widehat{\psi}_{(t,U)}  \right|_g^2\right.\nonumber\\
\\
&&\left.\,+\,\left(R_{ik}\,+\,Hess_{ik}\,f\right)\,\nabla^i\widehat{\psi}_{(t,U)}\,\nabla^k\widehat{\psi}_{(t,U)}   \right]\,p^{(\omega)} _t(y , z)\,d\omega(y)\nonumber\;.
\end{eqnarray}
\end{lem}
\begin{proof}
According to (\ref{heatfloweight}) we have
\begin{eqnarray}
\label{dergt}
&&\frac{d}{d t}\,g_t^{(\omega)}(U_\perp , U_\perp )\\
\nonumber\\
&&=\, 
2\,\int_{M}\,\left(\nabla\,\widehat{\psi}_{(t,U)}\cdot \nabla\,\tfrac{\partial}{\partial t}\,\widehat{\psi}_{(t,U)}\right)\,p^{(\omega)} _t(y , z)\,d\omega(y)\nonumber\\
\nonumber\\
&&\,+\,\int_{M}\,\left|\nabla\,\widehat{\psi}_{(t,U)}\right|_g^2\,
\bigtriangleup _\omega\,p^{(\omega)} _t(y , z)\,d\omega(y)\nonumber\;.
\end{eqnarray} 
The Green formula for the weighted Laplacian $\bigtriangleup_\omega$, (which we shall use repeatedly also for the heat weighted Laplacian $\bigtriangleup_{p_t(d\omega)}$), 
\begin{eqnarray}
\label{green}
&&\int_{M}\,\varphi \bigtriangleup_\omega \psi\,d\omega\,=\, \int_{M}\,\varphi\, div_\omega\cdot \nabla  \psi\,d\omega\\
\nonumber\\
&&\,=\,-\,\int_{M}\,\nabla\varphi \cdot \nabla \psi\,d\omega\,=\,
\int_{M}\,\psi \bigtriangleup_\omega \varphi\,d\omega\;,\nonumber
\end{eqnarray} 
which holds pointwise for any $\varphi,\;\psi \in C^{\infty }(M,\mathbb{R})$, allows to compute  
the first term in (\ref{dergt}) according to
\begin{eqnarray}
\label{floc}
&&\int_{M}\,\left(\nabla\,\widehat{\psi}_{(t,U)}\cdot \nabla\,\tfrac{\partial}{\partial t}\,\widehat{\psi}_{(t,U)}\right)\,p^{(\omega)} _t(y , z)\,d\omega(y)\\
\nonumber\\
&&=\,-\,
\int_{M}\,\widehat{\psi}_{(t,U)}\,div_\omega\,\left[p^{(\omega)} _t(y , z)\, \nabla\,\tfrac{\partial}{\partial t}\,\widehat{\psi}_{(t,U)} \right]\,d\omega(y)\nonumber\\
\nonumber\\
&&\,\underset{(\ref{firstrel})}{=}\,-\,
\int_{M}\,\widehat{\psi}_{(t,U)}\,\bigtriangleup_\omega\,\left[div_\omega\left(p^{(\omega)} _t(y , z)\, \nabla\,\widehat{\psi}_{(t,U)}\right) \right]\,d\omega(y)\nonumber\\
\nonumber\\
&&+\,
\int_{M}\,\widehat{\psi}_{(t,U)}\,div_\omega\,\left[\bigtriangleup _\omega\, p^{(\omega)} _t(y , z)\; \nabla\,\widehat{\psi}_{(t,U)} \right]\,d\omega(y)\nonumber\\
\nonumber\\
&&=\,\int_{M}\,\left[\nabla\,\widehat{\psi}_{(t,U)}\cdot \nabla\,\bigtriangleup _\omega\,\widehat{\psi}_{(t,U)}\,-\,
\bigtriangleup _\omega\,\left|\nabla\,\widehat{\psi}_{(t,U)} \right|^2_g\right]\,p^{(\omega)} _t(y , z)\,d\omega(y)\nonumber\;,
\end{eqnarray}
where ${}^{\underset{(\ref{firstrel})}{=}}$ refer to the relation used in the computation, and where in the last line we have integrated by parts. Inserted  into (\ref{dergt}), the expression (\ref{floc}) eventually yields   
\begin{eqnarray}
\label{genbetafunct0}
&&\;\;\;\;\;\;\;\;\frac{d }{d t}\,g_t^{(\omega)}(U_\perp , U_\perp )\\ 
\nonumber\\
&&=\,2\,\int_{M}\,\left[\nabla\,\widehat{\psi}_{(t,U)}\cdot \nabla\,\bigtriangleup _\omega\,\widehat{\psi}_{(t,U)}\,-\,
\tfrac{1}{2}\,\bigtriangleup _\omega\,\left|\nabla\,\widehat{\psi}_{(t,U)} \right|^2_g\right]\,p^{(\omega)} _t(y , z)\,d\omega(y)\nonumber\\
\nonumber\\
&&=\,-\,2\,\int_{M}\,\left[\left|Hess\,\widehat{\psi}_{(t,U)}  \right|_g^2\right.\nonumber\\
\nonumber\\
&&\left.\,+\,\left(R_{ik}\,+\,Hess_{ik}\,f\right)\,\nabla^i\widehat{\psi}_{(t,U)}\,\nabla^k\widehat{\psi}_{(t,U)}   \right]\,p^{(\omega)} _t(y , z)\,d\omega(y)\nonumber\;,
\end{eqnarray}
where we have exploited the weighted Bochner--Weitzenb\"ock formula (\ref{wWB}).
\end{proof}

\noindent From Lemma  \ref{MGlemma}  and the expression (\ref{Riccicompt}) of the the Ricci tensor of $(M, g_t^{(\omega)})$  we directly get

\begin{thm} (The heat kernel induced Ricci flow).
\label{rftheom}
Along the   weighted heat kernel embedding $(0,\infty)\times M\ni (t,z)\longmapsto p_t^{(\omega)}(\cdot , z)d\omega(\cdot )\in \rm{Prob}_{ac}(M,g)$, the beta function $\beta(g_t^{(\omega)}))$ associated to the scale dependent metric $(M,\,g_t^{(\omega)})$ is provided by 
\\
\begin{eqnarray}
\label{GenRicPer}
&&\frac{d }{d t}\,g_t^{(\omega)}(U_\perp, W_\perp  )=\,-\,2\widetilde{R}ic^{(t)}(U_\perp, W_\perp  )\\
\nonumber\\
&&-\,2\,\int_M\,\left(\nabla \psi_{(t,U)}\cdot \nabla \nabla f\cdot \nabla \psi_{(t,W)}\right)\,p^{(\omega)} _t(y, z)d\omega(y)\nonumber\\
\nonumber\\
&&\,-\,2\,\int_M\,\left(Hess\,\widehat{\psi}_{(t,U)}\cdot\,Hess\,\widehat{\psi}_{(t,W)}\right)\,
 p^{(\omega)} _t(y, z)\,d\omega(y)\nonumber\;,    
\end{eqnarray}
where $\widetilde{R}ic^{(t)}$ denotes the Ricci curvature of the evolving metric $(M,\,g_t^{(\omega)})$.
\end{thm}
\begin{proof}
According to Lemma \ref{MGlemma}, we have  (cf. (\ref{genbetafunct})),
\begin{eqnarray}
\label{dergbet}
&&\frac{d }{d t}\,g_t^{(\omega)}(U_\perp , W_\perp)(z)\,=\,-\,2\,\int_{M}\,\left[Hess\,\widehat{\psi}_{(t,U)}\cdot   Hess\,\widehat{\psi}_{(t,W)}\right.\nonumber\\
\\
&&\left.\,+\,R_{ab}\,\nabla^a\widehat{\psi}_{(t,U)}\,\nabla^b\widehat{\psi}_{(t,W)}   \right]\,p^{(\omega)} _t(y , z)\,d\omega(y)\nonumber\\
\nonumber\\
&&-2\,\int_{M}\,\nabla^a\widehat{\psi}_{(t,U)}\,Hess_{ab}\,f\,\nabla^b\widehat{\psi}_{(t,W)}\,p^{(\omega)} _t(y , z)\,d\omega(y)\nonumber\;.
\end{eqnarray}
For $t>0$  we can apply Proposition \ref{smRiccT} which computes the Ricci curvature of $(M,\,g_t^{(\omega)})$ at $z$ in  the direction of the $2$--plane spanned by the vectors $(U_\perp , W_\perp)\in T_zM$.
\end{proof}
The above result indicates a striking connection between the beta function for the scale dependent flow  $[t, (M,g)]\mapsto (M,\,g_t^{(\omega)})$, $t>0$, and (a generalized version of) the DeTurck--Hamilton version of the Ricci flow. This is further supported by the behavior of (\ref{GenRicPer})  in the singular limits  $t\searrow 0$ and $q\nearrow \infty $. The former  controls how the curve of heat kernel embeddings $(0,\infty)\times M\ni (t,z)\longmapsto p_t^{(\omega)}(\cdot , z)d\omega(\cdot )\in \rm{Prob}_{ac}(M,g)$ approaches the isometric embedding of $(M,g)$ in the 
non--smooth $\rm{Prob}(M)\supset \rm{Prob}_{ac}(M,g)$. The latter  is related to our choice (\ref{kprobmeas}) of the dilatonic measure $d\omega(q)$ localizing the NL$\sigma$M  maps $\phi:\Sigma\longrightarrow (M,g,\,d\omega)$  around the center of mass of the constant maps $\{\phi_{(k)}\}_{k=1}^q$, (see Section \ref{lacmadw}). The two limits clearly interact since the heat kernel $p_t^{(\omega)}$ explicitly depends on the choice of the measure $d\omega(q)$. Since for any given $d\omega(q)$ we have a well--defined flow (\ref{GenRicPer}), it is geometrically natural to consider the large $q$ limit as defined by 
\begin{equation}
\lim_{q\nearrow \infty }\left\{\lim_{t\searrow 0}\,\frac{d}{d t}\,g^{(\omega(q))}_t \right\}_{q\in\mathbb{N}}\;.
\end{equation}
With these remarks along the way we have
\begin{prop}
\label{tangbetaflow}
Let $[0,1]\ni s\mapsto \gamma_s$, $\gamma(0)\equiv z$ denote a geodesic in $(M,g)$. Then the beta function (\ref{GenRicPer}) associated with the weighted heat kernel 
embedding $(M,g,\,d\omega)\,\longrightarrow \, \left(\rm{Prob}_{ac}(M,g),\,d_g^W\right)$  
is tangent for $t=0$ and $q\nearrow \infty $ to the perturbative beta functions for the dilatonic non--linear $\sigma$ model
\begin{equation}
\label{carlobasic2}
\left.\frac{d}{dt}\,\left.\,g^{(\omega)}_t(\dot{\gamma}_s,\dot{\gamma}_s)\right|_{t=0}\right|_{q\rightarrow \infty }\,=\,-2\,
\left[Ric_g(\dot{\gamma}_s,\dot{\gamma}_s)\,+\,Hess\,f(\dot{\gamma}_s,\dot{\gamma}_s)\right]\;,
\end{equation}
\\
\begin{equation}
\label{basic5}
\left.\frac{d}{dt}\,\left.f_t^{(\omega)}\right|_{t=0}\right|_{q\rightarrow \infty }\,=\,\,\bigtriangleup _g\,f\,
-\,|\nabla\,f|^2_g\;,
\end{equation}
where the equality holds for almost every $s\in [0,1]$.
\end{prop}
\begin{proof}
By directly adapting a basic result (Th. 4.6)  of \cite{Carlo}, one can control the $t\searrow 0$ limit in (\ref{genbetafunct}) and get
\begin{equation}
\label{carlobasic1}
\frac{d}{dt}\,\left.\int_0^1\,g^{(\omega)}_t(\dot{\gamma}_s,\dot{\gamma}_s)\,ds\right|_{t=0}\,=\,-2\,
\int_0^1\,\left[Ric_g(\dot{\gamma}_s,\dot{\gamma}_s)\,+\,Hess\,f(\dot{\gamma}_s,\dot{\gamma}_s)\right]\,ds\;,
\end{equation}
and almost everywhere for $s\in [0,1]$ we have 
\begin{equation}
\label{carlobasic20}
\frac{d}{dt}\,\left.\,g^{(\omega)}_t(\dot{\gamma}_s,\dot{\gamma}_s)\right|_{t=0}\,=\,-2\,
\left[Ric_g(\dot{\gamma}_s,\dot{\gamma}_s)\,+\,Hess\,f(\dot{\gamma}_s,\dot{\gamma}_s)\right]\;.
\end{equation}
If we couple this expression with  the flow (\ref{parabolicf}) for the function $f^{\,(\omega)}_t$, evaluated for $t=0$, 
\begin{eqnarray}
\label{parabolicf00}
\left.\frac{\partial}{\partial t}\,f^{\,(\omega)}_t\right|_{t=0}\,=\,\bigtriangleup _g\,f\,
-\,\tfrac{2+q}{q}\,|\nabla\,f|^2_g\;,
\end{eqnarray}
we immediately get the stated result in the $q\nearrow \infty $ limit. It is worthwhile stressing that the actual proof of (\ref{carlobasic20}) in \cite{Carlo}, on which the above result heavily relies, is rather technical. By contrast, the underlying rationale is simple. To wit, since the (Kantorovich) potential  $\widehat{\psi}_{(t,\dot{\gamma}_s)}$ in (\ref{genbetafunct0}), (for $U_\perp \equiv \dot{\gamma}_s$),  is generated by the tangent to the geodesic $s\longmapsto \gamma_s$, one expects that  $\left|Hess\,\widehat{\psi}_{(t,\dot{\gamma}_s)}  \right|_{t=0}$ vanishes at $\delta_z$ and that it should remain small for $0<\,t\,<\,\varepsilon $, with $\varepsilon $ small enough.  As reasonable as it appears, the vanishing of $Hess\,\widehat{\psi}_{(t,\dot{\gamma}_s)}$ for $t\searrow 0$ is deceptively difficult to prove, and commands a technical tour de force, (cf. Lemma 4.4 and Prop. 4.5 in \cite{Carlo}).    
\end{proof}
It is also important to observe that we can connect the singular limit $t\searrow 0$, $q\nearrow \infty $ of the beta function (\ref{GenRicPer})  to the Hamilton--Perelman Ricci flow. To this end we couple the evolutions of the rescaled dilaton  $f^{(\omega)}_t$ and of the metric $g^{(\omega)}_t$ via the  probability measure $d\omega$. In particular if we require that $\frac{d}{dt}\,d\mu_{h_t^{(q)}}|_{t=0}\,=\,0$, 
(recall that $d\mu_{h_t^{(q)}}|_{t=0}=d\omega$), we have
\begin{prop}
\label{HamPerTangency}
Under the constraint 
\begin{equation}
\left.\frac{d}{dt}\,d\mu_{h_t^{(q)}}\right|_{t=0}\,=\,0\;
\end{equation} 
the beta function (\ref{GenRicPer})
 is tangent, for $t=0$ and $q\nearrow \infty $, to the generators of the  Hamilton--Perelman Ricci flow according to 
\begin{equation}
\label{carlobasic2}
\left.\frac{d}{dt}\,\left.\,g^{(\omega)}_t(\dot{\gamma}_s,\dot{\gamma}_s)\right|_{t=0}\right|_{q\rightarrow \infty }\,=\,-2\,
\left[Ric_g(\dot{\gamma}_s,\dot{\gamma}_s)\,+\,Hess\,f(\dot{\gamma}_s,\dot{\gamma}_s)\right]\;,
\end{equation}
\\
\begin{equation}
\label{basic5}
\left.\frac{d}{dt}\,\left.f_t^{(\omega)}\right|_{t=0}\right|_{q\rightarrow \infty }\,=\,-\,\triangle _g\,f\,-\,R^{(g)}\;,
\end{equation}
where the equality holds for almost every $s\in [0,1]$.
\end{prop}
\begin{proof}
To prove this result, it is useful to recast (\ref{carlobasic2}) in terms of the metric $h_t$ defined by (\ref{newmetr0}). Denote by $\{y^i\}_{1=1}^n$ local (geodesic) coordinates on $(M,g_t)$, and $\{\zeta ^\alpha\}_{1=1}^q$ local coordinates on $(\mathbb{T}^q, \delta)$. Then, according to a standard formula for warped product metrics  (cf. \emph{e.g.} p.46 of \cite{Glickenstein}), we have
\begin{equation}
R^{(h)}_{ij}\,=\,R^{(g)}_{ij}\,+\,\nabla _i\nabla_j\,f\,
-\,\frac{1}{q}\,\nabla _i\,f\,\nabla_j\,f\;,
\end{equation}
\begin{equation}
R^{(h)}_{\alpha\beta}\,=\,\frac{1}{q}\,e^{-\,\tfrac{2f(y)}{q}}\,\delta _{\alpha\beta}\,\left(
\triangle _g\,f\,
-\,|\nabla\,f|_g^2\right)\;,
\end{equation}
\begin{equation}
R^{(h)}_{i\alpha}\,=\,0\,,\;\;\;\;\;\;i=1,\ldots,n,\;\;\;\alpha=1,\ldots,q\;,
\end{equation}
where $R^{(g)}_{ij}$ denotes the components of the Ricci tensor of $(M,g)$, (and similarly the subscript $g$ refers to gradient, hessian, and Laplacian relative to the metric $g$). From
\begin{equation}
\label{dtfactorize}
\frac{d}{dt}\,h_t\,=\,\frac{d}{dt}\,g_t^{(\omega)}\,-
\,\frac{2}{q}\,e^{-\,\tfrac{2f_t^{(\omega)}}{q}}\,\delta_{\mathbb{T}^q}\,\frac{d}{dt}\,f_t^{(\omega)}\;,
\end{equation}
and (\ref{parabolicf00}), we get 
\begin{equation}
\label{carlobasic3}
\left.\frac{d}{dt}\,h_t\right|_{t=0}\,=\,-2\,Ric^{(h)}\,-\frac{2}{q}\,\nabla f\otimes \nabla f\;.
\end{equation} 
If we impose, at $t=0$, the preservation along (\ref{carlobasic3}) of the probability measure density $d\mu_{h_t}$, (note that $d\mu_{h_t}|_{t=0}=d\omega$), then we get  
\begin{equation}
\label{preservolume}
\left.\tfrac{d}{dt}\,d\mu_{h_t}\right|_{t=0}\,=\,
\left.\tfrac{d}{dt}\,e^{-\,\tfrac{2f_t^{(\omega)}}{q}}\,d\mu_{g_t}\right|_{t=0}\,=\,0\;.
\end{equation}
From (\ref{carlobasic3})  we  compute
\begin{equation}
\left.\tfrac{d}{dt}\,d\mu_{h_t}\right|_{t=0}\,=\,-\,\left[R^{(h)}+\tfrac{1}{q}\,|\nabla\,f|^2_g\right]\,d\mu_h
\,=\,0\;,
\end{equation}
 where
\begin{equation}
R^{(h)}\,=\,R^{(g)}\,+\,2\triangle_g\,f\,-\,\frac{q+1}{q}\,|\nabla\,f|_g^2\;,
\end{equation}
denotes the scalar curvature of the metric $h$ expressed in terms of the scalar curvature $R^{(g)}$ of $(M,g)$. Thus, 
 \begin{equation}
 \left.\tfrac{d}{dt}d\mu_{h_t}\right|_{t=0}\,=\,0\,\Rightarrow \,|\nabla\,f|_g^2\,=\,
 R^{(g)}\,+\,2\triangle_g\,f\;,
 \end{equation}
 which, introduced in (\ref{parabolicf00}), provides the backward heat equation
 \begin{equation}
\label{basic50}
\frac{d}{dt}\,\left.f_t^{(\omega)}\right|_{t=0}\,=\,-\,\triangle _g\,f\,-\,R^{(g)}\,-\,\frac{2}{q}\,\,|\nabla\,f|_g^2\;.
\end{equation}
It follows that the flow $t\longmapsto h_t$, defined by (\ref{genbetafunct}) by enforcing the volume density preservation 
(\ref{preservolume}),   has a tangent which, at $t=0$, and for $q\nearrow \infty $, reduces to the tangent vector defining the  Hamilton--Perelman flow.
\end{proof}
The strong similarity between (\ref{GenRicPer}) and the (DeTurck version of the) Ricci flow, and the tangency conditions described in Propositions \ref{tangbetaflow} and \ref{HamPerTangency}, may suggest that (\ref{GenRicPer}) is indeed the Ricci flow in disguise. In the case of the standard heat kernel embedding  \cite{Carlo}, the induced flow on the distance function, tangential to the Ricci flow for $t=0$, is well defined for any $t\geq 0$, and with strong control on the topology of $M$ and good continuity properties with respect to (measured) Gromov--Hausdorff convergence. As argued in \cite{Carlo}, these properties strongly contrast with the typical behavior of the Ricci flow, characterized by the development of curvature singularities and by a poor control on Gromov--Hausdorff limits of sequences of Ricci evolved manifolds. The  explicit expression (\ref{GenRicPer}) for the heat kernel induced flow $(t,g)\longmapsto  g_t^{(\omega)}$, and in particular the presence of the norm--contracting term $\int_M\,|Hess\,\widehat{\psi}_{(t,U)}|^2_g\,p^{(\omega)} _t(y, z)\,d\omega(y)$, (cf. (\ref{GenRicPer}) for $U_\perp =W_\perp $), indicates that along the flow $(t,g)\longmapsto  g_t^{(\omega)}$ there is a strong control of the metric geometry of $g_t^{(\omega)}$.  To provide evidence in this direction without belaboring on the subtle aspects of measured Gromov--Hausdorff convergence we describe the monotonicity properties of (\ref{GenRicPer}). Not surprisingly, they turn to be somehow stronger than those in Ricci flow theory.

\subsection{Monotonicity and gradient flow properties}

 From Lemma \ref{Lagr} and \ref{MGlemma}  we get  
\begin{prop} 
\label{decrprop}
Let
\begin{eqnarray}
&&(0,\infty)\,\ni \,t\longmapsto \varphi (t)\in \mathcal{D}iff(M)\nonumber\\
\label{fidiff}\\
&&\frac{\partial }{\partial t}\varphi (t)\,=\,X_{(t,z)},\;\;\;\;\;\lim_{t\searrow 0}\varphi (t)\,=\,Id_M\nonumber
\end{eqnarray}
be the curve of diffeomorphisms generated by the vector field $X_{(t,z)}$ defined by (\ref{Xvectfield}). Then, for any $U_\perp,\, V_\perp\in T_zM$,
\begin{eqnarray}
\label{LieMonot0}
&&\left(\varphi^{-1}(t)\right)^*\frac{d}{d t}\,\left(\varphi^* g_t^{(\omega)}\right)(U_\perp , V_\perp )\\
\nonumber\\
&&=\,-2\,\int_M\,\widehat{\psi}_{(t,U)} \,\bigtriangleup^2 _{p _t\,(d\omega)}\,\widehat{\psi}_{(t,V)}\,
p^{(\omega)} _t(y , z)\,d\omega(y)\;,\nonumber
\end{eqnarray}
In particular, the pull-back $\left(\varphi^* g_t^{(\omega)}\right)(U_\perp , U_\perp )$ is a monotonically decreasing function of $t$,
\begin{eqnarray}
\label{LieMonot}
&&\left(\varphi^{-1}(t)\right)^*\frac{d}{d t}\,\left(\varphi^* g_t^{(\omega)}\right)(U_\perp , U_\perp )\\
\nonumber\\
&&=\,-2\,\int_M\,\left(\bigtriangleup _{p _t\,(d\omega)}\,\widehat{\psi}_{(t,U)}  \right)^2\,
p^{(\omega)} _t(y , z)\,d\omega(y)\,< \,0\;,\nonumber
\end{eqnarray}
and moreover we have
\begin{eqnarray}
\label{monotpsi}
&&\frac{d}{d t}\,\left(\left(\varphi^{-1}(t)\right)^*\frac{d}{d t}\,\left(\varphi^* g_t^{(\omega)}\right)  \right)(U_\perp , U_\perp )\nonumber\\
\\
&&\,=\,4\,\int_M \left|\nabla^{(y)}\left(\bigtriangleup _{p _t\,(d\omega)}^{(y)}\,\widehat{\psi}_{(t,z,U_\perp )}(y)  \right)\right|^2\,p^{(\omega)} _t(y , z)\, d\omega(y)\;.\nonumber
\end{eqnarray} 
\end{prop}

\begin{proof}
We can factor out the $\varphi (t)$ in (\ref{LieMonot}) by exploiting the familiar DeTurck argument \cite{DeTurck}, (cf. \cite{6} for details). Explicitly, for $0\,<\,t\,<\,\epsilon $, one computes 
\begin{eqnarray}
&&\frac{d}{d t}\,\left(\varphi^*(t)\, g_t^{(\omega)}\right)\,=\,\left.
\frac{d}{d s}\right|_{s=0}\,\left(\varphi^*(t+s)\, g_{t+s}^{(\omega)}\right)\\
\nonumber\\
&&=\,\varphi^*(t)\,\left(\frac{d}{d t}\, g_t^{(\omega)}\right)\,+ 
\,\left.
\frac{d}{d s}\right|_{s=0}\,\left(\varphi^*(t+s)\, g_{t}^{(\omega)}\right)  \nonumber\\
\nonumber\\
&&=\,\varphi^*(t)\,\left(\frac{d}{d t}\, g_t^{(\omega)}\right)\,+ \,
\varphi^*(t)\,\left(\mathcal{L}_{X_{(t,z)}}\, g_t^{(\omega)}\right)\;.\nonumber
\end{eqnarray}
Hence, we get
\begin{equation}
\left(\varphi^{-1}(t)\right)^*\frac{d}{d t}\,\left(\varphi^*(t)\, g_t^{(\omega)}\right)\,=\,
\frac{d}{d t}\, g_t^{(\omega)}\,+\,\mathcal{L}_{X_{(t,z)}}\, g_t^{(\omega)}\;,
\end{equation}
which can be identified with the convective derivative of $g_t^{(\omega)}$ along the curve of $t$--dependent diffeomorphisms (\ref{fidiff}) associated with the flow 
$(t,\delta _z)\longmapsto \left(t,\,p^{(\omega)} _t(\cdot , z)d\omega(\cdot )\right)$,\;$t\in (0, \infty )$.
According to Lemma  \ref{Lagr} we have 
\begin{eqnarray}
\label{comoving20}
&&\left(\varphi^{-1}(t)\right)^*\frac{d}{d t}\,\left(\varphi^*(t)\, g_t^{(\omega)}\right)(U_\perp , V_\perp )\\
\nonumber\\
&&=\frac{d }{d t} \int_{M}\,g_{ik}(y)\,\nabla ^{i}_{(y)}\,\widehat{\psi}_{(t,U)}\,\nabla ^{k}_{(y)}\,\widehat{\psi}_{(t,V)}\,p^{(\omega)} _t(y , z)\,d\omega(y)\;,\nonumber\\
\nonumber\\
&&+2\,\int_{M}\,Hess_{ik}^{(y)}\,\ln p^{(\omega)} _t(y, z)\,\nabla ^{i}_{(y)}\,\widehat{\psi}_{(t,U)}\,\nabla ^{k}_{(y)}\,\widehat{\psi}_{(t,V)}\,p^{(\omega)} _t(y , z)\,d\omega(y)\nonumber\\
\nonumber\\
&&=\,-2\,\int_{M}\,\left[Hess\,\widehat{\psi}_{(t,U)}\cdot Hess\,\widehat{\psi}_{(t,V)}  \right.
\nonumber\\
\nonumber\\
&&\left.\,+\,\left(R_{ik}\,+\,Hess_{ik}\,\left(f-\ln\,p^{(\omega)} _t\right)\right)\,\nabla^i\widehat{\psi}_{(t,U)}\,\nabla^k\widehat{\psi}_{(t,V)} \right]\,p^{(\omega)} _t(y , z)\,d\omega(y)\nonumber\;,
\end{eqnarray}
where we have taken into account (\ref{genbetafunct}). On the other hand for any  $t>0$ and the identities 
\begin{eqnarray}
&&\bigtriangleup _{p _t\,(d\omega)}\left(\nabla \widehat{\psi}_{(t,U)}\,\cdot \,\nabla 
\widehat{\psi}_{(t,V)}  \right)\,=\,\nabla \widehat{\psi}_{(t,U)}\,\cdot \bigtriangleup _{p _t\,(d\omega)}
\nabla \widehat{\psi}_{(t,V)}\nonumber\\
\\
&&+\,\nabla \widehat{\psi}_{(t,V)}\,\cdot \bigtriangleup _{p _t\,(d\omega)}
\nabla \widehat{\psi}_{(t,U)}\,+\,2\,Hess\,\widehat{\psi}_{(t,U)}\cdot Hess\,\widehat{\psi}_{(t,V)}\;,\nonumber
\end{eqnarray} 
and 
\begin{eqnarray}
&&\nabla \widehat{\psi}_{(t,U)}\,\cdot \,\nabla \,\left(\bigtriangleup_{p_t(d\omega)} \,\widehat{\psi}_{(t,V)}\right)\,=\,\bigtriangleup _{p _t\,(d\omega)}\left(\nabla \widehat{\psi}_{(t,U)}\,\cdot \,\nabla 
\widehat{\psi}_{(t,V)}  \right)\nonumber\\
\\
&&\,-\,\nabla \widehat{\psi}_{(t,U)}\,\cdot \bigtriangleup _{p _t\,(d\omega)}
\nabla \widehat{\psi}_{(t,V)}-\,2\,Hess\,\widehat{\psi}_{(t,U)}\cdot Hess\,\widehat{\psi}_{(t,V)}\\
\nonumber\\
&&\,-\,\nabla^i\widehat{\psi}_{(t,U)}\,
\left(R_{ik}\,+\,Hess_{ik}\,\left(f-\ln\,p^{(\omega)} _t\right)\right)\,\nabla^k\widehat{\psi}_{(t,V)}\;, \nonumber
\end{eqnarray}
(the former a direct consequence of the definition (\ref{heatweight}) of the heat kernel weighted 
laplacian $\bigtriangleup_{p_t(d\omega)}$, the latter a  weighted  Bochner--Weitzenb\"ock Ricci tensor identity for $\bigtriangleup_{p_t(d\omega)}$, easily derived with a long but otherwise standard computation), we get the integral relation
\begin{eqnarray}
\label{sqlap0biss}
&&\int_M\,\widehat{\psi}_{(t,U)} \,\bigtriangleup^2 _{p _t\,(d\omega)}\,\widehat{\psi}_{(t,V)}\,
p^{(\omega)} _t(y , z)\,d\omega(y)\\
\nonumber\\
&&=\,-\,\int_{M}\,\nabla \widehat{\psi}_{(t,U)}\,\cdot \,\nabla \,\left(\bigtriangleup_{p_t(d\omega)} \,\widehat{\psi}_{(t,V)}\right)\,p^{(\omega)} _t(y , z)\,d\omega(y)\nonumber
\nonumber\\
\nonumber\\
&&=\,\int_{M}\,\left[Hess\,\widehat{\psi}_{(t,U)}\cdot Hess\,\widehat{\psi}_{(t,V)} \right.
\nonumber\\
\nonumber\\
&&\left.\,+\,\left(R_{ik}\,+\,Hess_{ik}\,\left(f-\ln\,p^{(\omega)} _t\right)\right)\,\nabla^i\widehat{\psi}_{(t,U)}\,\nabla^k\widehat{\psi}_{(t,V)} \right]\,p^{(\omega)} _t(y , z)\,d\omega(y)\nonumber\;.
\end{eqnarray}
By taking into account (\ref{comoving20}), this proves (\ref{LieMonot0}). If we set $U=V$ in the above relations  then we get the integrated weighted  Bochner--Weitzenb\"ock formula, (cf. (\ref{wWBheat})),
\begin{eqnarray}
\label{sqlap0}
&&\int_M\,\widehat{\psi}_{(t,U)} \,\bigtriangleup^2 _{p _t\,(d\omega)}\,\widehat{\psi}_{(t,U)}\,
p^{(\omega)} _t(y , z)\,d\omega(y)\\
\nonumber\\
&&\,=\,
\int_M\,\left(\bigtriangleup _{p _t\,(d\omega)}\,\widehat{\psi}_{(t,U)}  \right)^2\,
p^{(\omega)} _t(y , z)\,d\omega(y)\nonumber\\
\nonumber\\
&&=\,-\,\int_{M}\,\nabla \widehat{\psi}_{(t,U)}\,\cdot \,\nabla \,\left(\bigtriangleup_{p_t(d\omega)} \,\widehat{\psi}_{(t,U)}\right)\,p^{(\omega)} _t(y , z)\,d\omega(y)\nonumber
\nonumber\\
\nonumber\\
&&=\,\int_{M}\,\left[|Hess\,\widehat{\psi}_{(t,U)} |_g^2\right.
\nonumber\\
\nonumber\\
&&\left.\,+\,\left(R_{ik}\,+\,Hess_{ik}\,\left(f-\ln\,p^{(\omega)} _t\right)\right)\,\nabla^i\widehat{\psi}_{(t,U)}\,\nabla^k\widehat{\psi}_{(t,U)} \right]\,p^{(\omega)} _t(y , z)\,d\omega(y)\nonumber\;,
\end{eqnarray}
which together with (\ref{comoving20}), (for  $U=V$), proves (\ref{LieMonot}). To show that the monotonicity result in (\ref{LieMonot}) is strict, let us observe that the vanishing of $\int_M\,(\bigtriangleup _{p _t\,(d\omega)}\,\widehat{\psi}_{(t,U)})^2\,
p^{(\omega)} _t(y , z)\,d\omega(y)$ would necessarily imply $\bigtriangleup _{p _t\,(d\omega)}\,\widehat{\psi}_{(t,U)}=0$, hence according to proposition \ref{CNprop} and (\ref{carlo12Hnew}) the corresponding vanishing of $U_\perp(z)$, contradicting the stated hypotheses. Finally,  in order to prove (\ref{monotpsi}) let us set for notational convenience
\begin{equation}
\label{Anot}
A(z,t)\,:=\,
\int_M\,\left(\bigtriangleup _{p _t\,(d\omega)}\,\widehat{\psi}_{(t,z,U_\perp )}(y)  \right)^2\,
p^{(\omega)} _t(y , z)\,d\omega(y)\;.
\end{equation}
Since the heat kernel $p^{(\omega)} _t(y , z)$ and the associated weighted Laplacian $\bigtriangleup _{p _t\,(d\omega)}^{(y)}$ are smooth for $t>0$, we compute
\begin{eqnarray}
&&\frac{d }{d t}\,A(z,t)\\
&&=\,2
\int_M\,\left(\bigtriangleup _{p _t\,(d\omega)}\,\widehat{\psi}_{(t,U)}  \right)\frac{\partial }{\partial t}\left(\bigtriangleup _{p _t\,(d\omega)}\,\widehat{\psi}_{(t,U)} \right)\,
p^{(\omega)} _t(y , z)\,d\omega(y)\nonumber\\
\nonumber\\
&&+\,\int_M\,\left(\bigtriangleup _{p _t\,(d\omega)}\,\widehat{\psi}_{(t,U)}\right)^2\frac{\partial }{\partial t}
p^{(\omega)} _t(y , z)\,d\omega(y)\nonumber\\
\nonumber\\
&& \underset{(\ref{zlimit3})}{=}\,2\,
\int_M\,\left(\bigtriangleup _{p _t\,(d\omega)}\,\widehat{\psi}_{(t,U)}  \right)\left[\bigtriangleup _{p _t\,(d\omega)}\left(\bigtriangleup _{p _t\,(d\omega)}\,\widehat{\psi}_{(t,U)}  \right) \right.\nonumber\\
\nonumber\\
&&\left.+\,\nabla \ln{p _t\,(d\omega)}\cdot \nabla \left(\bigtriangleup _{p _t\,(d\omega)}\,\widehat{\psi}_{(t,U)}  \right)\right]\,
p^{(\omega)} _t(y , z)\,d\omega(y)\nonumber\\
\nonumber\\
&& \underset{(\ref{heatfloweight})}{+}\,\int_M\,\left(\bigtriangleup _{p _t\,(d\omega)}\,\widehat{\psi}_{(t,U)}\right)^2
\bigtriangleup _\omega
p^{(\omega)} _t(y , z)\,d\omega(y)\nonumber\\ 
\nonumber\\
 &&=\, 2\,
\int_M\,\left(\bigtriangleup _{p _t\,(d\omega)}\,\widehat{\psi}_{(t,U)}\right) 
\bigtriangleup _{p _t\,(d\omega)}\left(\bigtriangleup _{p _t\,(d\omega)}\,\widehat{\psi}_{(t,U)}  \right)
p^{(\omega)} _t(y , z)\,d\omega(y)\nonumber\\
\nonumber\\
&&+\,\int_M\,
\nabla{p^{(\omega)} _t(y , z)}\cdot \nabla \left(\bigtriangleup _{p _t\,(d\omega)}\,\widehat{\psi}_{(t,U)}\right)^2\,
d\omega(y)\nonumber
\nonumber\\
&& +\,\int_M\,\left(\bigtriangleup _{p _t\,(d\omega)}\,\widehat{\psi}_{(t,U)}\right)^2
\bigtriangleup _\omega
p^{(\omega)} _t(y , z)\,d\omega(y)\nonumber\\
\nonumber\\
&&=\,-2\,\int_M \left|\nabla^{(y)}\left(\bigtriangleup _{p _t\,(d\omega)}^{(y)}\,\widehat{\psi}_{(t,z,U_\perp )}(y)  \right)\right|^2\,p^{(\omega)} _t(y , z)\, d\omega(y)\;.\nonumber
\end{eqnarray}
where, in the last passage, we have integrated by parts both with respect to the $p^{(\omega)} _t(y , z)\, d\omega(y)$ and the $d\omega(y)$ measures.
\end{proof}
Note that  (\ref{LieMonot}) can be rewritten as
\begin{eqnarray}
\label{}
&&\frac{d }{d t} g_t^{(\omega)}(U_\perp , U_\perp )(z)=\,-\,{\mathcal{L}}_{X_{(t,z)}} g_t^{(\omega)}(U_\perp , U_\perp )\\  
\nonumber\\
&&\,-2\,\int_M\,\left(\bigtriangleup _{p _t\,(d\omega)}\,\widehat{\psi}_{(t,U)}  \right)^2\,
p^{(\omega)} _t(y , z)\,d\omega(y)\;,\nonumber
\end{eqnarray}
which shows that the evolution of $g_t^{(\omega)}(U_\perp , U_\perp )(z)$ results from a balance between 
the $-2\int_M(\bigtriangleup _{p _t\,(d\omega)}\,\widehat{\psi}_{(t,U)})^2p^{(\omega)} _t d\omega$ term,\, which tends to contract the $g_t^{(\omega)}$--norm of vectors, and the  term
\begin{eqnarray}
&&\overline{\mathcal{L}}_{\vee _{(t,z)}} g_t^{(\omega)}(U_\perp , U_\perp )\\
\nonumber\\
&&\,=\,2\,\int_{M}\,Hess\,\ln\,p^{(\omega)} _t(y , z)\,\left(\nabla\,
\widehat{\psi}_{(t,U)},\,\nabla\,\widehat{\psi}_{(t,U)}\right)\,p^{(\omega)} _t(y , z)\,d\omega(y)\;,\nonumber
\end{eqnarray}
 which, as already stressed, computes  $Hess\;{d_g^2(y, z)}$ in the sense of distributions and, by the local convexity of $d_g^2(y, z)$, tends to contrast this contraction.\\
 \\
 It is not difficult to show that the pulled back flow (\ref{LieMonot})
is a gradient flow. Let $\{E_{(a)}(z)\}_{a=1}^n$ denote a basis in $T_zM$, orthonormal with respect to the 
given $(M, g)$, and for any $t>0$ consider the functional
 \\
\begin{eqnarray}
\label{hodefF}
&&\mathcal{F}\left((M,\,g_t^{(\omega)})\right)\\
\nonumber\\
\nonumber\\
&&:=\,\int_{\varphi_{(z,t)}(M)}\,d\omega(z)\,g^{ab}(z)
\int_M\,\widehat{\psi}_{(t,a)}\,\bigtriangleup^2 _{p _t\,(d\omega)}\,\widehat{\psi}_{(t,b)}\,
p^{(\omega)} _t(y , z)\,d\omega(y)\;,\nonumber
\end{eqnarray}
where $\varphi_{(z,t)}$ denotes the curve of diffeomorphisms defined by (\ref{fidiff}).
\begin{rem}
In $\mathcal{F}\left((M,\,g_t^{(\omega)})\right)$ one may consider more natural tracing with respect to $g_t^{(\omega)}(z)$ rather than with respect to $g(z)$. In such a case, the corresponding functional can be identified (modulo the action of the curve of diffeomorphisms $t\mapsto \varphi(t)$) with the time--derivative of the Riemannian volume of $(M,\,g_t^{(\omega)})$. As a consequence, this apparently more general functional,  has a non--trivial $L^2(d\omega)$ gradient only along the scalar variation of $g_t^{(\omega)}$, and does not capture all possible tensorial variations (and deformation) of $g_t^{(\omega)}$, variations which are fully described by the $L^2(d\omega)$--gradient of $\mathcal{F}\left((M,\,g_t^{(\omega)})\right)$. Indeed, we have 
\end{rem}
\begin{lem}
Let $z\longmapsto \tfrac{1}{2}\chi_{ab}(z)\in \otimes ^2_{sym}T_z^*M$ be a smooth symmetric bilinear form on $(M,\,g_t^{(\omega)})$ thought of as acting fiberwise as a tangent bundle endomorphism. For $0\,<\,\epsilon\,<\,1$ sufficiently small, let us consider the variation of  basis vectors $\{E_{(a)}(z)\} \in T_zM$ defined by 
$E_{(a)} \longmapsto E_{(b)} ^{(\epsilon )}:=\,(\delta^a_b+\tfrac{\epsilon}{2}\,\chi^a_b (z) )E_{(a)}$. If we let $\widehat{\psi}^{(\epsilon)}_{(t,b)}$ denote the induced variation in the 
potentials $\widehat{\psi}_{(t,b)}:=\widehat{\psi}_{(t, E_{(b)})}$, then for any $t>0$
\begin{equation}
\label{epsi}
\widehat{\psi}^{(\epsilon)}_{(t,b)}(y)\,=\,(\delta^a_b+\tfrac{\epsilon}{2}\,\chi^a_b (z) )\,
\widehat{\psi}_{(t,a)}(y)\;, 
\end{equation}
and  the corresponding linearization of 
the functional $\mathcal{F}\left((M,\,g_t^{(\omega)})\right)$ in the direction of the variation defined by the bilinear form $\chi$ is provided by 
\\
\begin{eqnarray}
\label{Df}
&&D\,\mathcal{F}\left((M,\,g_t^{(\omega)})\right)\circ \chi\nonumber\\
\\
&&\,=\,\int_{\varphi_{(z,t)}(M)}\,\left(
\int_M\,\widehat{\psi}_{(t,a)}\,\bigtriangleup^2 _{p _t\,(d\omega)}\,\widehat{\psi}_{(t,b)}\,
p^{(\omega)} _t(y , z)\,d\omega(y)\right)\,\chi^{ab}(z)\,d\omega(z)\;.\nonumber
\end{eqnarray} 
\end{lem}

\begin{proof}
By the linearity of the defining elliptic PDE (\ref{carlo12H}) it immediately follows that
the potential $\widehat{\psi}^{(\epsilon)}_{(t,b)}$ associated to the rescaled basis vector
$E_{(b)} ^{(\epsilon )}:=\,(\delta^a_b+\tfrac{\epsilon}{2}\,\chi^a_b (z) )E_{(a)}$ is  provided by (\ref{epsi}). 
Similarly, from the very definition of the metric $g_t^{(\omega)}(z)$, (cf.(\ref{gt}) ), we directly have 

\begin{eqnarray}
\label{}
&&g_t^{(\omega, \epsilon)}\left(E_{(c)}, E_{(d)} \right)\,:=\,g_t^{(\omega)}
\left(E_{(c)}^{(\epsilon)}, E_{(d)}^{(\epsilon)} \right)\\
\nonumber\\
&&\,:=\,\int_{M}\,g_{ik}(y)\,\nabla ^{i}_{(y)}\,\widehat{\psi}^{(\epsilon)}_{(t, c)}\,\nabla ^{k}_{(y)}\,
\widehat{\psi}^{(\epsilon)}_{(t, d)}\,p^{(\omega)} _t(y , z)\,d\omega(y)\nonumber\\
\nonumber\\
&&\,=\,\left(\delta^a_c+\tfrac{\epsilon}{2}\,\chi^a_c (z) \right)
\left(\delta^b_d+\tfrac{\epsilon}{2}\,\chi^b_d (z) \right)\nonumber\\
\nonumber\\
&&\,\times \,\int_{M}\,g_{ik}(y)\,\nabla ^{i}_{(y)}\,\widehat{\psi}_{(t, a)}\,\nabla ^{k}_{(y)}\,
\widehat{\psi}_{(t, b)}\,p^{(\omega)} _t(y , z)\,d\omega(y)\;.\nonumber
\end{eqnarray}
Hence, to leading order in $\epsilon$, we get
\begin{equation}
\label{metrvart}
g_t^{(\omega, \epsilon)}\left(E_{(c)}, E_{(d)} \right)\,=\,
g_t^{(\omega)}\left(E_{(c)}, E_{(d)} \right)\,\,+\,\epsilon \,\chi_{cd}(z)\,+\,O(\epsilon^2)\;,\nonumber
\end{equation}
as expected. In particular, by letting $\chi$ vary in the space of all symmetric bilinear form 
$C^{\infty}(M, \otimes ^2_{sym}T^*M)$ we can interpret $\chi$ as describing the generic metric variation of  $g_t^{(\omega)}$. From

\begin{eqnarray}
\label{easycomp}
&&g^{cd}(z)\,\int_M\,\widehat{\psi}^{(\epsilon)}_{(t,c)}\,\bigtriangleup^2 _{p _t\,(d\omega)}\,
\widehat{\psi}^{(\epsilon)}_{(t,d)}\,
p^{(\omega)} _t(y , z)\,d\omega(y)\,\\
\nonumber\\
&&=\,g^{cd}(z)\,\left(\delta^a_c+\tfrac{\epsilon}{2}\,\chi^a_c (z) \right)
\left(\delta^b_d+\tfrac{\epsilon}{2}\,\chi^b_d (z) \right)\nonumber\\
\nonumber\\
&&\times \,
\int_M\,\widehat{\psi}_{(t,a)}\,\bigtriangleup^2 _{p _t\,(d\omega)}\,
\widehat{\psi}_{(t,b)}\,
p^{(\omega)} _t(y , z)\,d\omega(y)\nonumber\\
\nonumber\\
&&=\,\left(g_{ab}(z)+\epsilon\,\chi_{ab}(z)+\frac{\epsilon^2}{4}\chi_{ad}(z)\chi_b^d(z) \right)\nonumber\\
\nonumber\\
&&\times \,
\int_M\,g^{ac}(y)g^{bd}(y)\,\widehat{\psi}_{(t,c)}\,\bigtriangleup^2 _{p _t\,(d\omega)}\,
\widehat{\psi}_{(t,d)}\,
p^{(\omega)} _t(y , z)\,d\omega(y)\nonumber\;,
\end{eqnarray}
we easily compute
\begin{eqnarray}
&&D\,\mathcal{F}\left((M,\,g_t^{(\omega)})\right)\circ \chi\,:=\,
\left.\frac{d}{d\epsilon}\right|_{\epsilon=0}\,\mathcal{F}\left((M,\,g_t^{(\omega,\epsilon)})\right)\\
\nonumber\\
&&=\,\int_{\varphi_{(z,t)}(M)}\,\left(
\int_M\,g^{ac}g^{bd}\,\widehat{\psi}_{(t,c)}\,\bigtriangleup^2 _{p _t\,(d\omega)}\,\widehat{\psi}_{(t,d)}\,\right.\nonumber\\
\nonumber\\
\nonumber\\
&&\left.\times\,p^{(\omega)} _t(y , z)\,d\omega(y)\right)\,\chi_{ab}(z)\,d\omega(z)\nonumber\;.
\end{eqnarray}
\end{proof}
According to (\ref{metrvart}) the bilinear form $\chi\in C^{\infty}(M, \otimes ^2_{sym}T^*M)$ induces the generic metric variation of  $g_t^{(\omega)}$, hence we have
\begin{thm}
\label{thgradflow}
The heat induced Ricci flow (\ref{LieMonot0})   is the gradient flow of the functional $\mathcal{F}\left((M,\,g_t^{(\omega)})\right)$ with respect to the $L^2(d\omega)$ inner product on the space of metrics $\mathcal{M}et(M)$.  
\end{thm}
Along the same lines is easy to prove that the flow (\ref{GenRicPer}) is the gradient flow of the functional obtained from 
$\mathcal{F}\left((M,\,g_t^{(\omega)})\right)$ by undoing the action of the diffeomorphism $\varphi_{(t,z)}$, 
\emph{i.e.}
\\
\begin{eqnarray}
\label{perlike}
&&\widehat{\mathcal{F}}\left((M,\,g_t^{(\omega)})\right)\\
\nonumber\\
&&\,:=\,\,\int_{M}\,d\omega(z)\,g^{ab}(z)\,
\int_M\,\widehat{\psi}_{(t,a)}\,\bigtriangleup^2 _{p _t\,(d\omega)}\,\widehat{\psi}_{(t,b)}\,
p^{(\omega)} _t(y , z)\,d\omega(y)\nonumber\\
\nonumber\\
&&\,
+\,\frac{1}{2}\,\int_M\,g^{ab}(z)\,\mathcal{L}_{X_{(t,z)}}\,g_t^{(\omega)}(E_{(a)}, E_{(b)})\,d\omega(z)\;,\nonumber
\end{eqnarray}
(cf. (\ref{hessianlieder}) for the definition of the Lie derivative along $X_{(t,z)}$).
In particular we have
\begin{lem}
\label{Fdecreasing}
\begin{equation}
\frac{d}{dt}\,\widehat{\mathcal{F}}\left((M,\,g_t^{(\omega)})\right)\,=\,
D\,\widehat{\mathcal{F}}\left((M,\,g_t^{(\omega)})\right)\circ \frac{d g_t^{(\omega)}}{dt}\,\leq\,0\;.
\end{equation}
\end{lem}
\begin{proof}
By the chain rule we have, for $0\,<\,t\,<\,\epsilon $,
\begin{eqnarray}
&&\frac{d}{dt}\,\widehat{\mathcal{F}}\left((M,\,g_t^{(\omega)})\right)= 
\left.\frac{d}{d\epsilon}\right|_{\epsilon=0}\,\widehat{\mathcal{F}}\left((M,\,g_{t+\epsilon}^{(\omega)})\right)\\
\nonumber\\
&&=\,D\,\widehat{\mathcal{F}}\left((M,\,g_t^{(\omega)})\right)\circ \frac{d g_t^{(\omega)}}{dt}\nonumber\;.
\end{eqnarray}
From (\ref{sqlap0}) and (\ref{hessianlieder})  we get
\\
\begin{eqnarray}
&&\widehat{\mathcal{F}}\left((M,\,g_t^{(\omega)})\right)\\
\noindent
&&=\, \int_M\,d\omega(z)g^{ab}(z)\left( \int_{M}\,\left[Hess\,\widehat{\psi}_{(t,a)}\cdot Hess\,\widehat{\psi}_{(t,ab)}\right.\right.\nonumber\\
\nonumber\\
\nonumber\\
&&\left.\left.\,+\,\nabla^i\widehat{\psi}_{(t,a)}\,\left(R_{ik}\,+\,Hess_{ik}\,f\right)\,\nabla^k\widehat{\psi}_{(t,b)} \right]\,p^{(\omega)} _t(y , z)\,d\omega(y)\right)\nonumber\;,
\end{eqnarray}
whose linearization in the direction $\chi:=\tfrac{d}{dt}\,g_t^{(\omega)}$ can be easily computed, (as in (\ref{easycomp})), to be
\begin{eqnarray}
&&D\,\widehat{\mathcal{F}}\left((M,\,g_t^{(\omega)})\right)\circ \frac{d g_t^{(\omega)}}{dt}\\
\noindent
&&=\, \int_M\,d\omega(z)\frac{d}{dt}(g_t^{(\omega)})_{ab}(z)\left( \int_{M}\,g^{ac}g^{bd}\left[Hess\,\widehat{\psi}_{(t,c)}\cdot Hess\,\widehat{\psi}_{(t,d)}\right.\right.\nonumber\\
\nonumber\\
\nonumber\\
&&\left.\left.\,+\,\nabla^i\widehat{\psi}_{(t,c)}\,\left(R_{ik}\,+\,Hess_{ik}\,f\right)\,\nabla^k\widehat{\psi}_{(t,d)} \right]\,p^{(\omega)} _t(y , z)\,d\omega(y)\right)\nonumber\;.
\end{eqnarray}
The lemma follows by inserting the expression (\ref{GenRicPer}) for the  heat hernel induced Ricci flow.  
\end{proof}
The geometrical meaning of  $\widehat{\mathcal{F}}\left((M,\,g_t^{(\omega)})\right)$ is quite natural since we get 
\begin{prop}
\label{redF}
For $t>0$ the functional $\widehat{\mathcal{F}}\left((M,\,g_t^{(\omega)})\right)$ is a deformation of the Perelman $\mathcal{F}$--energy functional and  
\begin{eqnarray}
\lim_{t\searrow 0}\,\widehat{\mathcal{F}}\left((M,\,g_t^{(\omega)})\right)\,=\,\int_{M}\,\left(R(g)\,+\left|\nabla f\right|_g^2\right)\,e^{-f(z)}\,d\mu_g(z)\;, 
\end{eqnarray}
in the weak sense.
\end{prop}
\begin{proof}
Let us assume that the measure $d\omega$ is localized in a ball $B(r, z):= \{q\in M|\,d_g(q,z)\leq \tfrac{1}{3}\,inj\,(M,g)\}$, (otherwise introduce in the integrals defining $\widehat{\mathcal{F}}$  a smooth radial cutoff function $\zeta (d_g(q,z)):=1$ for $d_g(q,z)\leq \tfrac{1}{4}\,inj\,(M,g)$, and $\zeta (d_g(q,z)):=0$ for  $d_g(q,z)\geq \tfrac{1}{2}\,inj\,(M,g)$). We let  $\{E_{(a)}(z)\}_{a=1}^n$ denote a basis in $T_zM$, orthonormal with respect to $(M, g)$, and let 
$[0,1]\ni s\mapsto \gamma_{(a)}$, $\gamma_{(a)}(0)\equiv z$,\, $\gamma'_{(a)}(0)\equiv E_{(a)}(z)$  denote the corresponding  geodesics in $(M,g)$. Again, from (\ref{sqlap0}) and (\ref{hessianlieder})  we compute
\\
\begin{eqnarray}
&&\sum_{a=1}^n \int_M\left(\bigtriangleup _{p _t\,(d\omega)}\widehat{\psi}_{(t,z,E_{(a)})}  \right)^2
p^{(\omega)} _t(y , z)\,d\omega(y)\\
\nonumber\\
&&\,+\,\frac{1}{2}\sum_{a=1}^n\,\mathcal{L}_{X_{(t,z)}}\,g_t^{(\omega)}(E_{(a)}, E_{(a)})\nonumber\\
\nonumber\\
&&=\, \sum_{a=1}^n\left( \int_{M}\,\left[|Hess\,\widehat{\psi}_{(t,a)} |_g^2\right.\right.\nonumber\\
\nonumber\\
\nonumber\\
&&\left.\left.\,+\,\left(R_{ik}\,+\,Hess_{ik}\,f\right)\,\nabla^i\widehat{\psi}_{(t,a)}\,\nabla^k\widehat{\psi}_{(t,a)} \right]\,p^{(\omega)} _t(y , z)\,d\omega(y)\right)\nonumber\;.
\end{eqnarray}
According to Proposition \ref{tangbetaflow}, we 
have $|Hess\,\widehat{\psi}_{(t,a)} |_g^2\rightarrow 0$ almost everywhere as $t\searrow 0$.  For $0<t$ small enough $p^{(\omega)} _t$ is localized around $z$ and  the trace over the $\{\nabla \widehat{\psi}_{(t,a)}\}$ reduces to the trace over the orthonormal vectors   $\{\gamma'_{(a)}(0)\equiv E_{(a)}(z)\}$. Abusing notation, (\emph{i.e.}, tracing over $\{\gamma'_{(a)}(0)\}$ before taking the $\searrow 0$ limit), we can write 
\\
\begin{eqnarray}
\label{defsumasy}
&&\sum_{a=1}^n\left(\int_M\left(\bigtriangleup _{p _t\,(d\omega)}\widehat{\psi}_{(t,z,E_{(a)})}  \right)^2
p^{(\omega)} _t(y , z)\,d\omega(y)\right)\\
\nonumber\\
&&\,+\,\frac{1}{2}\sum_{a=1}^n\,\mathcal{L}_{X_{(t,z)}}\,g_t^{(\omega)}(E_{(a)}, E_{(a)})\nonumber\\
\nonumber\\
&&=\,\int_{M}\,\left(R_g(y)\,+\,\triangle _g f\right)\,
p^{(\omega)} _t(y , z)\,d\omega(y)\,+\, O(t^{3/2})\nonumber\\
\nonumber\\
&&=\,\int_{M}\,\left(R_g(y)\,+\,\triangle _\omega\,f\,+\,|\nabla f|^2_g\right)\,
p^{(\omega)} _t(y , z)\,d\omega(y)\,+\, O(t^{3/2})\nonumber\;,
\end{eqnarray}
where $R_g$ and $\triangle _g$ respectively  denote the scalar curvature and the Laplace--Beltrami laplacian on $(M,g)$, where we exploited the relation $\triangle _g\,f=\triangle _\omega\,f\,+\,|\nabla f|^2_g$. Hence, we can write
\begin{eqnarray}
&&\widehat{\mathcal{F}}\left((M,\,g_t^{(\omega)})\right)\\
\nonumber\\
&&=\int_{M}
\left( \int_{M}\left(R_g(y)+\triangle _\omega\,f+|\nabla f|^2_g\right)
p^{(\omega)} _t(y , z)d\omega(y)\right)d\omega(z)+O(t^{3/2})\nonumber\;.
\end{eqnarray}
By the symmetry of the heat kernel and letting $t\searrow 0$, we compute
\begin{eqnarray}
&&\lim_{t\searrow 0}\,\widehat{\mathcal{F}}\left((M,\,g_t^{(\omega)})\right)\\
\nonumber\\
&&=\int_{M}
\left(R_g(z)+\triangle _\omega\,f(z)+|\nabla f|^2_g(z)\right)\,d\omega(z)\nonumber\\
\nonumber\\
&&=\int_{M}
\left(R_g(z)+|\nabla f|^2_g(z)\right)\,e^{-\,f(z)}d\mu_g(z)\nonumber\;,
\end{eqnarray}
which is the expression for the Perelman $\mathcal{F}$--energy, (at $t=0$), associated with the function $f$ on $(M,g)$. It follows that $\widehat{\mathcal{F}}\left((M,\,g_t^{(\omega)})\right)$ can be seen as a deformation of Perelman's\, $\mathcal{F}$--energy.
\end{proof}
This is a suitable point at which we should come back to our sponsor, the interplay between Ricci flow and NL$\sigma$M. We do so by observing that it would be appealing to associate $\widehat{\mathcal{F}}\left((M,\,g_t^{(\omega)})\right)$, or a variant thereof, to Zamolodchikov's  $c$--theorem \cite{Zamo}. Recall that this result, (established for two dimensional quantum field theories, but rather conjectural for higher dimensional theories\footnote{There have been recent indications \cite{komar} to a proof in dimension $4$.}), concerns the existence of a functional  of the coupling constants of the theory, (\emph{i.e.}, $(M,g,\,d\omega)$ in the dilatonic NL$\sigma$M case), which is non--increasing along the RG flow and stationary at the fixed points of the flow, where it takes the value of the central charge $c$ of the conformal field theory described by the fixed point of the RG action. The relation between Ricci flow and the renormalization group for the non--linear $\sigma$ model has suggested \cite{Guenther},  \cite{Oliynyk}, \cite{Tseytlin}  that
Perelman's\, $\mathcal{F}$--energy may be a natural candidate for such a functional. As natural as it appears, this identification is rather delicate since it involves a strong extrapolation of the underlying perturbative regime governing the relation between Ricci and RG flow. If we factor out the unphysical Perelman type constraint, fixing the dilatonic measure $d\omega$ along the Ricci--Perelman flow, the functional $\mathcal{F}$ appears in  NL$\sigma$M theory  as a \emph{spacetime} action generating the (conformally invariant) fixed points of the  $1$--loop RG flow. It is clear that the perturbative nature of this characterization makes difficult if not impossible to prove that $\mathcal{F}$ plays indeed the role of a full $c$--functional, (see however \cite{Tseytlin}). In this connection, the RG avatar defined by the heat kernel embedding has an obvious advantage since, as we have shown in Propositions \ref{defharmenerg}  and \ref{devbehavior} , we can associate to it the non--perturbative effective action $\mathcal{E}[\widehat{\Psi}_{t,\phi_{cm}}]$ defined by the deformed harmonic map functional. This suggests the following characterization. Let  
$\{\phi_{(j)}\}_{j=1}^q \Sigma\longrightarrow M$ denote the collection of maps fluctuating according to the Gaussian measure $\mathcal{Q}_t[d\phi_{(j)}]$  around the \emph{classical} background provided by their center of mass $\phi_{cm}$, (cf. Section \ref{renGroup}). We can associate with this background the natural modification of  (\ref{perlike}) obtained by localizing (the heat source $z$ dependence in) 
$\widehat{\mathcal{F}}_{\phi_{cm}(\Sigma)}\left((M,\,g_t^{(\omega)})\right)$ to $\phi_{cm}(\Sigma)\subseteq M$, \emph{i.e.}
\\
\begin{eqnarray}
\label{}
&&\widehat{\mathcal{F}}_{\phi_{cm}(\Sigma)}\left((M,\,g_t^{(\omega)})\right)\\
\nonumber\\
&&\,:=\,\,\int_{\phi_{cm}(\Sigma)}\,d\omega(z)\,g^{ab}(z)\,
\int_M\,\widehat{\psi}_{(t,a)}\,\bigtriangleup^2 _{p _t\,(d\omega)}\,\widehat{\psi}_{(t,b)}\,
p^{(\omega)} _t(y , z)\,d\omega(y)\nonumber\\
\nonumber\\
&&\,
+\,\frac{1}{2}\,\int_{\phi_{cm}(\Sigma)}\,g^{ab}(z)\,\mathcal{L}_{X_{(t,z)}}\,g_t^{(\omega)}(E_{(a)}, E_{(b)})\,d\omega(z)\;.\nonumber
\end{eqnarray}
Since  $\mathcal{E}[\widehat{\Psi}_{t,\phi_{cm}}]$ is the large deviation functional associated with the distribution  $\Pi _{j=1}^q\,\mathcal{Q}_t[d\phi_{(j)}]$, the functional $\widehat{\mathcal{F}}_{\phi_{cm}(\Sigma)}$ may be heuristically interpreted as describing the average deformation in $\mathcal{E}[\widehat{\Psi}_{t,\phi_{cm}}]$ induced, along the flow $(M,g)\longmapsto (M,\,g_t^{(\omega)})$, by the fluctuating $\{\phi_{(j)}\}_{j=1}^q$. According to Lemma \ref{Fdecreasing} and  Proposition \ref{redF}, the functional $\widehat{\mathcal{F}}_{\phi_{cm}(\Sigma)}$ is monotonically decreasing along $(t,g)\longmapsto g_t^{(\omega)}$,\,$t\in [0,\infty )$, and reduces to Perelman's 
$\mathcal{F}$--functional, (localized to  $\phi_{cm}(\Sigma)\subseteq M$), as $t\searrow 0$. Its role as a candidate $c$--functional for the heat kernel induced RG action on NL$\sigma$M  rests on the characterization of the nature of the fixed points for the generalized Ricci flow evolution $(t,g)\longmapsto g_t^{(\omega)}$,\,$t\in [0,\infty )$. This is clearly an open and very difficult problem, and perhaps an appropriate point at which to end this long analysis.\\
\\
\noindent In Ricci flow theory, the true added value of Wasserstein geometry and of its connection to optimal transport theory lies in the observation that the properties of $d_{\,g}^{\,W}$ are deeply related to the Ricci curvature  of the underlying metric measure space $(M,g,\,d\omega)$, a fact  that has been independently noticed and exploited by various authors  \cite{fokker, lott1, lott2, topping0}. The deep connections among NL$\sigma$M theory, Ricci flow, and Wasserstein geometry explored here add a further perspective to this rich interplay.

\subsection*{Acknowledgment}
I am particularly grateful to Nicola Gigli, Carlo Mantegazza, and Giuseppe Savar\'e for useful conversations and stimulating ideas in the preliminary stage of preparation of this paper.\\
\\
\noindent
This work has been partially supported by the PRIN Grant  2010JJ4KPA 006  \emph{Geometrical and analytical theories of finite and infinite  dimensional Hamiltonian systems}

\end{document}